\pdfoutput=1

\documentclass[10pt,journal,final,compsoc]{IEEEtran}

\ifCLASSOPTIONcompsoc
  \usepackage[nocompress]{cite}
\else
  \usepackage{cite}
\fi

\ifCLASSINFOpdf
\else
\fi
\usepackage{balance}  
\usepackage{amsfonts}
\usepackage{url}
\usepackage{algorithm}
\usepackage{algorithmic}
\usepackage{float}
\usepackage{amssymb,amsmath,amsthm,epsfig,graphics,enumerate,float,subfigure}
\usepackage{graphicx}
\usepackage{color}
\usepackage{xspace}
\usepackage{multirow}

\newcommand{\ignore}[1]{}
\newcommand{\nop}[1]{}
\newcommand{\eat}[1]{}
\newtheorem{definition}{Definition}
\newtheorem{lemma}{Lemma}
\newtheorem{theorem}{Theorem}
\newtheorem{example}{Example}
\newtheorem{property}{Property}

\newtheorem{remark}{Remark}
\newtheorem{observation}{Observation}

\newcommand{\squishlisttight}{
 \begin{list}{$\bullet$}
  { \setlength{\itemsep}{2pt}
    \setlength{\parsep}{0pt}
    \setlength{\topsep}{2pt}
    \setlength{\partopsep}{0pt}
    \setlength{\leftmargin}{2em}
    \setlength{\labelwidth}{1.5em}
    \setlength{\labelsep}{0.5em}
} }

\newcommand{\squishnumlist} {
\newcounter{qcounter}
\begin{list}{\arabic{qcounter}.~}{\usecounter{qcounter}} 
{  \setlength{\itemsep}{0pt}
    \setlength{\parsep}{0pt}
    \setlength{\topsep}{0pt}
    \setlength{\partopsep}{0pt}
    \setlength{\leftmargin}{2em}
    \setlength{\labelwidth}{1.5em}
    \setlength{\labelsep}{0.5em}
}}

\newcommand{\squishend}{
  \end{list}
}

\mathchardef\mhyphen="2D
\newcommand{\kw}[1]{{\ensuremath {\mathsf{#1}}}\xspace}
\newcommand{\kwnospace}[1]{{\ensuremath {\mathsf{#1}}}}

\newcommand{\jbhuang}[1]{{\color{black}{#1}}}
\newcommand{\jinbin}[1]{{\color{black}{#1}}}
\newcommand{\sigmodreview}[1]{{\color{green}{#1}}}

\newcommand{\ego}{\kwnospace{ego}-\kw{network}}
\newcommand{\egos}{\kwnospace{ego}-\kw{networks}}

\newcommand{\score}{\kw{score}}
\newcommand{\context}{\kw{SC}}

\newcommand{\TSD}{\kw{TSD}}
\newcommand{\baseline}{\kw{baseline}}
\newcommand{\bound}{\kw{bound}}

\newcommand{\Naive}{\kw{Hybrid}}

\newcommand{\ADV}{\kw{GCT}}
\newcommand{\ADVindex}{\kwnospace{GCT}-\kw{index}}
\newcommand{\GCT}{\kw{GCT}}

\newcommand{\TSDindex}{\kwnospace{TSD}-\kw{index}}

\newcommand{\random}{\kw{Random}}

\newcommand{\bitmap}{\kw{Bits}}
\newcommand{\component}{\kwnospace{Comp}-\kw{Div}}
\newcommand{\core}{\kwnospace{Core}-\kw{Div}}
\newcommand{\trussdiv}{\kwnospace{Truss}-\kw{Div}}

\newcommand{\stitle}[1]{\vspace{1ex} \noindent{\bf #1}}

\begin{document}
\title{Truss-based Structural Diversity Search in Large Graphs}


\author{Jinbin Huang,
        Xin Huang,
        and Jianliang Xu
\IEEEcompsocitemizethanks{
\IEEEcompsocthanksitem J. Huang, X. Huang and J. Xu are with Hong Kong Baptist University.\protect\\
E-mail: \{jbhuang, xinhuang, xujl\}@comp.hkbu.edu.hk}
}




\IEEEtitleabstractindextext{%

\begin{abstract}
Social decisions made by individuals are easily influenced by information from their social neighborhoods. A key predictor of social contagion is the multiplicity of social contexts inside the individual’s contact neighborhood, which is termed structural diversity. However, the existing models have limited decomposability for analyzing large-scale networks, and suffer from the inaccurate reflection of social context diversity. In this paper, we propose a new truss-based structural diversity model to overcome the weak decomposability. Based on this model, we study a novel problem of truss-based structural diversity search in a graph G, that is, to find the r vertices with the highest truss-based structural diversity and return their social contexts. o tackle this problem, we propose an online structural diversity search algorithm in $O(\rho(m+\mathcal{T}))$ time, where $\rho$, $m$, and $\mathcal{T}$ are respectively the arboricity, the number of edges, and the number of triangles in $G$.  To improve the efficiency, we design an elegant and compact index, called TSD-index, for further expediting the search process. We further optimize the structure of TSD-index into a highly compressed GCT-index. 
Our GCT-index-based structural diversity search utilizes the global triangle information for fast index construction and finds answers in $O(m)$ time. 
Extensive experiments demonstrate the effectiveness and efficiency of our proposed model and algorithms, against state-of-the-art methods.
\end{abstract}
}

\maketitle


\IEEEdisplaynontitleabstractindextext

%
\IEEEpeerreviewmaketitle

\IEEEraisesectionheading{\section{Introduction}\label{sec:introduction}}

\IEEEPARstart{O}{nline} social networks (Twitter, Facebook, Instagram, etc.) have been
important platforms for individuals to exchange information with their friends.
Social contagion \cite{burt1987social,DBLP:conf/kdd/KempeKT03,pastor2001epidemic,UganderBMK12} is a phenomenon that individuals are influenced by the information received from their social neighborhoods, e.g., acting the same as friends in sharing posts or adopting political opinions. 
Social decisions made by individuals often depend on the \emph{multiplicity of distinct social contexts} inside his/her contact neighborhood, which is termed  \emph{structural diversity} \cite{UganderBMK12,huang2013top,chang2017scalable}.  
Many studies on Facebook \cite{UganderBMK12} show that users are much more likely to join Facebook and become engaged if they have a larger structural diversity, i.e., a larger number of distinct social contexts. Given the important role of structural diversity, a fundamental problem of structural diversity search is to find the $r$ users with the highest structural diversity in graphs \cite{huang2013top,chang2017scalable}, which can be beneficial to political campaigns \cite{huckfeldt1995citizens}, viral marketing \cite{DBLP:conf/kdd/KempeKT03}, promotion of health practices \cite{UganderBMK12}, cooperation in social dilemmas \cite{qin2018neighborhood}, and so on.

The problem of structural diversity search has been recently studied based on two structural diversity models of $k$-sized component \cite{huang2013top,chang2017scalable} and $k$-core \cite{XHuang15}. 
However, one significant limitation of both models  is their limited decomposability for analyzing large-scale networks, 
which may lead to inaccurate reflection of social context diversity. To address this issue, in this paper, 
we propose a new structural diversity model based on $k$-truss.
A $k$-truss requires that every edge is contained in at least ($k$-2) triangles in the $k$-truss  \cite{cohen2008}. Intuitively, a $k$-truss signifies strong social ties among the members in this social group, while tending to break up weak-tied social groups and discard tree-like components. Our model treats each maximal connected $k$-truss as a distinct social context. 
As we will demonstrate, our model has several major advantages. First, thanks to $k$-truss, our model has a strong decomposability for analyzing large-scale networks at different levels of granularity. Second, a compact and elegant  index can be designed for efficient truss-based structural diversity search in a linear cost w.r.t. graph size. 
Third, when compared with other models, our model shows superiority in the evaluation of influence propagation on real-world networks.




\begin{figure}[t]
\centering \mbox{
\subfigure[Graph $G$]{\includegraphics[width=0.45\linewidth]{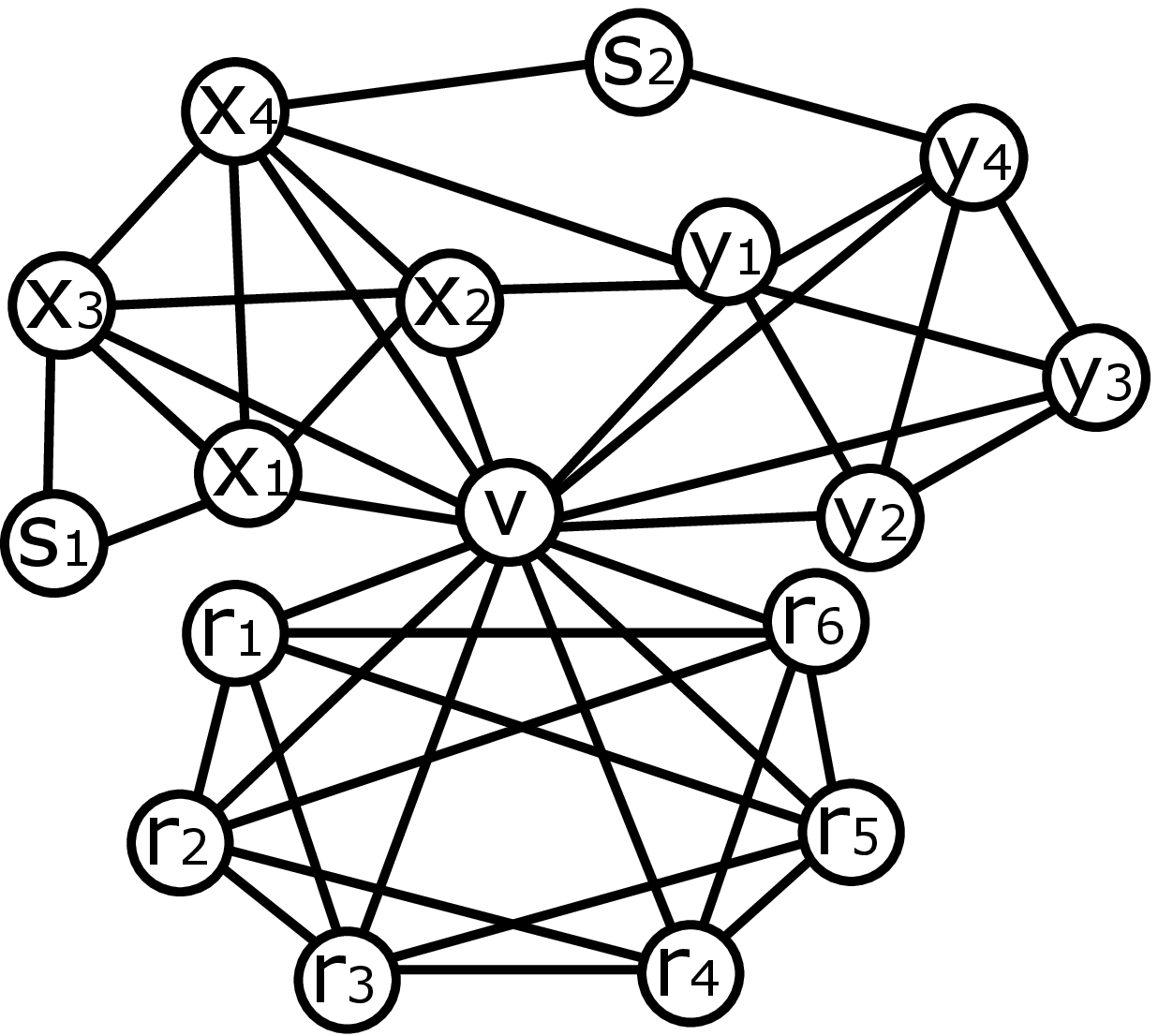}} \quad
\subfigure[Ego-Network $G_{N(v)}$]{\includegraphics[width=0.50\linewidth]{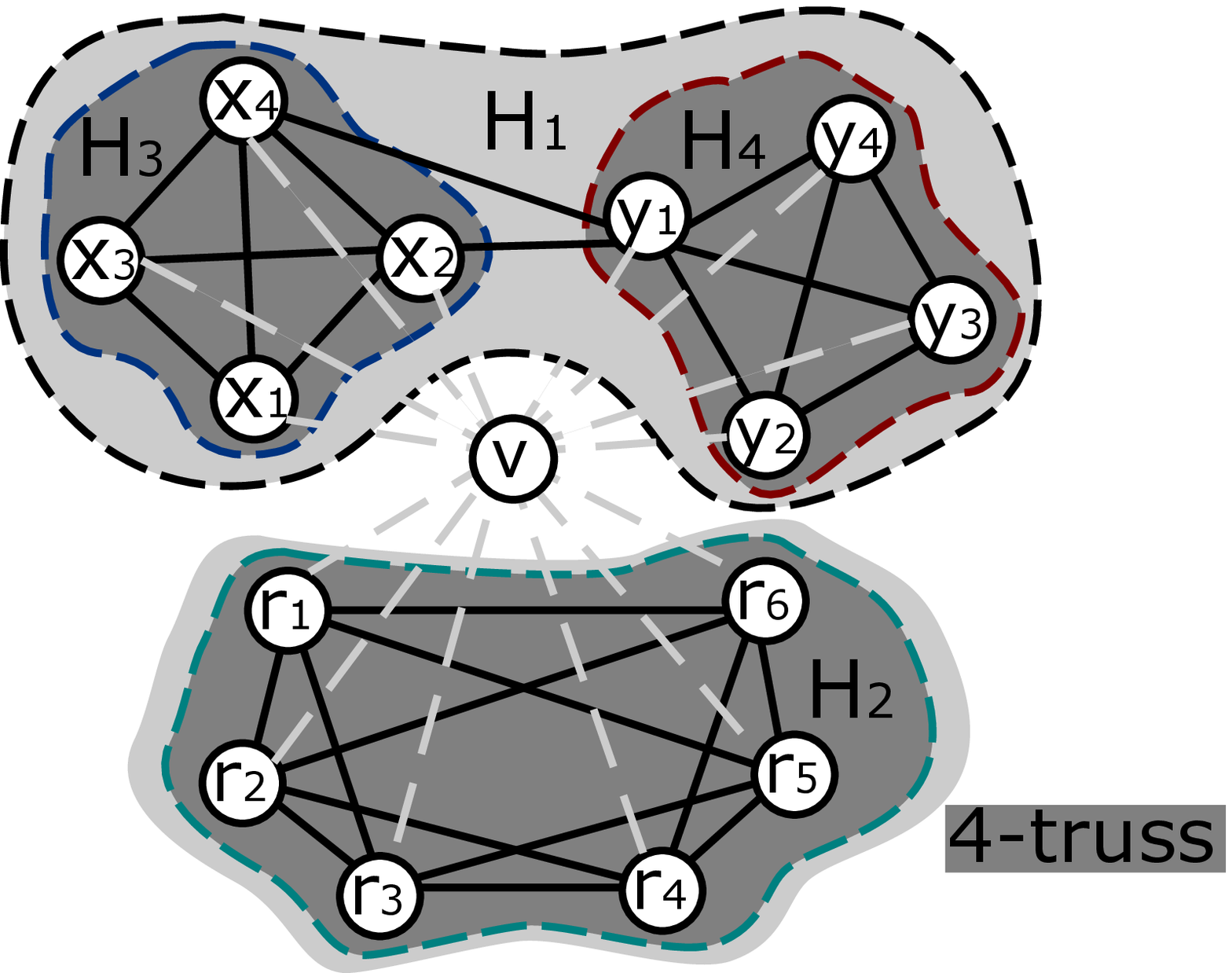}} }
\vspace*{-0.3cm}
\caption{A running example}
\label{fig.intro_example}
\end{figure}

\stitle{\emph{Motivating Example.}} Consider a social network $G$ in Figure~\ref{fig.intro_example}(a). The \ego of an individual $v$ is a subgraph of $G$ formed by all $v$'s neighbors as shown in the light gray region (excluding 
vertex $v$) in Figure~\ref{fig.intro_example}(b). To analyze the social contexts in Figure~\ref{fig.intro_example}(b), different structural diversity models have substantial differences:

\squishlisttight
\item 
\emph{Component-based structural diversity model} regards each connected component of vertex size at least $k$ as a social context \cite{huang2013top,chang2017scalable}. The component $H_1$ having 8 vertices is regarded as one social context. However, in terms of graph structure, two subgraphs $H_3$ and $H_4$ shown 
in Figure~\ref{fig.intro_example}(b) are loosely connected through edges $(x_2,y_1)$ and $(x_4,y_1)$, and vertices ($x_1$ and $x_3$) span long distances to vertices ($y_2$, $y_3$ and $y_4$). Thus, $H_3$ and $H_4$ can be reasonably treated as two different social contexts.
Unfortunately, the attempt of adjusting parameter \emph{$k$ using any value} does not help 
the decomposition of $H_1$.

\item 
\emph{Core-based structural diversity model} regards a maximal connected $k$-core as a social context \cite{UganderBMK12, XHuang15}. A $k$-core requires that every vertex has degree at least $k$ within the $k$-core. For \emph{$1\leq k\leq 3$}, $H_1$ is regarded as one maximal connected $k$-core, which cannot be decomposed into disjoint components; for \emph{$k\geq 4$}, $H_1$ is no longer counted as a feasible social context. 

\item 
\emph{Our truss-based structural diversity model} treats each maximal connected $k$-truss as a distinct social context. For $k=4$, $H_1$ is decomposed into two maximal connected 4-trusses $H_3$ and $H_4$ in Figure~\ref{fig.intro_example}(b), where each edge has at least two triangles. As a result, $H_2$, $H_3$ and $H_4$ are regarded as three distinct social contexts in the \ego of $v$, and the truss-based structural diversity of $v$ is 3.  
\end{list}

\jinbin{In light of the above example, truss-based structural diversity search is a pressing need. However, to the best of our knowledge, the problem of truss-based structural diversity search over graphs, has not been studied yet. In this paper, we invetigate the problem to find the $r$ vertices with the highest truss-based structural diversity and return their social contexts. We propose efficient algorithms for truss-based structural diversity search.}

\jinbin{However, efficient computation of truss-based structural diversity search raises significant challenges. A straightforward online search algorithm is to compute the structural diversity for all vertices and return the top-$r$ vertices, which is inefficient. 
Because it is costly 
to compute the structural diversity for all vertices in large graphs, from scratch without any pruning. The subgraph extraction of an \ego needs the costly operation of triangle listing \cite{Latapy08}, not even talking about the truss decompostion \cite{WangC12} 
for finding all maximal connected $k$-trusses. 
On the other hand, developing 
a diversity bound for pruning search space
is also difficult. Unlike the symmetry structure of \egos in the component-based model \cite{huang2013top,chang2017scalable}, non-symmetry structural properties 
restrict
our truss-based model to derive an efficient pruning bound. 
Therefore, existing structural diversity algorithms for component-based and core-based models \cite{huang2013top,chang2017scalable,XHuang15}  do not work for our 
truss-based model.   

}

\eat{Motivated by this example, in this paper, we study the problem of truss-based structural diversity search over graphs, that is to find the $r$ vertices with the highest truss-based structural diversity and return their social contexts. To tackle this problem, we first propose an online search algorithm to compute the structural diversity for all vertices and return  the top-$r$ vertices. Clearly, it is costly  to compute the structural diversity for all vertices from scratch without any pruning.
}


\jinbin{Fortunately, truss-based structural diversity has many desirable features for developing efficient indexes and algorithms. }
To  improve the efficiency of truss-based structural diversity search, we propose several 
useful optimization techniques. We develop an efficient top-$r$ search framework  
to prune vertices for avoiding structural diversity computation.
The heart of our framework is to exploit two important pruning techniques: (1) graph sparsification and (2) a diversity bound. Specifically, we first make use of structural properties of $k$-truss and propose graph sparsification to remove from the graph unqualified edges and nodes that will not be in any $k$-truss. Second, we develop an upper bound of diversity for pruning unqualified answers, leading to an early termination of our top-$r$ search. Furthermore, we develop a novel truss-based structural diversity index, called \TSDindex, which is a compact and elegant tree structure to keep the structural information for all \egos in $G$. Based on the \TSDindex, we propose an index-based top-$r$ search algorithm to quickly find answers. 
Furthermore, to explore the sharing computation across vertices, we utilize the global triangle listing one-shot for fast ego-network extraction and develop a fast bitmap technique for ego-network decomposition. Leveraging a new data structure of \ADVindex compressed from \TSDindex, we propose  \ADV for truss-based structural diversity search, which achieves a smaller index size and a faster query time.
To summarize, we make the following contributions:

\squishlisttight
\item We use a maximal connected $k$-truss to model a neighborhood 
social context in the \ego. We 
define the truss-based structural diversity and then 
formulate a new problem of truss-based structural diversity search over graphs. (Section~\ref{sec.problemdef})

\item We present a method of computing truss-based structural diversity using truss decomposition. Based on this, we develop an online search algorithm to tackle our problem, and 
give a comprehensive theoretical analysis of algorithm complexity. (Section~\ref{sec.baseline})

\item We analyze the structural properties of truss-based social contexts, and develop two useful pruning techniques of graph sparsification and a diversity bound. 
Equipped with them, we develop an efficient framework for structural diversity search with an early termination mechanism. (Section~\ref{sec.bound-search})

\item We design a space-efficient truss-based structural diversity index (\TSDindex) 
to keep the structural diversity information for all \egos. We propose a \TSDindex-based search algorithm to quickly find answers in a linear cost w.r.t. graph size. (Section~\ref{sec.indexbased})

\item 
 We propose  \ADV for truss-based structural diversity search based on the efficient techniques of fast ego-network truss decompostion and a compressed \ADVindex.
(Section~\ref{sec.bitmap})

\item 
We validate the efficiency and effectiveness of our proposed methods through extensive experiments. 
(Section~\ref{sec.exp})
\end{list} 

We discuss related work in Section~\ref{sec.related}, and conclude the paper with a summary in Section~\ref{sec.con}.

\section{Problem Definition}
\label{sec.problemdef}

We consider an undirected and unweighted simple graph $G=(V, E)$ with $n=|V|$ vertices and $m=|E|$ edges. We define $N(v)=\{u\in V: (v,u)\in E\}$ as the set of neighbors of a vertex $v$, and $d(v)=|N(v)|$ as the degree of $v$ in $G$. Let $d_{max}$ represent the maximum degree in $G$. For a set of vertices $S\subseteq V$, the induced subgraph of $G$ by $S$ is denoted by $G_S$, where the vertex set is $V(G_S) = S$ and the edge set is $E(G_S) = \{(v,u)\in E: v, u\in S\}$.  W.l.o.g. we assume that the considered graph $G$ is connected, indicating that $m \geq n - 1$ and  $n\in O(m)$. The assumption 
is similarly made in \cite{Latapy08,XHuang15}.  

\subsection{Ego-Network}  

We define an \ego~\cite{DBLP:journals/apin/DingZ18,mcauley2014discovering}  in the following.

\begin{definition}
\label{def.ego}
[Ego-Network] Given a vertex $v\in V$, the \ego of $v$, is a subgraph of $G$ induced by the vertex set $N(v)$, denoted by $G_{N(v)}$, where the vertex set $V(G_{N(v)})=N(v)$ and the edge set $E(G_{N(v)})$ $= \{(u, w)\in E: $ $u, w\in N(v)\}$.
\end{definition}

In the literature, the term ``neighborhood induced subgraph of $v$'' \cite{XHuang15} has also been used
to indicate the \ego of $v$, since the \ego is formed by all neighbors of $v$. 
For example, consider the graph $G$ in Figure~\ref{fig.intro_example}(a) and the vertex $v\in V$, the \ego of $v$ is shown as the gray region in Figure~\ref{fig.intro_example}(b), which is formed by the induced subgraph of $G$ by vertices $N(v)=\{x_1, \ldots, $ $ x_4, y_1, \ldots , $ $y_4, r_1,\ldots, r_6\}$
, \eat{Note that the \ego $G_{N(v)}$ excludes the center vertex $v$ with its incident edges (e.g., $(v, r_1)$) and also other irrelevant edges (e.g., $(s_2, x_4))$.} excluding the center vertex $v$ with its incident edges.
\eat{}

\subsection{Truss-based Social Context and Structural Diversity}

A triangle in $G$ is a cycle of length 3. Given three vertices $u, v, w \in V$, the triangle formed by $u, v, w$ is denoted by $\triangle_{uvw}$. Given a subgraph $H\subseteq G$, the support of an edge $e=(u,v)\in E(H)$ is defined as the number of triangles containing edge $e$ in $H$, i.e.,  $\sup_H(e)=|\{\triangle_{uvw} : (u, w), (v, w) \in E(H)\}|$. Figure~\ref{fig.trussandsupport}(a) shows the support of each edge in graph $H_1$.  There exists only one triangle $\triangle_{x_2x_4y_1}$ containing $(x_2,y_1)$, and $\sup_{H_1}(x_2, y_1) =1$.
We drop the subscript and denote the support as $\sup(e)$, when the context is obvious.

 A $k$-truss of graph $G$ is defined as the \text{largest} subgraph of $G$ such that every edge has support of at least $k-2$ in this subgraph \cite{WangC12,huang2014querying}. 
For a given $k \geq 2$, the $k$-truss of a graph $G$ is unique, which may be disconnected with multiple components. In our truss-based structural diversity model, we treat each connected component of the $k$-truss as a distinct \emph{social context}. The definition of social contexts in an \ego   is given below.  

\begin{definition}
[Social Contexts] Given a vertex $v$ and an integer $k\geq 2$, each connected component of the $k$-truss in  $G_{N(v)}$ is called a social context. Thus, the social contexts of $v$ are represented by all vertex sets of components, denoted by $\context(v) = \{ V(H): H $ is a connected component of the $k$-truss in $G_{N(v)} \}$.
\label{def.socialcontext}
\end{definition}

By Def.~\ref{def.socialcontext}, each social context is a component of $k$-truss, which is connected and also the maximal subgraph of the $k$-truss. Therefore, as an alternative,  we also call a social context as a \emph{maximal connected $k$-truss} throughout the paper. For example, consider an \ego $G_{N(v)}$ in Figure~\ref{fig.intro_example}(b) and $k=4$. The $4$-truss of $G_{N(v)}$ is presented by the darker gray region. We regard a connected component $H_3$ as a neighborhood social context in $G_{N(v)}$, which is represented by $V(H_3)=\{x_1, x_2, x_3, x_4\}$. Thus, the social contexts of $v$ have $\context(v) = \{\{x_1, x_2, $ $x_3, x_4\},\{y_1, y_2, y_3, y_4\}, $ $\{r_1, r_2, r_3, r_4, r_5, r_6\}\}$. 

Based on the definition of social contexts, we can define our key concept of \emph{truss-based structural diversity} as follows.

\begin{definition}
[Truss-based Structural Diversity]  Given a vertex $v$ and an integer $k\geq 2$, the truss-based structural diversity of $v$ is the multiplicity of social contexts $\context(v)$, denoted by $score(v) =|\context(v)|$.
\label{def.trussbased}
\end{definition}
The truss-based structural diversity is exactly the number of connected components of the $k$-trusses in the \ego.  Consider the \ego $G_{N(v)}$ in Figure~\ref{fig.intro_example}(b) and $k=4$, the $4$-truss of$G_{N(v)}$ has three connected components $H_2$, $H_3$, and $H_4$, thus $\score(v)=3$.

\eat{
\begin{example}
 Consider the \ego $G_{N(v)}$ of vertex $v$ in Figure~\ref{fig.intro_example}(b). Assuming that $k=3$, the \ego $G_{N(v)}$ has two maximal connected 3-trusses of $H_1$ and $H_2$; thus the truss-based structural diversity of vertex $v$ has $\score(v)=2$; For $k=4$, the \ego $G_{N(v)}$ has three maximal connected 4-trusses $H_2$, $H_3$, and $H_4$; thus $\score(v)=3$. 
\end{example}
}


\begin{figure}[t]
\centering \mbox{
\subfigure[Support]{\includegraphics[width=0.45\linewidth]{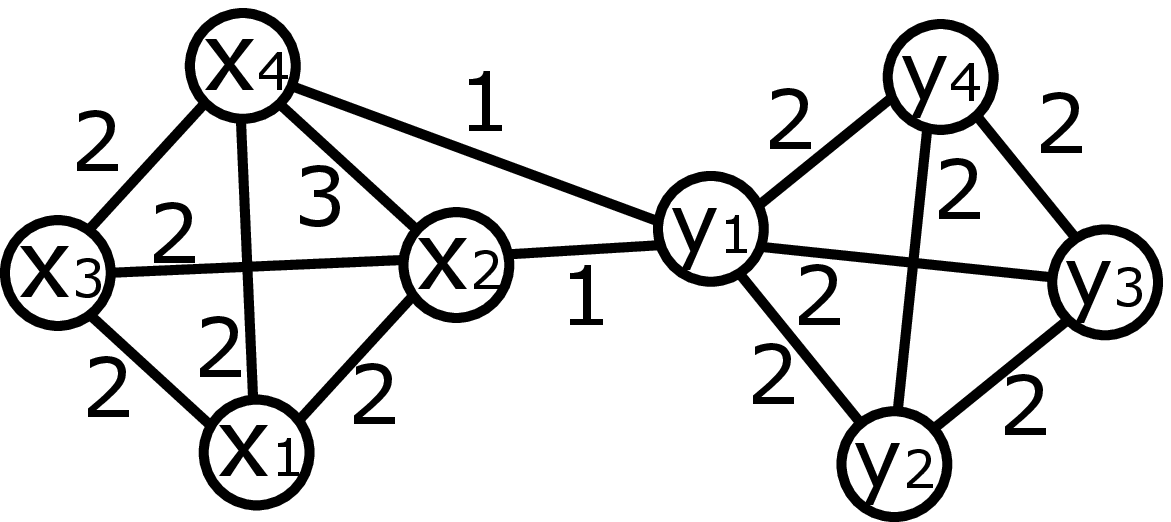}} 
\quad
\quad
\subfigure[Trussness]{\includegraphics[width=0.45\linewidth]{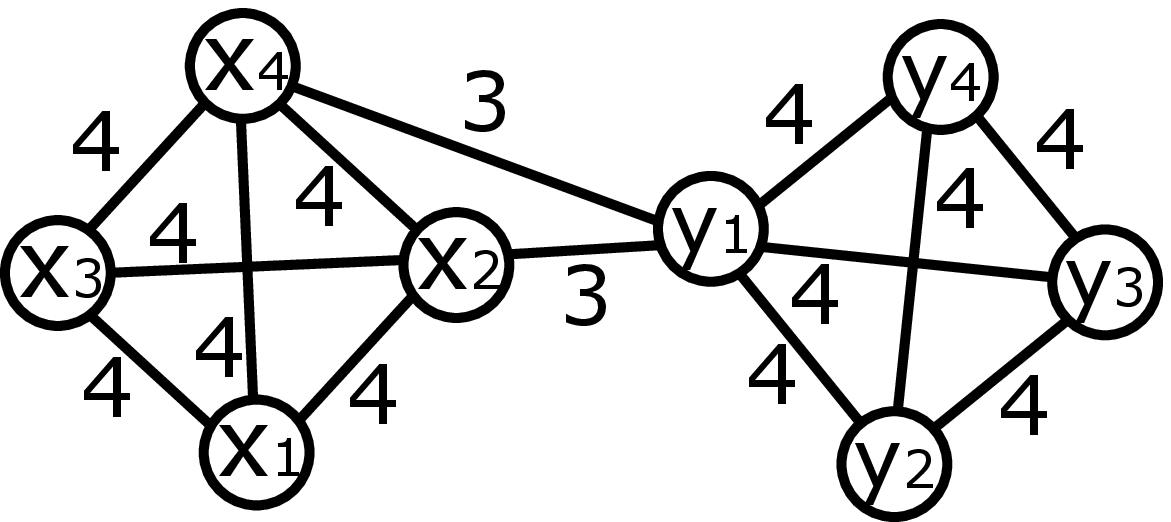}} 
}
\caption{The support and trussness of edges in $H_1$}
\label{fig.trussandsupport}
\end{figure}

\subsection{Problem Statement}
The problem of 
truss-based structural diversity search studied in this paper is formulated as follows.

\stitle{Problem statement}: Given a graph $G$ and two integers $r$ and $k$ where $1\leq r\leq n$  and $ k \geq 2$, the goal of top-$r$ truss-based structural diversity search  is to find a set of $r$ vertices in $G$  having the highest scores of truss-based structural diversity w.r.t.\ the trussness threshold $k$, and return their social contexts.

Consider the graph $G$ in Figure~\ref{fig.intro_example} with $r=1$ and $k=4$, the answer of our problem is the vertex $v$, which has the highest structural diversity $\score(v)= 3$ and its social contexts $\context(v) = \{\{x_1, x_2, x_3, x_4\},\{y_1, $ $y_2, y_3, y_4\},\{r_1, r_2, r_3, r_4, r_5, r_6\}\}$.




\section{Online Search Algorithm}\label{sec.baseline}
In this section, we develop an online search algorithm for top-$r$ truss-based structural diversity search. The idea of our method is intuitively simple. The algorithm first computes the structural diversity score for each vertex in $G$, and then returns an answer of $r$ vertices having the highest scores and their social contexts. 
In the following, we first introduce the truss decomposition for finding all $k$-trusses in a graph. Leveraging truss decomposition, we then present a procedure for structural diversity score computation. Finally, we present our online search algorithm and analyze the algorithm complexity.


\subsection{Truss Decomposition}

\stitle{Trussness}. 
We start with a useful definition of  trussness below.

\begin{definition}
[Trussness] Given a subgraph $H\subseteq G$, the trussness of $H$ is defined as the minimum support of edges in $H$ plus 2, denoted by $\tau(H) =\min_{e\in E(H)}$ $\{\sup_H(e)+2\}$. The trussness of an edge $e\in H$ denoted by $\tau_H(e)$ is defined as the largest number $k$ such that there exists a connected $k$-truss $H'\subseteq H$ containing $e$, i.e., $$\tau_H(e) = \max_{H'\subseteq H, e\in E(H')} \tau(H').$$
\vspace*{-0.2cm}
\label{def.edgetruss}
\end{definition}
Similar to the notation of support, we drop the subscript and denote the trussness $\tau_H(e)$ as $\tau(e)$ when the context is obvious. Also we can define the trussness of a vertex $v$ in the similar way, i.e., $\tau_H(v) = \max_{H'\subseteq H, v\in V(H')} \tau(H')$.

\begin{example}
Figure~\ref{fig.trussandsupport}(b) shows the trussness of each edge in graph $H_1$. First, according to the edge support in Figure~\ref{fig.trussandsupport}(a),  the trussness of subgraph $H_1$ is $\tau(H_1) = \min_{e\in E(H_1)}$ $\{\sup_{H_1}(e)+2\} = 1+2=3$. Thus, we have $\tau_{H_1}(x_2, y_1)= $ $\max_{H'\subseteq H_1, e\in E(H')}$ $ \tau(H') = 3$. 
\eat{
For an edge $(x_2, x_4)$, since a subgraph $H_3$ of $H_1$ in Figure~\ref{fig.intro_example}(b) is a 4-truss containing it, the trussness $\tau_{H_1}(x_2, x_4) =4$. There exists no 5-truss subgraph of $H_1$ containing $(x_2, x_4)$.  
Even if $\sup_{H_1}(x_2, x_4)=3$, the trussness  $\tau_{H_1}(x_2, x_4) =4$ could be less than  $\sup_{H_1}(x_2, x_4)+2 = 5$.

} 
\end{example}

\stitle{Algorithm of truss decomposition.}  Truss decomposition on graph $G$ is to find the $k$-trusses of $G$ for all possible $k$'s. Given any number $k$, the $k$-truss of $G$ is the union of all edges with trussness at least $k$. Equally, truss decomposition on graph $G$ is to compute the trussness of each edge in $G$. 

For the self-completeness of our techniques and reproducibility, the detailed algorithm of truss decomposition \cite{WangC12} is presented in Algorithm~\ref{algo:truss-decomp}. \jinbin{The algorithm starts from the computation of the support $\sup_G(e)$ for each edge $e\in E$, using the technique of triangle listing (line 1). It sorts all edges in the ascending order of their support, using the efficient technique of bin sort \cite{cormen2009introduction} (line 2). Let $k$ start from 2. The algorithm iteratively removes from graph $G$ an edge $e$ with the lowest support of $\sup_G(e) \leq k-2$, and assigns the trussness $\tau_G(e) = k$ (lines 5-6 and 11). Meanwhile, it updates the support of other affected edges due to the deletion of edge $e$ (lines 7-10). The algorithm terminates when the remaining graph $G$ is empty; Otherwise, it increases the number $k$ by 1 and repeats the above process of edge removal. Finally, it computes the trussness of each edge $e$ in $G$.
}

\begin{algorithm}[t]
\small
\caption{Truss Decomposition \protect\cite{WangC12}} \label{algo:truss-decomp}
\begin{flushleft} 
\textbf{Input:} $G=(V, E)$\\
\textbf{Output:} $\tau(e)$ for each $e\in E$\\
\end{flushleft}
\vspace*{-0.3cm}
\
\begin{algorithmic}[1]

\STATE Compute $sup_G(e)$ for each edge $e\in E$;

\STATE Sort all the edges in ascending order of their support;

\STATE $k \leftarrow 2$;

\STATE \textbf{while} ($\exists e$ such that $sup_G(e) \leq (k-2)$)

\STATE \hspace{0.3cm} Let $e=(u,v)$ be the edge with the lowest support;

\STATE \hspace{0.3cm} Assume, w.l.o.g, $d(u) \leq d(v)$;

\STATE \hspace{0.3cm} \textbf{for} (each $w\in N(u)$ \textbf{and} $(u,v)\in E$ \textbf{do})

\STATE \hspace{0.6cm} $sup_G((u,w)) \leftarrow sup_G((u,w))-1$;

\STATE \hspace{0.6cm} $sup_G((v,w)) \leftarrow sup_G((v,w))-1$;

\STATE \hspace{0.6cm} Reorder $(u,w)$ and $(v,w)$ according to their new support;

\STATE \hspace{0.3cm} $\tau_G(e) \leftarrow k$, remove $e$ from $G$;

\STATE if(not all edges in $G$ are removed)

\STATE \hspace{0.3cm} $k \leftarrow k+1$;

\STATE \hspace{0.3cm} \textbf{Goto} Step 4;

\STATE return $\{\tau_G(e)| e\in E\}$;

\end{algorithmic}
\end{algorithm}

\subsection{Computing $score(v)$}
Algorithm~\ref{algo:comp-score} presents a procedure of computing $score(v)$, which calculates the number of maximal connected $k$-trusses in the \ego $G_{N(v)}$. The algorithm first extracts $G_{N(v)}$ from graph $G$ (line 1), and then applies the truss decomposition in Algorithm~\ref{algo:truss-decomp} on $G_{N(v)}$ (line 2). After obtaining the trussness of all edges, it removes all the edges $e$ with $\tau_{G_{N(v)}}(e) < k$ from $G_{N(v)}$ (line 3). The remaining graph $G_{N(v)}$ is the union of all maximal connected $k$-trusses. Applying the breadth-first-search, all connected components are identified as the social contexts $\context(v) =  \{ V(H): H $ is a maximal connected $k$-truss in $G_{N(v)} \}$ (line 4).  Algorithm~\ref{algo:comp-score} finally returns the structural diversity $score(v) = |\context(v)|$ (lines 5-6).

\begin{algorithm}[t]
\small
\caption{Computing $score(v)$} \label{algo:comp-score}
\begin{flushleft} 
\textbf{Input:} $G=(V, E)$, a vertex $v$,  the trussness threshold $k$\\
\textbf{Output:} $score(v)$\\
\end{flushleft}
\vspace*{-0.3cm}
\
\begin{algorithmic}[1]

\STATE Extract an \ego of $v$ as $G_{N(v)}$ from $G$ by Def.~\ref{def.ego};

\STATE Apply the truss decomposition on $G_{N(v)}$ using Algorithm~\ref{algo:truss-decomp};

\STATE Remove all edges $e$ with $\tau_{G_{N(v)}}(e) < k$ from $G_{N(v)}$;

\STATE 
Identify all connected components in $G_{N(v)}$ as 
the social contexts $\context(v) =  \{ V(H): H $ is a maximal connected $k$-truss in $G_{N(v)} \}$;

\STATE $score(v) \leftarrow |\context(v)|$; 

\STATE \textbf{return} $score(v)$;
\end{algorithmic}
\end{algorithm}

\subsection{Online Search Algorithm}
\jinbin{Equipped with the procedure of computing $score(v)$, we present an online search algorithm to address the problem of top-$r$ structural diversity search, as shown in Algorithm~\ref{algo:baseline-alg}. }
It computes the structural diversity for all vertices in graph $G$ from scratch. Algorithm~\ref{algo:baseline-alg} first initializes an answer set $\mathcal{S}$ as empty (line 1). Then, each vertex $v \in V$ is enumerated to compute the structural diversity using Algorithm~\ref{algo:comp-score} (lines 2-3). The algorithm compares $score(v)$ with the smallest structural diversity in the answer set $\mathcal{S}$, and 
checks whether $v$ should be added into answer set $\mathcal{S}$ 
(lines 4-7). Finally, Algorithm~\ref{algo:baseline-alg} terminates by returning the answer set $\mathcal{S}$ and their social contexts $\context(v)$ for $v\in \mathcal{S}$ (line 8). 

\begin{example}
We apply  Algorithm~\ref{algo:baseline-alg} on graph $G$ in Figure~\ref{fig.intro_example} with $k=4$ and $r=1$. Accordingly, it computes the structural diversity for each vertex in $G$ and invokes Algorithm~\ref{algo:comp-score} in total of $|V|=17$ times. Finally, we obtain the top-1 structural diversity result of vertex $v$ with $score(v)=3$. 
\end{example}

\begin{algorithm}[t]
\small
\caption{Online Search Algorithm} \label{algo:baseline-alg}
\begin{flushleft} 
\textbf{Input:} $G=(V, E)$, an integer $r$, the trussness threshold $k$\\
\textbf{Output:} Top-$r$ truss-based structural diversity results\\
\end{flushleft}
\vspace*{-0.3cm}
\
\begin{algorithmic}[1]

\STATE Let an answer set $\mathcal{S} \leftarrow \emptyset$;

\STATE \textbf{for} each vertex $v\in V$

\STATE  \hspace{0.3cm}  Computing $score(v)$ using Algorithm~\ref{algo:comp-score};

\STATE  \hspace{0.3cm}  \textbf{if} $|\mathcal{S}|<r$ \textbf{then} $\mathcal{S} \leftarrow \mathcal{S} \cup \{v\}$;

\STATE  \hspace{0.3cm}  \textbf{else if} $score(v) > \min_{v'\in \mathcal{S}}score(v')$ \textbf{then}

\STATE  \hspace{0.3cm}  \hspace{0.3cm}  $u\leftarrow \arg\min_{v'\in \mathcal{S}}score(v')$;

\STATE  \hspace{0.3cm}  \hspace{0.3cm}  $\mathcal{S} \leftarrow (\mathcal{S}- \{u\} ) \cup \{v\}$;

\STATE \textbf{return} $\mathcal{S}$ and their social contexts $\context(v)$ for $v\in \mathcal{S}$; 

\end{algorithmic}
\end{algorithm}

\eat{
\stitle{Complexity Analysis. }
Algorithm~\ref{algo:baseline-alg} runs on graph $G$ taking $O(\rho (m+\mathcal{T}))$ time and $O(m)$ space, where $\mathcal{T}$ is the number of triangles in $G$, $\rho$ is the arboricity  \cite{ChibaN85} of $G$, and $\rho \leq \min\{\lfloor \sqrt{m} \rfloor, d_{max}\}$. Please refer to Appendix B for the detailed proof.
}

\jinbin{

\subsection{Complexity Analysis}

\begin{lemma}\label{lemma.score}
Algorithm~\ref{algo:comp-score} computes $score(v)$ for $v$ in $O(\sum_{u\in N(v)}$ $\min\{d(u), d(v)\} + \sum_{(u, w )\in E(G_{N(v)})}\min\{d(u), d(w)\})$ time and $O(m)$ space.
\end{lemma}

\begin{proof}
The algorithm obtains \ego $G_{N(v)}$ from $G$ (line 1 of  Algorithm~\ref{algo:comp-score})  taking $O(\sum_{u\in N(v)}\min\{d(u), d(v)\})$ time, since it needs to list all triangles $\triangle_{vuw}$ containing $v$ to enumerate the edges $(u,w)\in E(G_{N(v)})$ \cite{tsourakakis2009doulion}. Second, for $G_{N(v)}$ associated with the edge set  $E(G_{N(v)})$, the step of applying truss decomposition on $G_{N(v)}$ (line 2 of  Algorithm~\ref{algo:comp-score}) takes $O(\sum_{(u, w )\in E(G_{N(v)})}$ $\min\{d(u), d(w)\}$ time \cite{huang2014querying}. In addition, the other two steps of edge removal and component identification both take $O(|E(G_{N(v)})|) $ $ \subseteq$ $ O(\sum_{(u,w)\in E(G_{N(v)})} 1)  $ time. Overall, the time complexity of  Algorithm~\ref{algo:comp-score} is $O(\sum_{u\in N(v)} \min\{d(u),$ $d(v)\} + $ $\sum_{(u,w) \in E(G_{N(v)})} $ $ \min\{d(u), d(w)\})$. 

We analyze the space complexity. Because of $G_{N(v)} \subseteq G$, an \ego  $G_{N(v)}$ takes $O(n+m)$ space. The social contexts $\context(v)$ take  $O(n)$ space. Hence, the space complexity of Algorithm~\ref{algo:comp-score} is $O(n+m)\subseteq O(m)$, due to $n\in O(m)$ by our assumption of graph connectivity.
\end{proof}

\begin{theorem}\label{theorem.baseline}
Algorithm~\ref{algo:baseline-alg} runs on graph $G$ taking 
$$O( \sum_{v\in V} \{
 \sum_{u\in N(v)} \min\{d(u), d(v)\} 
 + \sum_{(u,w) \in E(G_{N(v)})} \min\{d(u),$$$$ d(w)\}
 \} )$$
 time and $O(m)$ space.
\end{theorem}
\begin{proof}
Algorithm~\ref{algo:baseline-alg} uses Algorithm~\ref{algo:comp-score} to compute $score(v)$ for each vertex $v\in V$, which totally takes 
$O( \sum_{v\in V} \{
 \sum_{u\in N(v)} $ $\min\{d(u), d(v)\} $ $
 + \sum_{(u,w) \in E(G_{N(v)})} $ $\min\{d(u), d(w)\}
 \} )$
 time by Lemma~\ref{lemma.score}. 
 Moreover the top-$r$ results $S$ can be maintained in $O(n)$ time and $O(n)$ space, using 
 bin sort. As a result, Algorithm~\ref{algo:baseline-alg} takes $O( \sum_{v\in V} \{
 \sum_{u\in N(v)} $ $  \min\{d(u), d(v)\} $ $
 + $ $\sum_{(u,w) \in E(G_{N(v)})} $ $\min\{d(u), d(w)\}
 \} )$ time and $O(m+n) $ $\subseteq O(m)$ space. 
\end{proof}

\stitle{Complexity Simplification.}  Theorem~\ref{theorem.baseline} has a tight time complexity, but in a very complex form. We relax the time complexity to simplify form using graph arboricity \cite{ChibaN85}. 
Specifically, the arboricity $\rho$ of a graph $G$ is defined as the minimum number of spanning trees that cover all edges of graph $G$, and $\rho \leq \min\{\lfloor \sqrt{m} \rfloor, d_{max}\}$ \cite{ChibaN85}. For any subgraph $g \subseteq G$, the arboricity $\rho_{g}$ of $g$ has $\rho_{g} \leq \rho$. We have the following theorem. 

\begin{theorem}\label{theorem.final-baseline}
Algorithm~\ref{algo:baseline-alg} runs on graph $G$ taking $O(\rho (m+\mathcal{T}))$ time and $O(m)$ space, where  $\rho$ is the arboricity of $G$ and $\mathcal{T}$ is the number of triangles in $G$.
\end{theorem}

\begin{proof} 
According to \cite{ChibaN85}, $O(\sum_{(u,w) \in E(G)}$ $ \min $ $\{d(u), d(v)\}) $ $\subseteq $ $O(\rho m),$ where $\rho$ is the arboricity of $G$. Thus, we have 
$$O( \sum_{v\in V} \{\sum_{u\in N(v)} \min\{d(u), d(v)\}\}) $$ $$\subseteq O(\sum_{(v, u)\in E} \min\{d(v), d(u)\}\}) \subseteq O(\rho m).$$

Now, we consider the remaining part of time complexity in Theorem~\ref{theorem.baseline} using the arboricity of \egos. For a vertex $v \in V$, the \ego $G_{N(v)}$ has $n_v$ vertices and $m_v$ edges, where $n_v = |N(v)|$ and $m_v = |\{\triangle_{vuw}: u, w\in N(v), (u,w)\in E\}|$. Let the number of triangles in graph $G$ be $\mathcal{T}$, and obviously $\mathcal{T} =\frac{\sum_{v\in V} m_v}{3}$. In addition, as $G_{N(v)}$ $\subseteq G$, the arboricity $\rho_{v}$ of 
$G_{N(v)}$ has $\rho_{v} \leq \rho$. As a result, we have:

$$O( \sum_{v\in V} \{\sum_{(u, w )\in E(G_{N(v)})}\min\{d(u), d(w)\}\})$$  $$\subseteq O(\sum_{v\in V} \rho_{v} m_{v}) \subseteq O(\rho \cdot \sum_{v\in V} m_{v}) \subseteq O(\rho \mathcal{T}).$$   

Combining the above two equations, we have:
 
$$O( \sum_{v\in V} \{
 \sum_{u\in N(v)} \min\{d(u), d(v)\} 
 + \sum_{(u,w) \in E(G_{N(v)})} \min\{d(u),$$ $$ d(w)\}
 \} )$$
$$\subseteq O(\rho (m+\mathcal{T})).$$
\end{proof}


}

\section{An Efficient Top-r Search Framework}\label{sec.bound-search}

The online search algorithm is inefficient for top-$r$ search, because it computes the structural diversity for all vertices on the entire graph. To improve the efficiency, we develop an efficient top-$r$ search framework in this section. The heart of our framework is to exploit two important pruning techniques: (1) graph sparsification and (2) upper bounding $score(v)$.  
\eat{The general idea is to shrink the graph size for a fast computation of $score(v)$, and 
terminate the search process early using a strategy of bound pruning, which reduces the number of vertices whose structural diversities are computed. We will first present the techniques of graph sparsification and upper bounding $score(v)$, and then outline our improved algorithm.
}

\subsection{Graph Sparsification}
The goal of graph sparsification is to remove from graph $G$ the unnecessary vertices and edges, which are not included in the maximal connected $k$-truss for any \ego. This removal does not affect the answer, but shrinks the graph size for efficiency improvement. 

\stitle{Structural Properties of $k$-truss.} 
We start from a structural property of $k$-truss. 
\eat{}


\begin{property}\label{lemma.trussfilter}
Given an edge $e^* \in E$, if $\tau_G(e^*) < (k+1)$, $e^*$ will not be included in any maximal connected $k$-truss in the \ego $G_{N(v)}$ for any vertex $v\in V$.
\end{property}

\begin{proof}
We prove it by contradiction. Assume that $G_{N(v)}$ has a maximal connected $k$-truss $H \subseteq G_{N(v)} $ containing $e^*$, where $|V(H)| \geq k$ and $\sup_{H}(e) \geq k-2$ for any $e\in E(H)$. Then, we add the vertex $v$ and its incident edges to $H$, to generate another subgraph $H'$ of $G$ where $V(H')=V(H)\cup \{v\}$ and $E(H') = E(H) \cup \{(v, u): u\in V(H)\}$.  It is easy to verify that for any $e\in E(H')$, $\sup_{H'}(e)\geq (k-2)+1=k-1$ holds. Thus, the trussness of $H'$ has $\tau(H')\geq k+1$. By Def.~\ref{def.edgetruss}, the trussness of $e^*\in E(H')$ in graph $G$ has $\tau_G(e^*) $ 
 $\geq \tau(H') \geq k+1 $, which is a contradiction. 
\end{proof}

Based on Property~\ref{lemma.trussfilter}, we can safely remove any edge $e$ with $\tau_G(e) < (k+1)$ from graph $G$. The details of graph sparsification are described as follows. Specifically, we first apply truss decomposition \cite{WangC12} on graph $G$ to obtain the trussness of all edges, and then delete all the edges $e$ with $\tau_G(e) < (k+1)$ from $G$. Due to the removal of edges, some vertices may become isolated. We continue to delete all isolated nodes from $G$. 
\jinbin{
Obviously, graph sparsification is a useful preprocessing step, which benefits efficiency improvement in the following aspects. On one hand, it reduces the graph size of $G$ and \egos, leading to a fast computation of structural diversity. On the other hand, it avoids computing structural diversity for those isolated vertices.  In the following, we discuss the practicality of graph sparsification on real-world datasets, based on the analysis of 
edge trussness distribution.


\begin{figure}[t]
\centering 
\small
\includegraphics[width=0.7\linewidth]{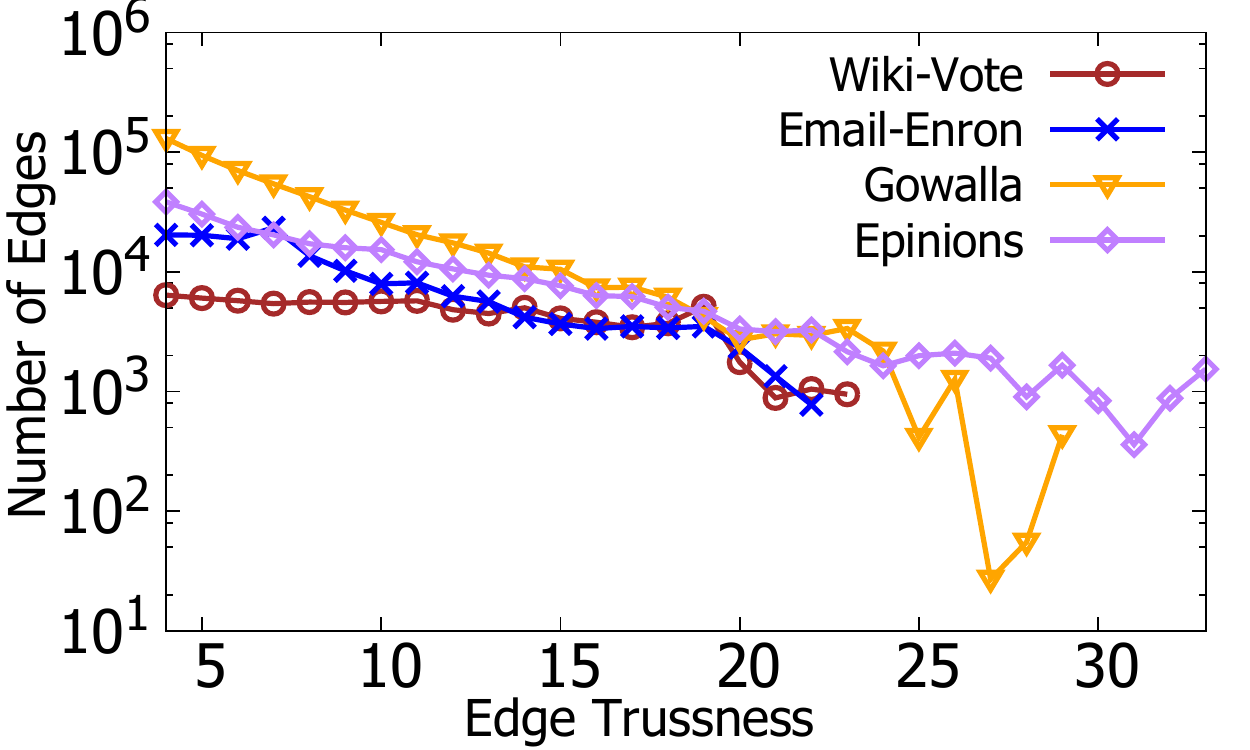}
\caption{The number of edges versus the increased edge trussness on four real-world graphs.}

\label{fig.truss_distribute}
\end{figure}

\stitle{Edge Trussness Distribution. } 
Figure~\ref{fig.truss_distribute} shows the edge trussness distribution of four real-world networks including Wiki-Vote, Email-Enron, Gowalla, and Epinions \cite{snapnets}. 
The range of edge trussness falls in [2, 33]. The number of edges in the y-axis are shown in the log plot. 
As we can see, the larger trussness is, the less number of edges has. Most edges have small trussness that can be filtered by graph sparsification. According to our statistics, graph sparsification can remove 45\% edges and 6.8\% isolated nodes from these four graphs on average for $k=5$. The significant pruning performance shows the technique of graph sparsification \jinbin{is well applicable} for our structural diversity search. In addition, we interestingly find that the number of edge trussness has a heavy-tailed distribution following a power-law property, which is similar to the vertex degree distribution \cite{barabasi1999emergence,faloutsos1999power}. 
}

\subsection{An Upper Bound of $score(v)$}\label{sec.upperbound}

\jinbin{
In this section, we analyze the structural properties of \egos and develop a tight upper bound of $score(v)$. Symmetry structure of \egos lends themselves to derive an efficient upper bound of structural diversity \cite{huang2013top,chang2017scalable}. However, the same symmetry properties fails in our truss-based structural diversity model. The following observation formalizes the property of non-symmetry.

\stitle{Non-Symmetry.} Consider three vertices $u$, $v$, $w$ form a triangle $\triangle_{uvw}$ in $G$. The non-symmetry of truss-based structural diversity shows that the edges $(v,w)$, $(u,w)$, $(u,v)$ may have different trussnesses in the \egos $G_{N(u)}$, $G_{N(v)}$, $G_{N(w)}$ respectively. In other words, $\tau_{G_{N(u)}}(v,w)$, $\tau_{G_{N(v)}}(u,w)$, and $\tau_{G_{N(w)}}(u,v)$ may not be the same. For example, we consider three vertices $v$, $r_1$, and $r_2$ in graph $G$ shown in Figure~\ref{fig.intro_example}(a). For \ego $G_{N(v)}$, we have $\tau_{G_{N(v)}}(r_1, r_2) = 4$; For \ego $G_{N(r_1)}$, we have $\tau_{G_{N(r_1)}}(v, r_2) = 3$. As a result,  $\tau_{G_{N(v)}}(r_1, r_2)$ $ \neq $ $\tau_{G_{N(r_1)}}$ $(v, r_2)$. The following observation formalizes this property of non-symmetry. 

\begin{observation}\label{lemma.symmetricity}
(Non-Symmetry) Consider an edge $e=(v, u)$ $\in E$ and a common neighbor $w\in N(v)\cap N(u)$. The \egos $G_{N(v)}$ and $G_{N(u)}$ have non-symmetry structure for vertex $w$ as follows. Even if edge $(u,w)$ in the \ego $G_{N(v)}$ has $\tau_{G_{N(v)}}(u,w) \geq k$, edge $(v,w)$ in the \ego $G_{N(u)}$ may have $\tau_{G_{N(u)}}(v,w) < k$ . 
\end{observation}

In view of this result, we infer that given an edge $(v, u)\in E$, the prospects for exploiting the process of computing $score(v)$ to derive an upper bound for $score(u)$ are not promising. It shows significant challenges for deriving an upper bound. The truss-based structural diversity cannot enjoy the nice symmetry properties of component-based structural diversity \cite{huang2013top,chang2017scalable}, which also brings challenges for score computation. We next investigate the structural properties of maximal connected $k$-truss, in search of prospects for an upper bound of $score(v)$. 
}

\stitle{An upper bound $\overline{score}(v)$.} Consider that the smallest maximal connected $k$-truss is a completed graph of $k$ vertices as $k$-clique. A $k$-clique has $k$ vertices and $\frac{k(k-1)}{2}$ edges. Based on the analysis of \ego size, we can infer the following useful lemma.

\begin{lemma}\label{lemma.trussbound}
For a vertex $v\in V$,  $score(v)$ has an upper bound of $\overline{score}(v) = \min\{\lfloor \frac{d(v)}{k}\rfloor, $ $ \lfloor \frac{2m_v}{k(k-1)} \rfloor\}$, where $m_v$ is the number of edges in \ego $G_{N(v)}$.  Thus, $score(v) \leq \overline{score}(v)$ holds. 
\end{lemma}

\begin{proof} 
First, $G_{N(v)}$ has $d(v)$ vertices. Since the minimum vertex size of a maximal connected $k$-truss is $k$, $G_{N(v)}$ has at most $\lfloor \frac{d(v)}{k}\rfloor$ maximal connected $k$-trusses in $G_{N(v)}$. Thus, $score(v) \leq \lfloor \frac{d(v)}{k}\rfloor$ holds. Second, $G_{N(v)}$ has $m_v$ edges. Since the minimum edge size of a maximal connected $k$-truss is $\frac{k(k-1)}{2}$ edges, $G_{N(v)}$ has at most $\lfloor  \frac{2m_v}{k(k-1)} \rfloor$ maximal connected $k$-trusses in $G_{N(v)}$. As a result, $score(v) \leq \min\{\lfloor \frac{d(v)}{k}\rfloor, $ $ \lfloor \frac{2m_v}{k(k-1)} \rfloor\}$ $=\overline{score}(v)$ holds. 
\end{proof}

\begin{algorithm}[t]
\small
\caption{Efficient Truss-based Top-$r$ Search Framework} \label{algo:bound-search}
\begin{flushleft} 
\textbf{Input:} $G=(V, E)$, an integer $r$, the trussness threshold $k$\\
\textbf{Output:} Top-$r$ truss-based structural diversity results\\
\end{flushleft}
\
\begin{algorithmic}[1]

\STATE  Apply the graph sparsification on $G$ by removing all edges $e$ with $\tau_{G}(e)\leq k$ and isolated nodes;

\STATE  \textbf{for} $v\in V$ \textbf{do}

\STATE  \hspace{0,3cm} $ \overline{score}(v) \leftarrow \min{\{\lfloor \frac{d(v)}{k}\rfloor , \lfloor \frac{2m_v}{k(k-1)} \rfloor\}}$;

\STATE $\mathcal{L} \leftarrow$ sort all vertices $V$ in descending order of $\overline{score}(v)$;

\STATE  $\mathcal{S} \leftarrow \emptyset$;

\STATE  \textbf{while} $\mathcal{L}\neq \emptyset$

\STATE  \hspace{0.3cm} $v^* \leftarrow \arg\max_{v\in \mathcal{L}} \overline{score}(v)$; Delete $v^*$ from $\mathcal{L}$;

\STATE  \hspace{0.3cm}  \textbf{if} $|\mathcal{S}|=r$ and $\overline{score}(v^*) \leq \min_{v\in \mathcal{S}}score(v)$ \textbf{then}

\STATE  \hspace{0.3cm}  \hspace{0.3cm} \textbf{break};


\STATE  \hspace{0.3cm}  Computing $score(v^*)$ using Algorithm~\ref{algo:comp-score};

\STATE  \hspace{0.3cm}  \textbf{if} $|\mathcal{S}|<r$ \textbf{then} $\mathcal{S} \leftarrow \mathcal{S} \cup \{v^*\}$;

\STATE  \hspace{0.3cm}  \textbf{else if} $score(v^*) > \min_{v\in \mathcal{S}}score(v)$ \textbf{then}

\STATE  \hspace{0.3cm}  \hspace{0.3cm}  $u\leftarrow \arg\min_{v\in \mathcal{S}}score(v)$;

\STATE  \hspace{0.3cm}  \hspace{0.3cm}  $\mathcal{S} \leftarrow (\mathcal{S}-  \{u\} ) \cup \{v^*\}$;

\STATE  \textbf{return} $\mathcal{S}$ and their social contexts $\context(v)$ for $v\in \mathcal{S}$; 

\end{algorithmic}
\end{algorithm}

\subsection{An Efficient Top-$r$ Search Framework}
Equipped with graph sparsification and an upper bound $\overline{score}(v)$
, we propose 
our efficient truss-based top-$r$ search framework as follows. 

\stitle{Algorithm.} Algorithm~\ref{algo:bound-search} outlines the details of  truss-based top-$r$ search framework. It first performs graph sparsification by applying truss decomposition on graph $G$ and removing  all the edges $e$ with $\tau_{G}(e)\leq k$ and isolated nodes from $G$ (line 1). Then, it computes the upper bound  of $\overline{score}(v)$ for each vertex $v\in V$ and sorts them in the decreasing order in $\mathcal{L}$ (lines 2-4). Next, the algorithm iteratively pops out a vertex $v^*$ with the largest  $\overline{score}(v)$ from $\mathcal{L}$ (lines 7). After that, the algorithm checks an early stop condition. If the answer set $\mathcal{S}$ has $r$ vertices and $\overline{score}(v^*) \leq \min_{v\in \mathcal{S}}score(v)$ holds, we can safely prune the remaining vertices in $\mathcal{L}$ and early terminates (lines 8-9); otherwise, it needs to invoke Algorithm~\ref{algo:comp-score} to compute structural diversity $score(v^*)$ (line 10) and checks whether $v^*$ should be added into the answer set $\mathcal{S}$ (lines 11-14). Finally, it outputs the top-$r$ results $\mathcal{S}$ and their social contexts $\context(v)$ for $v\in \mathcal{S}$ (line 15).

\begin{example}
We apply Algorithm~\ref{algo:bound-search} on graph $G$ in Figure~\ref{fig.intro_example}. Assume that $k=4$ and $r=1$. $\mathcal{L}$ ranks all vertices in the decreasing order of their upper bounds. At the first iteration, the vertex $v$ in $G$ 
has the highest upper bound $\overline{score}(v)=3$ of $\mathcal{L}$. It then computes $score(v)=3$ and adds $v$ into the answer set $\mathcal{S}$. At the next iteration, the highest upper bound of vertices in $\mathcal{L}$ is 1 (e.g., $\overline{score}(x_1)=1$), which triggers the early termination (lines 8-9 of Algorithm~\ref{algo:bound-search}). That is, $|\mathcal{S}|=1$ and $\overline{score}(v^*) = 1 \leq \min_{v\in \mathcal{S}}score(v) =3$. The algorithm terminates with an  answer $\mathcal{S}=\{v\}$. 
During the whole computing process, it  invokes Algorithm~\ref{algo:comp-score} only once  for  structural diversity calculation, which is much less than 17 times by the online search algorithm in Algorithm~\ref{algo:baseline-alg}. It demonstrates the pruning power of top-$r$ search framework.
\end{example}

\jinbin{
\subsection{Complexity Analysis}
We analyze the complexity of Algorithm~\ref{algo:bound-search}. Let the reduced graph be $G' \subseteq G$. 
Let $\rho'$, $m'$, and $\mathcal{T}'$ are respectively the arboricity, the number of edges, and the number of triangles in $G'$. Obviously,   $\rho' \leq \rho$, $m' \leq m$, and $\mathcal{T}' \leq \mathcal{T}$.

First, graph sparsification takes $O(\rho m)$ time by truss decomposition for graph $G$. 
Second, computing the upper bounds for all vertices takes $O(\rho' m')$ time on the reduced graph $G'$. In addition, $\mathcal{L}$ performs vertex sorting in the order of $\overline{score}(v^*)$ and maintains the list, which can be done in $O(n)$ time. In the worst case, Algorithm~\ref{algo:bound-search}  needs to compute $score(v)$ for every vertex $v$, which takes $O(\rho' (m'+\mathcal{T}'))$ by Theorem~\ref{theorem.final-baseline}. Overall, Algorithm~\ref{algo:bound-search}  takes $O(\rho' (m'+\mathcal{T}')+\rho m +n)$  $\subseteq O(\rho m+ \rho'\mathcal{T}')$ time and $O(m)$ space.

}

\section{A Novel Index-based Approach}\label{sec.indexbased}

Algorithm~\ref{algo:bound-search} is still not efficient for large networks, because the operation of computing $score(v)$ in Algorithm~\ref{algo:comp-score} applies truss decomposition on each \ego $G_{N(v)}$ from scratch in an online manner, which is highly expensive. It wastes lots of computations on the unnecessary access of disqualified edges whose trussness is less than $k$ in the \ego. To further speed up the calculation of $score(v)$, in this section, we develop a novel truss-based structural diversity index (\TSDindex). \TSDindex is a compact and elegant tree structure to keep the structural diversity information for all \egos in $G$. Based on  \TSDindex, we design a fast solution of computing $score(v)$ and propose an index-based top-$r$ search approach to quickly find $r$ vertices with the highest scores,
 which is particularly efficient to handle multiple queries with different $r$ and $k$ on the same graph $G$.

%
\subsection{TSD-Index Construction}

An intuitive indexing approach is to keep all maximal connected $k$-trusses in $G_{N(v)}$ by storing the trussness for all edges. However, it requires $O(\mathcal{T})$ space to store all \egos $G_{N(v)}$ for each vertex $v\in V$, which is inefficient for large networks. To develop efficient indexing scheme, we first start with the following observations.

\begin{figure}[h]
\centering \mbox{
\subfigure[$H_3$]{\includegraphics[width=0.45\linewidth]{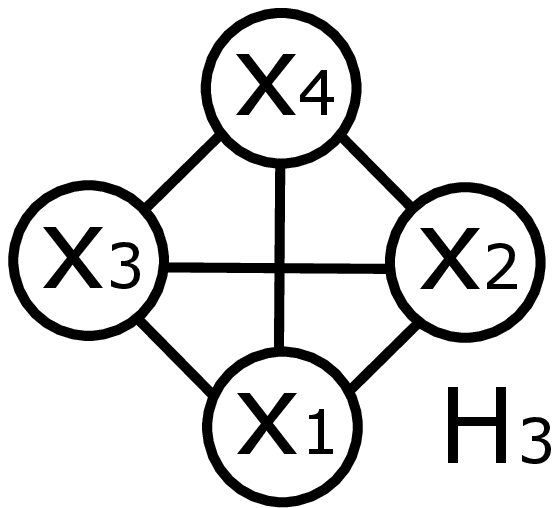}} \quad
\subfigure[Tree representation of $H_3$]{\includegraphics[width=0.45\linewidth]{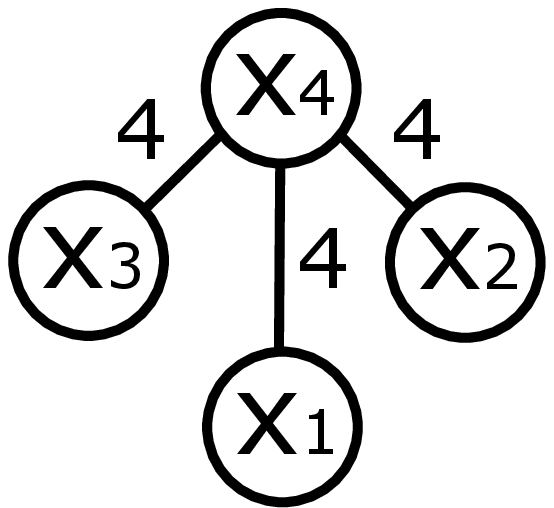}} }
\caption{An example of Observation~\ref{obs.tree}}
\label{fig.tsd_example}
\end{figure}

\begin{observation}
\label{obs.tree}
Figure~\ref{fig.tsd_example}(a) depicts a maximal connected 4-truss $H_3$ in the \ego $G_{N(v)}$ in Figure~\ref{fig.intro_example}(b). The definition of truss-based structural diversity only focuses on the number of maximal connected $k$-trusses, but ignore the connections between vertices in a maximal connected $k$-truss. It indicates that we do not need to store its whole structure. 
Figure~\ref{fig.tsd_example}(b) shows a tree-shaped structure with edge weights, which can clearly represent that $x_1,x_2,x_3,x_4$ are in the same maximal connected 4-truss. 
\end{observation}

\begin{figure}[h]
\centering \mbox{
\subfigure[$H_1$]{\includegraphics[width=0.48\linewidth]{figure/def_truss.eps}} \quad
\subfigure[Inaccurate tree 
]{\includegraphics[width=0.48\linewidth]{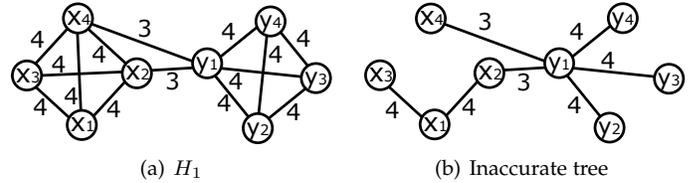}} }
\caption{An example of Observation~\ref{obs.forest}}
\label{fig.tsd_why_example}
\end{figure}

\begin{observation}
\label{obs.forest}
Figure~\ref{fig.tsd_why_example}(a) depicts a maximal connected 3-truss $H_1$ in the \ego $G_{N(v)}$ in Figure~\ref{fig.intro_example}(b). A tree structure is enough to represent the connectivity of vertices. However, if we keep an arbitrary tree structure of $H_1$ to connect all vertices, information loss of maximal connected $k$-trusses may happen. 
Consider the tree in Figure~\ref{fig.tsd_why_example}(b), for vertex $x_4$, it has no edges connecting with $x_1$, $x_2$ and $x_3$, but one incident edge with a weight of 3.  From this tree structure in Figure~\ref{fig.tsd_why_example}(b), we cannot infer that $x_4$ is involved in a maximal connected 4-truss $H_3$ shown in Figure~\ref{fig.tsd_example}(a).
\end{observation}

In summary, Observation~\ref{obs.tree} shows that the tree-shaped structure is enough to represent the identity of a maximal connected $k$-truss. Observation~\ref{obs.forest} further shows that the tree-shaped structure should have the maximum edge trussnesses to ensure no loss information of structural diversity, indicating a maximum spanning forest of $G_{N(v)}$ with the largest total weights of edge trussness.

\begin{algorithm}[t]
\small
\caption{TSD-Index Construction} \label{algo:tsd-con}
\begin{flushleft} 
\textbf{Input:} $G=(V, E)$\\
\textbf{Output:}  \TSDindex of $G$\\
\end{flushleft}
\vspace*{-0.3cm}
\
\begin{algorithmic}[1]

\STATE \textbf{for} $v\in V$ \textbf{do}

\STATE \hspace{0.3cm} Apply the truss decomposition in Algorithm~\ref{algo:truss-decomp} on $G_{N(v)}$;

\STATE \hspace{0.3cm} Construct a weighted graph $WG_v$ for $G_{N(v)}$, where each edge $e$ in $WG_v$ has a weight $w(e)=\tau_{G_{N(v)}}(e)$;

\STATE \hspace{0.3cm} Let a forest $\TSD_v$  formed by all isolated vertices $N(v)$;

\STATE \hspace{0.3cm} Let an edge set $\mathcal{L} \leftarrow E(WG_v)$; 

\STATE \hspace{0.3cm} \textbf{while} ($\mathcal{L} \neq \emptyset$)

\STATE \hspace{0.6cm} Let $e=(u, w) \in \mathcal{L}$ has the largest weight $w(e)$ in $\mathcal{L}$;

\STATE \hspace{0.6cm} \textbf{if} vertices $u$ and $w$ are disconnected in $\TSD_v$ \textbf{then}

\STATE \hspace{0.6cm} \hspace{0.3cm} Add a new edge $e$ with its weight $w(e)$ into $\TSD_v$;

\STATE \hspace{0.6cm} Delete $e$ from $\mathcal{L}$; 

\STATE return $\{\TSD_v|v\in V\}$;

\end{algorithmic}
\end{algorithm}

\stitle{TSD-Index Structure.} Based on the above observations, we are able to design our index structure of \TSDindex. We first define a weighted graph $WG_{v}$ for a vertex $v\in V$. $WG_{v}$ has the same vertex set and edge set with  $G_{N(v)}$ and $\forall e\in E(WG_{v})$ has a weight $w(e)=\tau_{G_{N(v)}}(e)$.  In other words, we assign a weight on each edge with its trussness on \ego $G_{N(v)}$ to form $WG_{v}$. As a result, the \TSDindex of $G_{N(v)}$ is defined as the maximum spanning forest of $WG_{v}$, denoted by $\TSD_v$.

\stitle{TSD-Index Construction. }
Algorithm~\ref{algo:tsd-con} describes a method of \TSDindex construction on graph $G$. The algorithm constructs the \TSDindex for each vertex $v\in G$ (lines 1-10). It first performs truss decomposition on $G_{N(v)}$ to obtain all edge trussnesses (line 2). The algorithm then constructs a weighted graph $WG_v$ for $G_{N(v)}$ where each edge $e$ has a weight $w(e)=\tau_{G_{N(v)}}(e)$ (line 3). Let 
 $\TSD_v$ be initially as all isolated vertices $N(v)$ (line 4). Then, we construct the maximum spanning forest of $WG_w$ by adding edges in the decreasing order of edge weights one by one into  $\TSD_v$ (lines 5-10). Let $\mathcal{L}$ be the edge set of $WG_v$ $E(WG_v)$. We visit each edge $e=(u,w)$ in the decreasing order of weight $w(e)$ in $\mathcal{L}$, and check whether $u, w$ are in the same component in $\TSD_v$. If $u, w$ are disconnected, we add an edge connecting $u$ and $w$ in $\TSD_v$. The process of constructing $\TSD_v$ breaks when all edges have been visited in $\mathcal{L}$ (lines 6-10).  Algorithm~\ref{algo:tsd-con} returns the \TSDindex of $G$ as $\{\TSD_v|v\in V\}$.

\begin{figure}[t]
\centering \mbox{
\subfigure[Step-1: Initialization with $N(v)$.]{\includegraphics[width=0.30\linewidth]{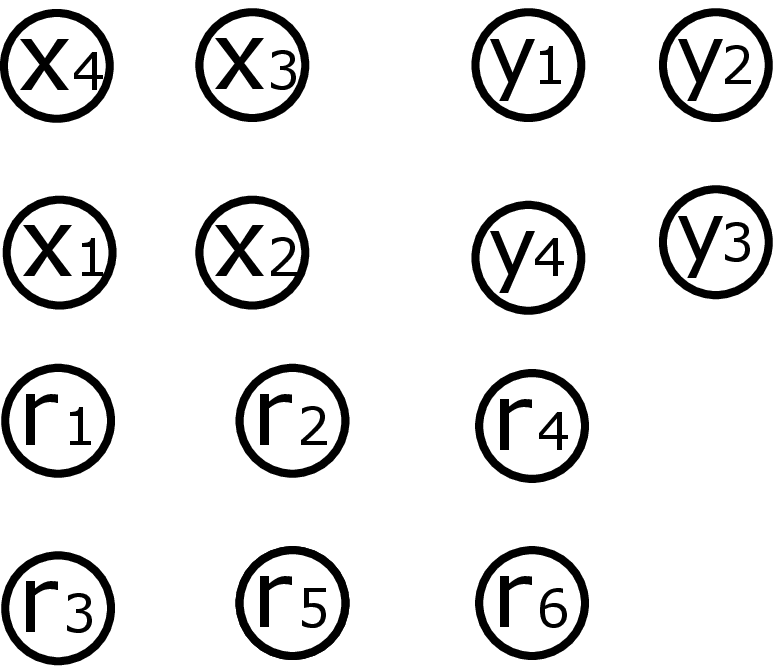}} \quad
\subfigure[Step-2: adding 4-truss edges.]{\includegraphics[width=0.30\linewidth]{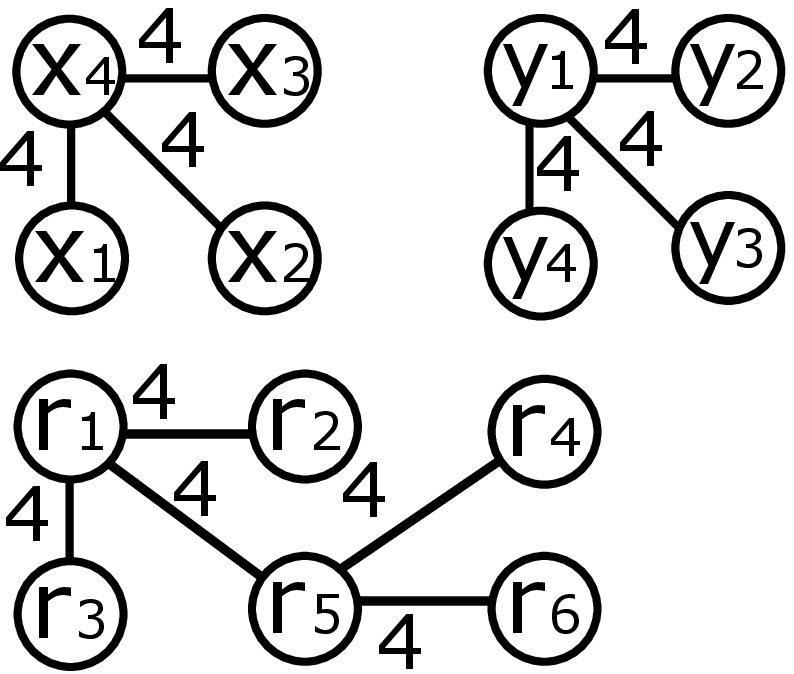}} \quad
\subfigure[Step-3: adding 3-truss edges. ]{\includegraphics[width=0.30\linewidth]{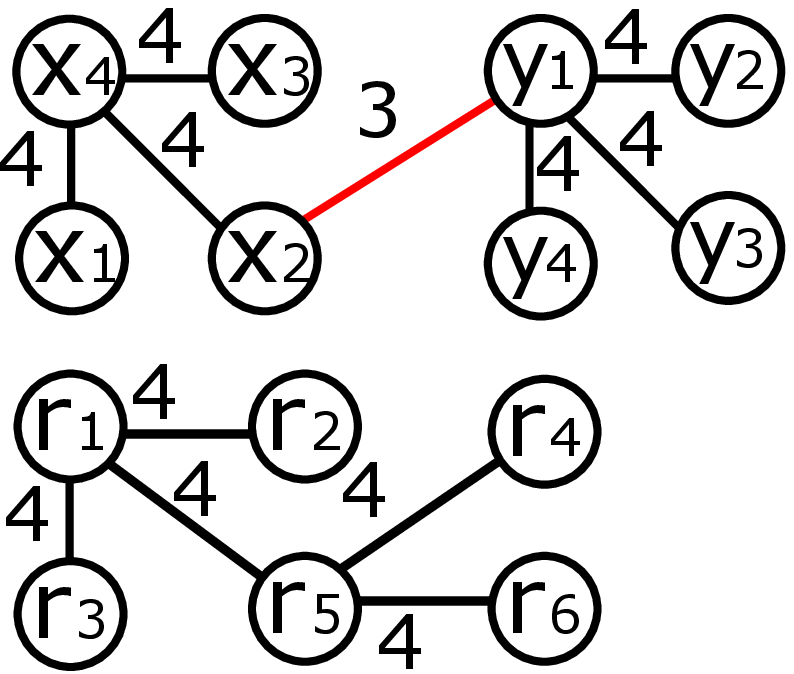}}
 }
\caption{Illustration of TSD-Index construction of $\TSD_v$}
\label{fig.tsd_index_con}
\end{figure}

\begin{example}

Figure~\ref{fig.tsd_index_con} illustrates the TSD-Index construction of $\TSD_v$ for a vertex $v$ in graph $G$ in Figure~\ref{fig.intro_example}.  Figure~\ref{fig.tsd_index_con} (a) shows that $\TSD_v$ is initialized to be a set of isolated nodes $N(v)$. Then, it checks all 4-truss edges and add qualified edges one by one into $\TSD_v$. According to Observation~\ref{obs.tree}, when Algorithm~\ref{algo:tsd-con} processes the edge $(x_3,x_1)$, it finds that $x_3$ and $x_1$ are in the same component in Figure~\ref{fig.tsd_index_con}(a), thus $(x_3,x_1)$ is not added to $\TSD_v$ in Figure~\ref{fig.tsd_index_con} (b). Afterwards, it adds the edge $e=(x_2, y_1)$ with weight $w(e)=3$ into $\TSD_v$ in Figure~\ref{fig.tsd_index_con}(c). \eat{Since there exist no 2-truss edges in $G_{N(v)}$, the algorithm stops by adding no edges with weights of 2.} The complete structure of $\TSD_v$ is finally depicted in Figure~\ref{fig.tsd_index_con}(c).
\end{example}

\begin{algorithm}[t]
\small
\caption{Computing $score(v)$ based on \TSDindex} \label{algo:tsd-query}
\begin{flushleft} 
\textbf{Input:} $G=(V, E)$, a vertex $v$, the trussness threshold $k$\\
\textbf{Output:} $score(v)$\\
\end{flushleft}
\vspace*{-0.3cm}
\
\begin{algorithmic}[1]

\STATE Let $H$ be a subgraph of $\TSD_v$ formed by all edges $e$ with $w(e)\geq k$;

\STATE $\context(v) \leftarrow \emptyset$; 

\STATE \textbf{for} each unvisited vertex $u \in V(H)$ do

\STATE \hspace{0.3cm} Traverse the component $X$ containing $u$ in $H$; 

\STATE \hspace{0.3cm} Let a social context $S \leftarrow$ the set of vertices in $X$;

\STATE \hspace{0.3cm} $\context(v) \leftarrow \context(v) \cup \{S\} $;

\STATE $\score(v) \leftarrow |\context(v)|$;

\STATE \textbf{return} $\score(v)$;
\end{algorithmic}
\end{algorithm}

\stitle{Remarks}. 
Note that our \TSDindex can answer queries of any $k$ and $r$. It is independent to parameters $k$ and $r$ once the \TSDindex is constructed. \TSDindex can not only be used for calculating the structural diversity scores, but also support the retrieval of all social contexts in \egos. Early pruning (Property~\ref{lemma.trussfilter} and Lemma~\ref{lemma.trussbound}) works for the online search algorithms, but not for \TSDindex construction in Algorithm~\ref{algo:tsd-con}.

\subsection{TSD-Index-based Top-$r$ Search}
In the following, we first propose an efficient algorithm for computing structural diversity scores using the \TSDindex. 
Based on it, we develop our \TSDindex-based top-$r$ search algorithm.

\stitle{Computing $score(v)$ based on TSD-Index. } Algorithm~\ref{algo:tsd-query} presents a method of computing $score(v)$ based on the \TSDindex. The algorithm first retrieves a subgraph $H$ of $\TSD_{v}$ formed by all edges $e$ with the edge weight $w(e)\geq k$ (line 1). Next, it finds all maximal connected $k$-trusses of $H$ that are the social contexts $\context(v)$ (lines 2-6). Applying the breadth-first-search strategy, it uses one hashtable to ensure each vertex to be visited once, and one queue  to visit the vertices of a neighborhood social context $S$ one by one (lines 3-6). After traversing each component in $H$, it keeps the social context $\context(v)$ by  the union of $S$ (line 6). Finally, it returns $score(v)$ as the multiplicity of social contexts $\context(v)$ (lines 7-8).

\stitle{\TSDindex-based Top-$r$ Search Algorithm.} Based on the $\TSD_v$, we design a new upper bound of $score(v)$ for pruning. The upper bound of $score(v)$ is defined as $\widetilde{score}(v) $ $ = \frac{|\{e\in \TSD_v: w(e)\geq k\}|}{k-1}$. The essence of $\widetilde{score}(v)$ holds because a maximal connected $k$-truss should have a tree-shaped representation of at least $(k-1)$ edges with weights of no less than $k$ in $\TSD_v$.
We can make a fast calculation of $\widetilde{score}(v)$ by sorting all edges of $\TSD_v$  in the decreasing order of edge weights, during the index construction.
Equipped with Algorithm~\ref{algo:tsd-query} of computing $score(v)$ and a new upper bound $\widetilde{score}(v)$, our \TSDindex-based top-$r$ structural diversity search algorithm invokes an efficient framework similarly as Algorithm~\ref{algo:bound-search}, which finds the top-$r$ answers by pruning those vertices $v$ that has the upper bound  $\widetilde{score}(v)$ no greater than the top-$r$ answer $\mathcal{S}$.

\subsection{Complexity Analysis}

\begin{theorem}\label{theorem.construction}
Algorithm~\ref{algo:tsd-con} constructs \TSDindex for a graph $G$ in $O(\rho(m+\mathcal{T}))$ time and $O(m)$ space. The index size is $O(m)$. Moreover,  \TSDindex-based  search approach tackles the problem of truss-based structural diversity search in $O(m)$ time and $O(m)$ space.
\end{theorem}

\eat{
\begin{proof} Please refer to Appendix D for the proof.
\end{proof}
}

\jinbin{
\begin{proof}
First, we analyze the time complexity of \TSD construction. For each vertex $v\in V$, Algorithm~\ref{algo:tsd-con} extracts $G_{N(v)}$ and applies truss decomposition on $G_{N(v)}$. This totally takes $O(\rho (m + \mathcal{T}))$ by Theorem~\ref{theorem.final-baseline}. In addition, for $v \in V$, a weighted graph $WG_{v}$ has $n_v$ vertices and $m_v$ edges. The sorting of weighted edges can be done in $O(m_v)$ time using a bin sort. Thus, applying Kruskal's algorithm \cite{cormen2009introduction} to find the maximum spanning forest from $WG_v$ takes $O(m_v)$ time. As a result, constructing  the \TSDindex for all vertices takes $O(\sum_{v\in V} m_v) \subseteq O(\mathcal{T})$. Therefore, the time complexity of Algorithm~\ref{algo:tsd-con} is $O(\rho(m+\mathcal{T}))$ in total.

Second, we analyze the space complexity of \TSD construction. The edge set $\mathcal{L}$ takes $O(m_v) \subseteq O(m)$ space. The index $\TSD_v$ takes $O(n_v)\subseteq O(n)$ space. The space complexity of Algorithm~\ref{algo:tsd-con} is  $O(m+n) \subseteq O(m)$.  

Third, we analyze the index size of \TSDindex of $G$.  For a vertex $v$, $\TSD_v$ is the maximum spanning forest of $WG_{v}$, which has no greater than $n_v-1$ edges. Thus, the size of $\TSD_v$ is $O(n_v)$. Overall, the index size of \TSDindex of $G$ is $O(\sum_{v\in V} n_v ) \subseteq O(m)$. 

Finally, we analyze the time and space complexity of \TSDindex-based  search approach. First, Algorithm~\ref{algo:tsd-query} takes $O(|N(v)|)$ time to compute $score(v)$ for a vertex $v\in V$. In the  worst case, the \TSDindex-based search approach needs to invoke Algorithm~\ref{algo:tsd-query} to compute $score(v)$ for all vertices. It takes $O(\sum_{v\in V} |N(v)|)$ $\subseteq O(m)$ time complexity. In addition, the upper bound $\widetilde{score}(v)$ takes $O(1)$ space for each vertex $v\in V$. Thus, the space complexity is \eat{still} $O(m)$.
\end{proof}
}

\jinbin{
\stitle{Remarks}. In summary, the \TSDindex-based search approach is clearly faster than the online search algorithms in Algorithm~\ref{algo:baseline-alg} and Algorithm~\ref{algo:bound-search}, in terms of their time complexities. 
In addition, \TSDindex can support efficient updates in dynamic graphs where the graph structure undergo frequently updates with nodes/edges insertions/deletions. Although an edge insertion may cause the structure change of many \egos, the updating techniques are still promising to be further developed with some carefully designed ideas, given by the existing theory and algorithms of $k$-truss updating on dynamic graphs \cite{YZhang12,huang2014querying}.   
}

\section{A Global Information Based Approach}\label{sec.bitmap}
In this section, we propose a new approach \GCT for truss-based structural diversity search, which utilizes the \underline{g}lobal triangle information for efficient ego-network truss decomposition and develops a \underline{c}ompressed \underline{t}russ-based diversity \ADVindex to improve \TSDindex. 

\subsection{Solution Overview}
We briefly introduce a solution overview of \GCT algorithm, which leverages one-shot global triangle listing and a compressed  \ADVindex for fast structural diversity search computation. The method of \ADVindex construction is outlined in Algorithm~\ref{algo:adv-construction}. \ADVindex equips with three new techniques and implementations: 1) fast ego-network extraction (lines 1-4 of Algorithm~\ref{algo:adv-construction}); 2) bitmap-based truss decomposition (lines 5-14 of Algorithm~\ref{algo:adv-construction}); and 3) \ADVindex construction for an ego-network (line 15 of Algorithm~\ref{algo:adv-construction}), which is detailed presented in Algorithm~\ref{algo:adi-con}. 

Note that there is non-trivial challenging to explore the sharing computation across vertices using global truss decomposition. We analyze the structural properties of truss-based social contexts in Section~\ref{sec.upperbound}.  
Unfortunately, Observation~\ref{lemma.symmetricity} shows that it cannot share the symmetry triangle-based structure in the ego-networks across different vertices, even two close neighbors $u$ and $v$. 
Thus, our truss-based model fails to enjoy the symmetry properties (e.g., edge supports and trussnesses) of \egos for fast structural diversity score computation 
as \cite{huang2013top}. On the other hand, we observe that the one-shot triangle listing of global truss decomposition can help to efficiently extract ego-networks for all vertices. Moreover, we realize that the bitwise operations can further improve the efficiency of  truss decomposition in such local ego-networks. In addition, we propose a compact index structure of \ADVindex, which maintains only supernodes and superedges to discard the edges within the same $k$-level of social contexts. \ADVindex based query processing can be done more efficient than the \TSDindex-based approach.

\subsection{Fast Ego-network Truss Decomposition}

In this section, we propose a fast method of \ego truss decomposition, which leverages on the global triangle listing and bitmap-based truss decomposition. 

\stitle{Global Triangle Listing based Ego-network Extraction}. 
Ego-network extraction is the first key step of score computation in Algorithm~\ref{algo:comp-score} and \TSDindex construction in Algorithm~\ref{algo:tsd-con}. However, it suffers from heavily duplicate triangle listing. Specifically, for each vertex $v$, it needs to perform a triangle listing to find all triangles $\triangle_{vuw}$ and generate an edge $(u,w)$ in \ego $G_{N(v)}$. $\triangle_{vuw}$ is generated twice, which checks the common neighbors of $N(v)\cap N(u)$ and $N(v)\cap N(w)$ for two edges $(v, u)$ and $(v, w)$ respectively. Similarly, for vertices $u$ and $w$, $\triangle_{vuw}$ is generated twice respectively for extracting \egos $G_{N(u)}$ and $G_{N(w)}$. Unfortunately, $\triangle_{vuw}$ is repeatedly enumerated for six times, which is inefficient for local \ego extraction. 

To this end, we propose to utilize global triangle listing once to generate all the \egos in $G$. The details of fast \ego extraction is presented in Algorithm~\ref{algo:adv-construction} (lines 1-4). Specifically, for each edge $e=(u, v) \in E$, it identifies triangle $\triangle_{vuw}$ by enumerating all the common neighbors $w\in N(u)\cap N(v)$, and adds edge $e$ into \ego $G_{N(w)}$ (lines 2-4). Thus, it finishes the construction for all \egos, which can be directly used in the following \ego truss decomposition. Each triangle $\triangle_{vuw}$ is enumerated for three times, which saves a half of original computations using six enumeration times.  
Overall, our method of fast \ego extraction makes use of global triangle listing for best sharing in local \ego computations.

\begin{algorithm}[t]
\small
\caption{\ADVindex Construction} \label{algo:adv-construction}
\begin{flushleft} 
\textbf{Input:} Graph $G$\\
\textbf{Output:} \ADVindex of all vertices\\
\end{flushleft}
\vspace*{-0.3cm}
\
\begin{algorithmic}[1]

\STATE Let be $G_{N(v)}$ as an empty graph for each $v\in V$;

\STATE \textbf{for} each edge $e=(u,v) \in E$ \textbf{do}

\STATE \hspace{0.3cm} \textbf{for} each vertex $w \in N(u)\cap N(v)$ \textbf{do}

\STATE \hspace{0.3cm} \hspace{0.3cm} Add the new edge $e$ into $G_{N(w)}$;

\STATE \textbf{for} each vertex $v$ in $G$ \textbf{do}

\STATE \hspace{0.3cm}  Retrieve an ego-network $G_{N(v)}$ directly based on \textbf{Steps 2-4}, which avoids the duplicate triangle listing;

\STATE \hspace{0.3cm}  Give IDs to all vertices in $G_{N(v)}$ sequentially from 1 to $L$, where $L= |N(v)|$.

\STATE \hspace{0.3cm} \textbf{for} each vertex $u \in N(v)$ \textbf{do}

\STATE \hspace{0.3cm} \hspace{0.3cm} Create a bitmap $\bitmap_u$ of all 0 bits with $|\bitmap_u| = L$.

\STATE \hspace{0.3cm} \hspace{0.3cm} \textbf{for} each vertex $w \in N_{G_{N(v)}}(u)$ \textbf{do}

\STATE \hspace{0.3cm} \hspace{0.3cm} \hspace{0.3cm} $\bitmap_u[w] \leftarrow 1$;

\STATE \hspace{0.3cm} \textbf{for} each edge $e = (u, w)\in E(G_{N(v)})$ \textbf{do}


\STATE \hspace{0.3cm} \hspace{0.3cm} $\sup_{G_{N(v)}}(e) \leftarrow  \bitmap_x$   \verb AND  $\bitmap_y $;

\STATE \hspace{0.3cm} Apply a bitmap-based peeling process for truss decomposition \cite{WangC12} on $G_{N(v)}$;

\STATE \hspace{0.3cm}  Apply \ADVindex construction in Algorithm~\ref{algo:adi-con} on $G_{N(v)}$ to obtain $\ADV_v$;

\STATE \textbf{return} the \ADVindex $\{\ADV_v: v\in V\}$;

\end{algorithmic}
\end{algorithm}{}








\stitle{Bitmap-based Truss Decomposition}. 
We propose a bitmap-based approach to accelerate the truss decomposition. To apply truss decomposition on an obtained \ego $G_{N(v)}$, an important step is support computation, i.e., calculating $\sup_{G_{N(v)}}(e)$ as the number of triangles containing $e = (x, y)$ for each edge $e\in E(G_{N(v)})$. The existing method of computing $\sup_{G_{N(v)}}(e)$~\cite{WangC12}  uses the triangle listing, which checks each neighbor $z\in N(x)$ in \ego $G_{N(v)}$ to see whether $z \in N(y)$ using hashing technique. The hash checking takes constant time $O(1)$ in theoretical analysis, but in practice costs an expensive time overhead of support computation appeared in large graphs for frequent hash updates and checks.  
To this end, we propose to use a bitmap technique to accelerate the support computation. 
Firstly, we give a order ID to every vertex in $G_{N(v)}$ sequentially from 1 to $L$, where $L= |N(v)|$. For each vertex $x\in N(v)$, we create a binary bitmap $\bitmap_x$ with all 0 bits. For each edge $e= (x, y) \in E(G_{N(v)})$, we set to 1 for both the $x$-th bit of bitmap $\bitmap_y$ and the $y$-th bit of bitmap $\bitmap_x$, indicating $x\in N_{G_{N(v)}}(y)$ and $y\in N_{G_{N(v)}}(x)$. Then, the support of $\sup(e)$ equals to the number of 1 bits commonly appeared in $\bitmap_x$ and $\bitmap_y$, denoted by $\sup_{G_{N(v)}}(e) = |N(u) \cap N(v)|= \bitmap_x $  \verb AND  $\bitmap_y $. Note that the binary operation of bitwise \verb AND $ $ can be done efficiently.

Algorithm~\ref{algo:adv-construction} presents the detailed procedure of bitmap-based truss decomposition (lines 5-15). The algorithm first retrieve \ego $G_{N(v)}$ directly from the global triangle listing (line 6). It then initializes the $\bitmap_x$ for all vertices $x\in N(v)$ and calculates the support  $\sup_{G_{N(v)}}(e)$ as $\bitmap_x$ \verb AND  $\bitmap_y $ for all edges $e\in E(G_{N(v)})$ (lines 8-13). Next, The algorithm applies a bitmap-based peeling process for truss decomposition \cite{WangC12} on $G_{N(v)}$. Specifically, when an edge $(x, y)$ is removed from a graph, it updates $\bitmap_x[y]=0$ and $\bitmap_y[x]=0$. Due to the limited space, we omit the details of similar bitmap-based peeling process (line 14).  After obtaining all the edge trussnesses, we invoke Algorithm~\ref{algo:adi-con} (to be introduced in Section~\ref{sec.advindex}) to construct \ADVindex (line 15).

\subsection{GCT-index Construction and Query Processing} \label{sec.advindex}
In this section, we propose a new data structure of \ADVindex, which compresses the structure of \TSDindex in a more compact way. 


We start with discussing the limitations of \TSDindex. Each social context is defined as a maximal connected $k$-truss.
The spanning forest structure of \TSDindex stores not only the edge connections between different social contexts, but also the internal edges within a social context. However, such information of internal edges is redundant, which can be avoided for indexing. For example, consider the \TSDindex of vertex $v$ in Figure~\ref{fig.idx_comp}(a). The vertices $\{x_1, x_2, x_3, x_4\}$ forms a social context of maximal connected 4-truss. The edges $(x_4, x_1)$, $(x_4, x_2)$, and $(x_4, x_3)$ can be ignored for indexing storage. Instead, we keep a node list of $\{x_1, x_2, x_3, x_4\}$, which is enough to recover the information of social contexts by saving time-consuming cost of edge listing.


\stitle{GCT-index Structure}. 
\ADVindex keeps a maximum-weight forest-like structure similar as \TSDindex, which consists of supernodes and superedges. Specifically, for a vertex $v$, the \ADVindex of $v$ is denoted by $\ADV_v=(\mathcal{V}_v,\mathcal{E}_v)$, where $\mathcal{V}_v \subseteq N(v)$ and $\mathcal{E}_v$ are the set of supernodes and superedges respectively. A supernode $S \in \mathcal{V}_v$ represents a group of vertices that are connected via the edges of the same trussness $\tau(S_u)$ in a social context. Each supernode is associated with two features, including the trussness of connecting edges $\tau(S_u)$ and the vertex list $V_S$ of vertices belonging to this social context. Based on the isolated supernodes of $ \mathcal{V}_v$, we add the superedges $\mathcal{E}_v = \{(S_i, S_j): S_i, S_j \in \mathcal{V}_v \text{ and } \exists v_i \in V_{S_i}, v_j\in V_{S_j} \text{ such that the edge } (v_i, v_j)\in E \}$ into $\ADV_v$, such that all vertices forms a forest with the largest weight. Note that the weight of a superedge $(S_i, S_j)\in \mathcal{E}_v$ is denoted by the corresponding edge trussness in $G_{N(v)}$, i.e., $w((S_i, S_j))= \max_{v_i\in V_{S_i}, v_j\in V_{S_j} } \tau_{G_{N(v)}}(v_i, v_j)$. For example, for a vertex $v$, the corresponding \TSDindex  in Figure~\ref{fig.idx_comp}(a) is compressed into a small \ADVindex $\ADV_v$ as shown in Figure~\ref{fig.idx_comp}(b). $\ADV_v=(\mathcal{V}_v, \mathcal{E}_v)$ where $\mathcal{V}_v =\{S_1, S_2, S_3\}$ and  $\mathcal{E}_v = \{(S_1, S_2)\}$. The supernode $S_1$ consists of $\tau(S_1)=4$ and $V_{S_1}=\{x_1, x_2, x_3, x_4\}$ that belong to 4-truss social context. The superedge $(S_1, S_2)$ has a weight of $w((S_1, S_2))=3$, due to $\tau_{G_{N(v)}}(x_2, y_1)= 3$. This edge indicates that the vertices in $S_1$ and $S_3$ belong to the same 3-truss social context, i.e.,  $V_{S_1} \cup V_{S_2}= \{x_1, x_2, x_3, x_4, y_1, y_2, y_3, y_4\}$.

\begin{figure}[t]
\vspace*{-0.3cm}
\centering \mbox{
\subfigure[\TSDindex $\TSD_v$]{\includegraphics[width=0.4\linewidth]{figure/tsd_index_b.eps}} \quad
\subfigure[\ADVindex $\ADV_v$]{\includegraphics[width=0.45\linewidth]{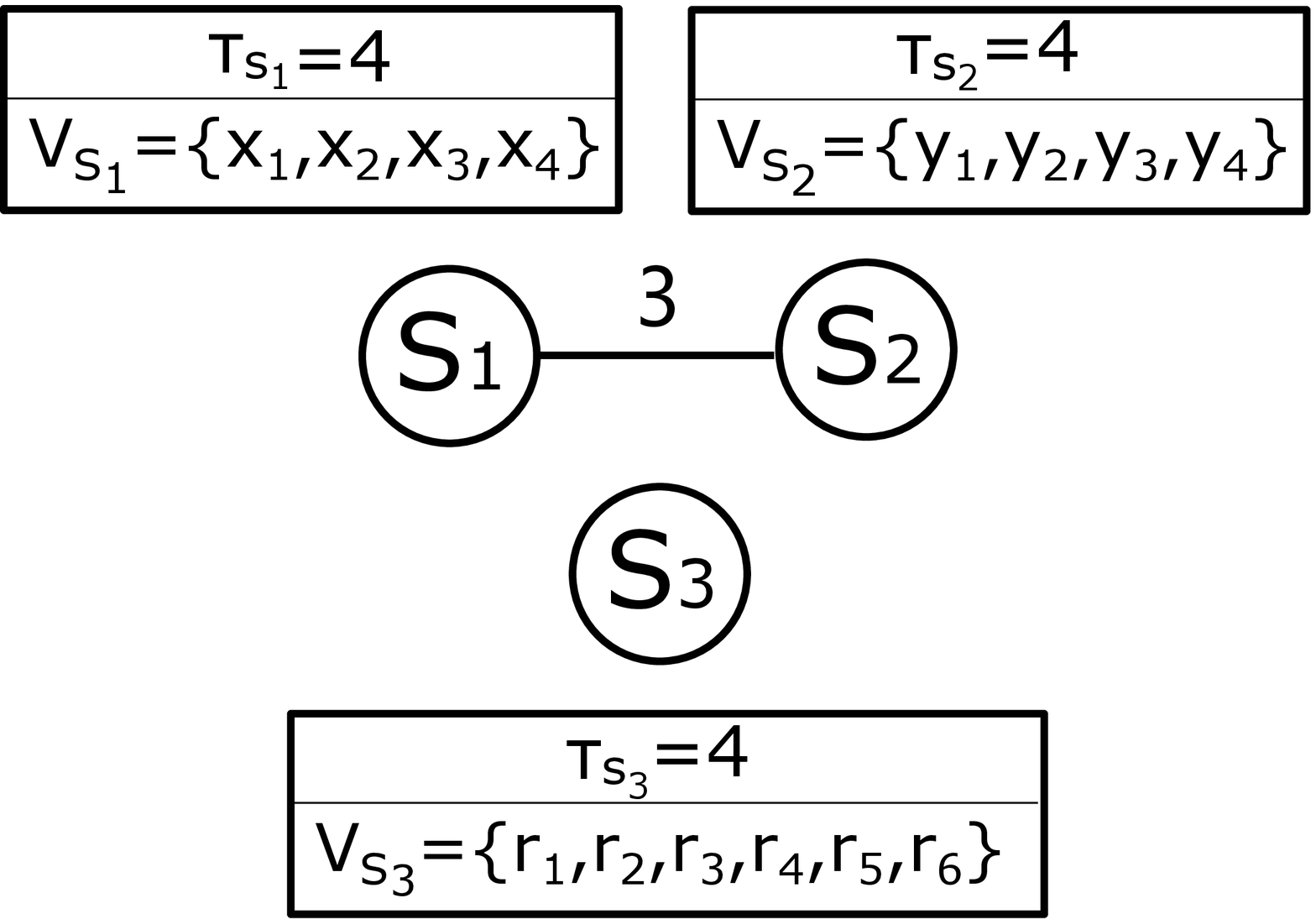}}
 }
 \vspace*{-0.4cm}
\caption{$\ADV_v$ is a compressed data structure of $\TSD_v$ for vertex $v$ in graph $G$ as shown in Figure~\ref{fig.intro_example}(a).}
\vspace*{-0.2cm}
\label{fig.idx_comp}
\end{figure}

\stitle{GCT-index Construction}.  Algorithm \ref{algo:adi-con} presents the procedures of constructing \ADVindex in an \ego $G_{N(v)}$ for a vertex $v$. The algorithm first creates the supernodes $S_u$ for each vertex $u$ in \ego $G_{N(v)}$ (lines 2-4). For each supernode $S_u$, the trusssness $\tau(S_u)$ is initialized as the vertex trussness of $\tau_{G_{N(v)}}(u)$ and $V_{S_u}=\{u\}$ (line 3).  Next, the algorithm continues to construct \ADVindex by adding superedges and merging supernodes, via a traverse of the whole set of edges $L=E(G_{N(v)})$ (lines 5-15). In each iteration, it retrieves an edge $e=(u, w)$ with the largest trussness in $L$ (lines 7). If two vertices $u$ and $w$ belong to the same supernode or their supernodes $S_u$ and $S_w$ are already connected in $\ADV_v$, then it continues to check the next edge in $L$ (lines 8-9). If two different supernodes  $S_u$ and $S_w$ have the same trussnesses as $\tau_{G_{N(v)}}(e)$, it merge two supernodes into one by assigning all $S_w$'s feature to $S_u$. Specifically, it unions two vertex lists as $V_{S_u} = V_{S_u} \cup V_{S_w}$ and assign $S_u$ the edges that are incident to supernode $S_w$, and then remove $S_w$ from  $\mathcal{V}_v$ (lines 10-12); Otherwise, it adds a superedge between $S_u$ and $S_w$ and assigns the edge weight as  $w((S_u, S_w)) = \tau_{G_{N(v)}}(e)$ (line 14-15).  After processing all edges in $L$, the algorithm finally returns the \ADVindex as $\ADV_v=(\mathcal{V}_v, \mathcal{E}_v)$ (line 16). 


\begin{algorithm}[t]
\small
\caption{\ADVindex Construction for an Ego-network} \label{algo:adi-con}
\begin{flushleft} 
\textbf{Input:} an \ego $G_{N(v)}$ for a vertex $v$\\
\textbf{Output:}  \ADVindex of $v$\\
\end{flushleft}
\vspace*{-0.3cm}
\
\begin{algorithmic}[1]

\STATE $\mathcal{V}_v \leftarrow \emptyset$; $\mathcal{E}_v \leftarrow \emptyset$; 

\STATE \textbf{for} each vertex $u \in N(v)$ \textbf{do}

\STATE \hspace{0.3cm} Super-node $S_u$: $\tau(S_u) = \tau_{G_{N(v)}}(u)$ and $V_{S_u}=\{u\}$;

\STATE \hspace{0.3cm} $\mathcal{V}_v \leftarrow \mathcal{V}_v \cup \{S_u\}$;

\STATE Let an edge set $\mathcal{L} \leftarrow E(G_{N(v)})$; 

\STATE \textbf{while} $\mathcal{L} \neq \emptyset$ \textbf{do}

\STATE \hspace{0.3cm} Pop out an edge $e=(u, w) \in \mathcal{L}$ with the largest trussness $\tau_{G_{N(v)}}(e)$ from $\mathcal{L}$;

\STATE \hspace{0.3cm} Identify the corresponding supernodes $S_u$ and $S_w$ for $u$ and $w$ respectively.

\STATE \hspace{0.3cm} \textbf{if} $S_u = S_w$ \textbf{or} $S_u$ and $S_w$ are connected \textbf{then} \textbf{continue};

\STATE \hspace{0.3cm} \textbf{if} $\tau(S_u)=\tau(S_w)=\tau_{G_{N(v)}}(e)$ \textbf{then}

\STATE \hspace{0.3cm} \hspace{0.3cm} Two supernodes merge: $V_{S_u} \leftarrow V_{S_u} \cup V_{S_w}$;

\STATE \hspace{0.3cm} \hspace{0.3cm} Assign all $S_w$'s incident edges to $S_u$ and delete $S_w$;

\STATE \hspace{0.3cm} \textbf{else}

\STATE \hspace{0.3cm} \hspace{0.3cm} Superedge insertion: $\mathcal{E}_v \leftarrow \mathcal{E}_v \cup (S_u, S_w)$;

\STATE \hspace{0.3cm} \hspace{0.3cm} $w((S_u, S_w)) \leftarrow \tau_{G_{N(v)}}(e)$;

\STATE \textbf{return} $\ADV_v=(\mathcal{V}_v, \mathcal{E}_v)$;

\end{algorithmic}
\end{algorithm}

\stitle{\ADVindex-based Query Processing}. Thanks to a very elegant and compact structure of \ADVindex, we next introduce a fast method to compute $score(v)$ for a given vertex $v$.


\begin{lemma}\label{lemma.adiscore}
For a vertex $v\in V$ and a number $k$,  the structural diversity score of $v$ is $score(v)= N_k -M_k$, where $N_k$ and $M_k$ are the number of supernodes and superedges with trussness no less than $k$ in $\ADV_v$, i.e.,  $N_k = |\{S\in \mathcal{V}_v: \tau(S)\geq k\}|$ and  $M_k = |\{e\in \mathcal{E}_v: \tau(e) \geq k\}|$.
\end{lemma}

\begin{proof} 
Let be $score(v)=x$ w.r.t. a particular $k$. This indicates that \ego $G_{N(v)}$ has $x$ social contexts. 
In terms of the structural properties of \ADVindex, each maximal connected $k$-truss is represented by a connected structure of spanning tree or just one single supernode. In the $i$-th spanning tree (or $i$-th single supernode), the number of supernodes is denoted as $n_i$, and the number of superedges is $n_i-1$. Thus, $N_k=\sum_{i=1}^{x}n_i$ and $M_k=\sum_{i=1}^{x} n_i-1$. As a result, $N_k-M_k=\sum_{i=1}^{x} n_i - \sum_{i=1}^{x}(n_i-1)=\sum_{i=1}^{x}1=x$.

Note that the \ADVindex-based query processing for structural diversity search takes $O(m)$ time in worst, where $m$ is the number of edges in $G$.







\end{proof}

\section{Experiments}\label{sec.exp}

In this section, we evaluate the effectiveness and efficiency of our proposed algorithms on real-world networks. 
All algorithms mentioned above are implemented in C++ and complied by gcc at -O3 optimization level. The experiments are run on a Linux computer with 2.2GHz quard-cores CPU and 32GB memory. \jbhuang{}

\stitle{Datasets:} We use eight datasets of real-world networks, and treat them as undirected graphs. Except for socfb-konect,\footnote{\scriptsize{\url{{http://networkrepository.com/socfb_konect.php}}}} all other datasets are available from the Stanford Network Analysis Project \cite{snapnets}. 
The network  statistics are described in Table~\ref{tab:dataset}. We report the node size $|V|$, the edge size $|E|$, the maximum degree $d_{max}$, the maximum edge trussness $\tau^*_{G}= \max_{e\in E} \tau_{G}(e)$,  the maximum edge trussness among all \egos $\tau^*_{ego}$ $=\max_{v\in V, e\in E(G_{N(v)})}$ $ \{\tau_{G_{N(v)}}(e) \}$, and the number of triangles $\mathcal{T}$.

\stitle{Compared Methods and Evaluated Metrics:} To evaluate the effectiveness of top-$r$ truss-based structural diversity model, we conduct the simulation of social influence process and report the number of affected vertices of the $r$ selected vertices by all methods. We test and compare our truss-based structural diversity method with three other methods as follows.

\squishlisttight
\item \random: is to select $r$ vertices from graph by random.
\item \component: is to select $r$ vertices with the highest $k$-sized component-based structural diversity \cite{chang2017scalable}. 
\item \core: is to select $r$ vertices with the highest $k$-core-based structural diversity \cite{XHuang15}. 
\item \trussdiv: is our method by selecting $r$ vertices with the highest $k$-truss-based structural diversity.
\end{list} 

In addition, to evaluate the efficiency of improved strategies, we compare our algorithms with two state-of-the-art methods \component \cite{chang2017scalable} and \core \cite{XHuang15}. Note that the implementation of \component in \cite{chang2017scalable} is much faster than the method in \cite{huang2013top}. We also test and compare four algorithms proposed in this paper as follows.

\squishlisttight
\item \baseline: is the simple approach to compute structural diversity for all vertices in Algorithm~\ref{algo:baseline-alg}.
\item \bound: is the efficient approach using graph sparsification and an upper bound for pruning vertices in Algorithm~\ref{algo:bound-search}.
\item \TSD: is the \TSDindex based approach, which uses Algorithm~\ref{algo:tsd-query} to compute structural diversity. 
\item \ADV: is the \ADVindex based  approach in Algorithm~\ref{algo:adv-construction}.
\end{list} 

We compare them by reporting the running time in seconds and the search space as the number of vertices whose structural diversities are computed in search process. The less running time and search space are, the better efficiency performance is. 

\stitle{Parameters: } We set the parameters $r=100$ and $k=3$ by default. We also evaluate the methods by varying the parameters $k$ in \{2, 3, 4, 5, 6\} and $r$ in $\{50,100,150,200,250,300\}$.

\begin{table}[t!]
\begin{center}
\scriptsize
\caption[]{\textbf{Network Statistics(K$=10^3$ and M$=10^6$)}}\label{tab:dataset}
\begin{tabular}{|c|c|c|c|c|c|c|c|}
\hline
Name & $|V|$ & $|E|$ & $d_{max}$ & $\tau^*_{G}$ & $\tau^*_{ego}$ & $\mathcal{T}$ \\
\hline \hline
Wiki-Vote	& 7\textbf{K}	 & 103\textbf{K} & 1,065 & 23 & 22 & 608,389 \\ \hline
Email-Enron	& 36\textbf{K} & 183\textbf{K} & 1,383 & 22 & 21 & 727,044 \\ \hline
Epinions & 75\textbf{K} & 508\textbf{K} & 3,044 & 33 & 32 & 1,624,481 \\ \hline
Gowalla	& 196\textbf{K} & 950\textbf{K} & 14,730 & 29 & 28 & 2,273,138 \\  \hline
NotreDame	& 325\textbf{K} & 1.4\textbf{M} & 10,721 & 155 & 154 & 8,910,005 \\ \hline
LiveJournal	& 4\textbf{M} & 34.7\textbf{M} & 14,815 & 352 & 351 & 177,820,130 \\ \hline
socfb-konect	& 59\textbf{M} & 92.5\textbf{M} & 4,960 & 7 & 6 & 6,378,280 \\ \hline
Orkut	& 3.1\textbf{M} & 117\textbf{M}  & 33,313 & 73 & 72 & 412,002,900 \\ \hline

\end{tabular}
\end{center}
\end{table}

\begin{table*}[t]
\begin{center}\small 
\caption[]{Comparison of running time (in seconds) and search space (the number of vertices whose structural diversity are computed) of different algorithms. Here $k=3$ and $r=100$. \textbf{}}\label{tab:querytime}
\begin{tabular}{|l|c|c|c|c|c|c|c|c|}
\hline \multirow{2}{*}{Network} & \multicolumn{4}{|c|}{Running Time} & \multicolumn{4}{|c|}{Search Space} \\
\cline{2-9} & \baseline & \bound & \TSD  & $R_t$ & \baseline & \bound & \TSD  & $R_s$ \\\hline
Wiki-Vote    & 10.7s   & 10.2s   & 7.0ms  & 1,529 & 8,297 & 2,704 & 2,628 &  3.1 \\\hline
Email-Enron  & 11.8s   & 11.3s   & 18.2ms & 648 & 36,692 & 4,284 & 4,274 &  8.6 \\\hline
Epinions    &  37.7s  & 34.2s   & 31.9ms & 1,182 & 75,887 & 6,810 & 6,531 &  11.6 \\\hline
Gowalla    & 52.2s   & 42.2s   & 70.2ms & 743 & 196,591 & 22,267 & 21,674 &  9.0 \\\hline
NotreDame    & 291s   & 283s   & 106ms & 2,745 & 325,729 & 24,285 & 24,188 & 13.4 \\\hline
LiveJournal    &  10,418s  &  9,456s  & 4.9s  & 2,126 & 4,036,537 & 208,722 & 182,646 & 22.1\\\hline
socfb-konect    &  1,591s  &   15.3s & 6s & 265 & 59,216,214 & 18,630 & 17,649 & 3,355 \\\hline
orkut    &  21,381s  &  18,071s  & 10.7s & 1,998 & 3,072,626 & 370,343 & 353,606 & 8.6 \\\hline

\end{tabular}
\end{center}
\end{table*}

\begin{table*}[t]
\begin{center}\small 
\caption[]{Comparison of  \TSD and \ADV indexing methods in terms of the index size, index construction time, and query time.}\label{tab:indexcompare}
\begin{tabular}{|l|c|c|c|c|c|c|c|}
\hline \multirow{2}{*}{Network} & \multirow{2}{*}{Graph Size}  & \multicolumn{2}{|c|}{Index Size} & \multicolumn{2}{|c|}{Index Construction Time} & \multicolumn{2}{|c|}{Query Time} \\
\cline{3-8}&  & \TSD & \ADV & \TSD & \ADV & \TSD & \ADV \\\hline
Wiki-Vote    & 1.1MB	& 4.2MB	& 4MB & 9.82s & 8.45s & 7.0ms & 1.8ms \\\hline
Email-Enron  & 3.9MB	& 7.2MB	& 5.6MB & 10.80s & 8.82s & 18.2ms & 5.5ms \\\hline
Epinions     & 5.4MB	& 13.3MB & 13.1MB & 35.36s & 25.79s & 31.9ms & 6.3ms \\\hline
Gowalla      & 21MB	& 34.9MB & 29.7MB & 49.24s & 30.17s & 70.2ms & 23.7ms \\\hline
NotreDame    & 20MB	& 45.4MB & 19.8MB & 286s & 223s & 106ms & 65.4ms \\\hline
LiveJournal  & 478MB & 1,670MB & 1,352MB & 9,297s & 6,689s & 4.9s & 1.2s \\\hline
socfb-konect & 1,510MB & 663MB & 106MB & 1,603s & 629s & 6s & 1.6s \\\hline
orkut        & 1,130MB	& 4,090MB & 3,812MB & 16,012s & 9,819s & 10.7s & 1.7s \\\hline

\end{tabular}
\end{center}
\end{table*}

\begin{figure}[t]
\centering \mbox{
\subfigure[Gowalla]{\includegraphics[width=0.33\linewidth]{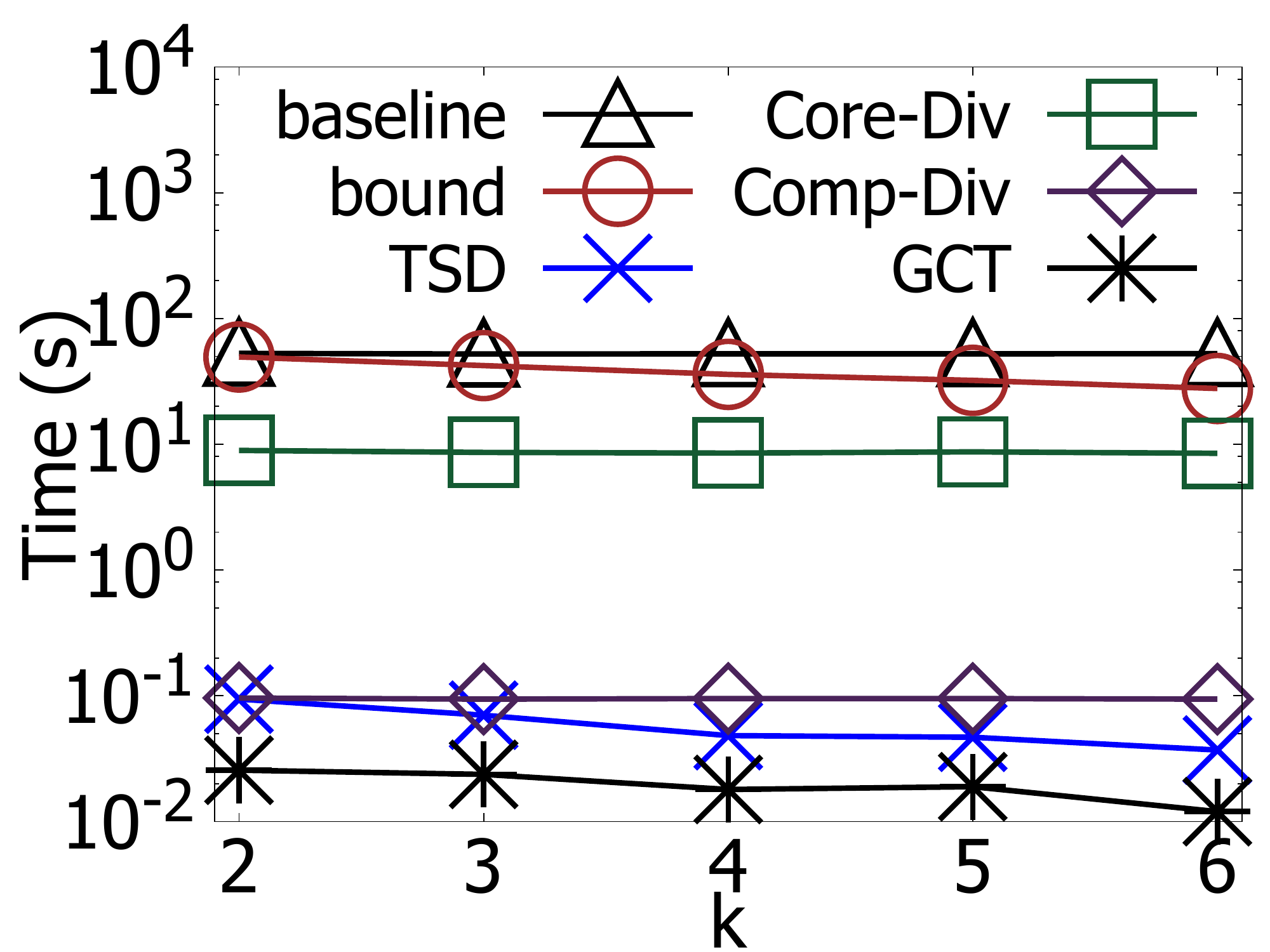}} 
\subfigure[LiveJournal]{\includegraphics[width=0.33\linewidth]{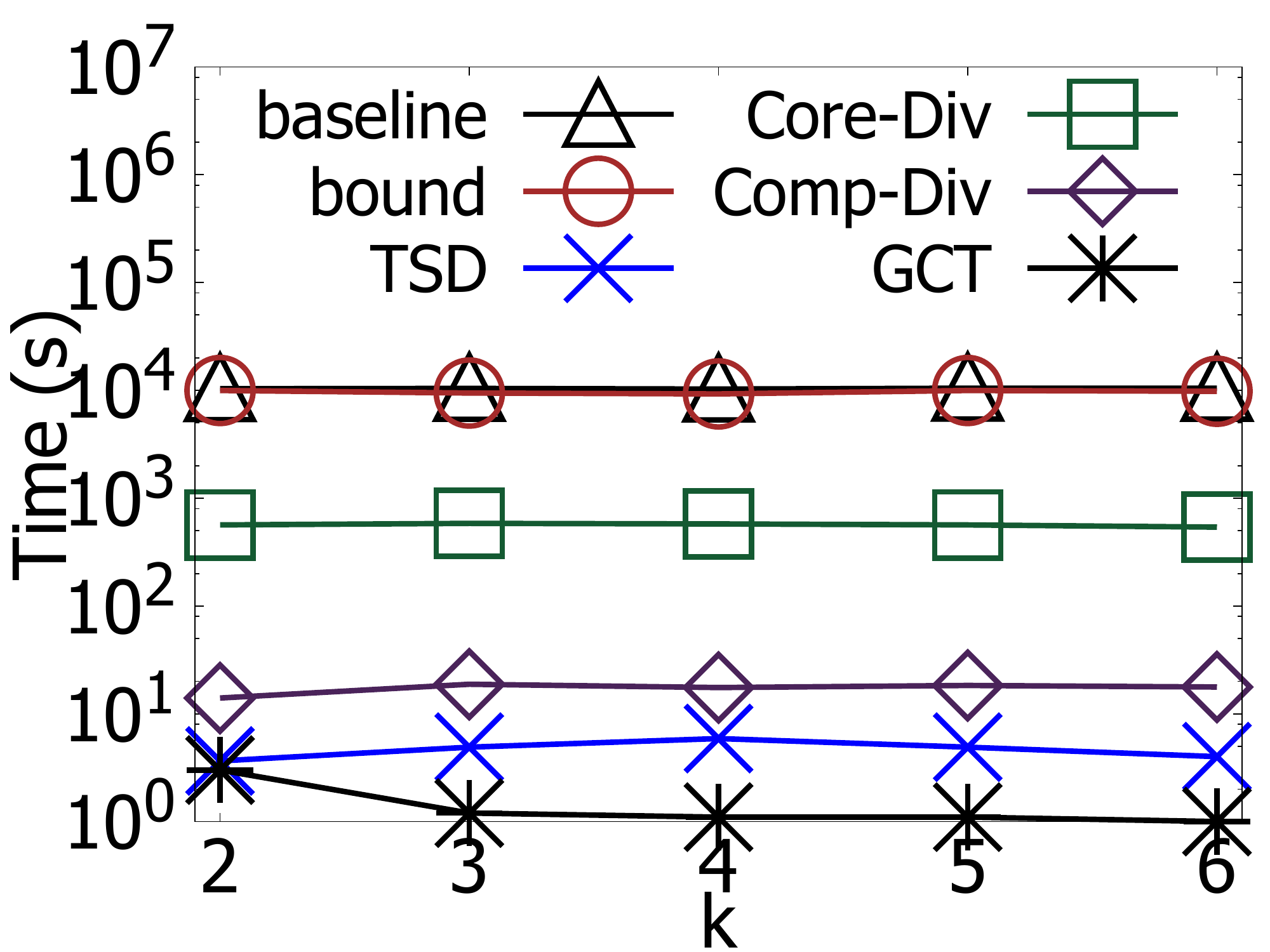}} 
\subfigure[Orkut]{\includegraphics[width=0.33\linewidth]{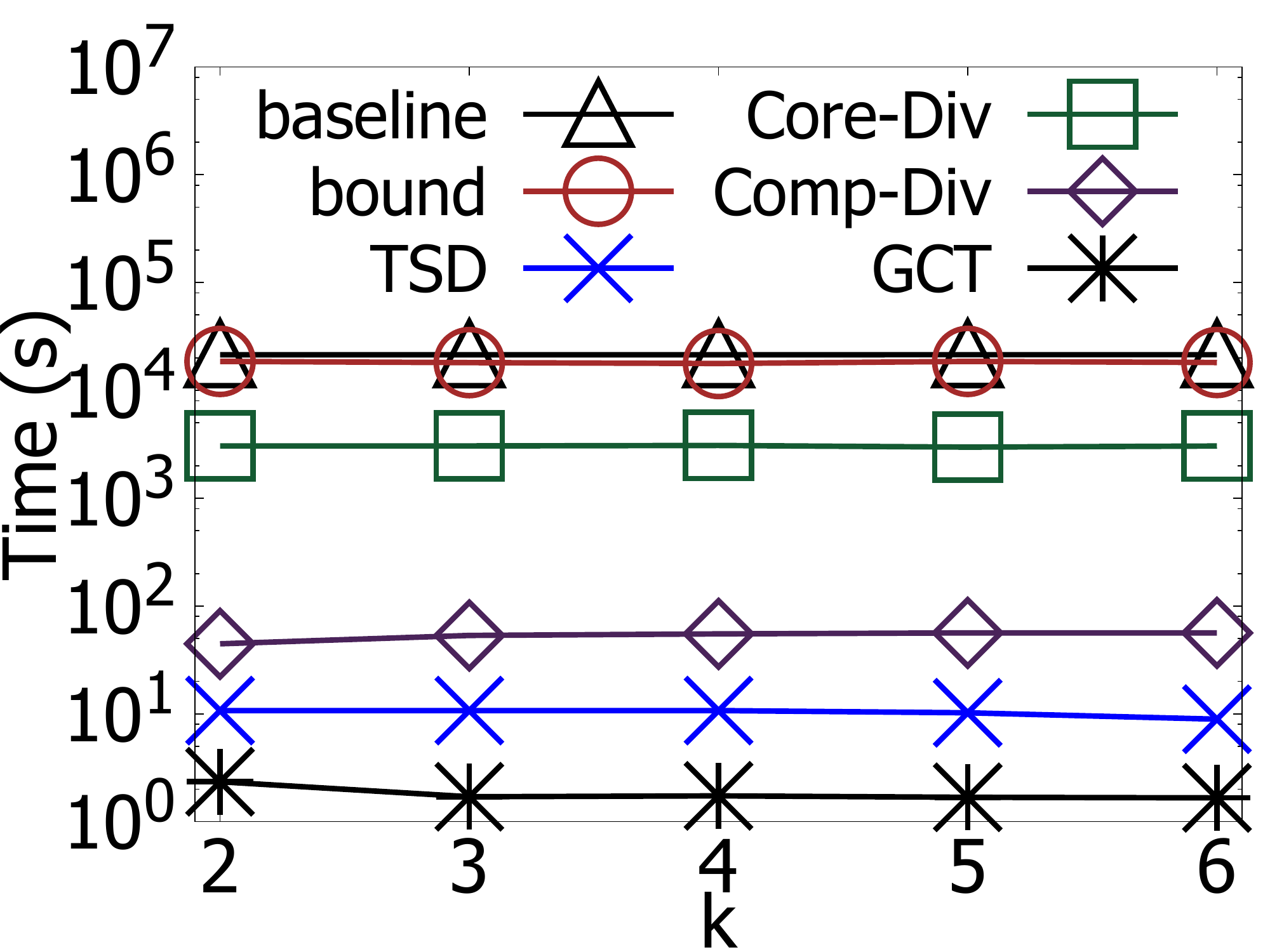}} }
\caption{Comparsion of \baseline, \bound, \core, \component and \TSD in terms of running time (in seconds).}
\label{fig.qt_com}
\end{figure}

\begin{figure}[t]
\centering \mbox{
\subfigure[Gowalla]{\includegraphics[width=0.33\linewidth]{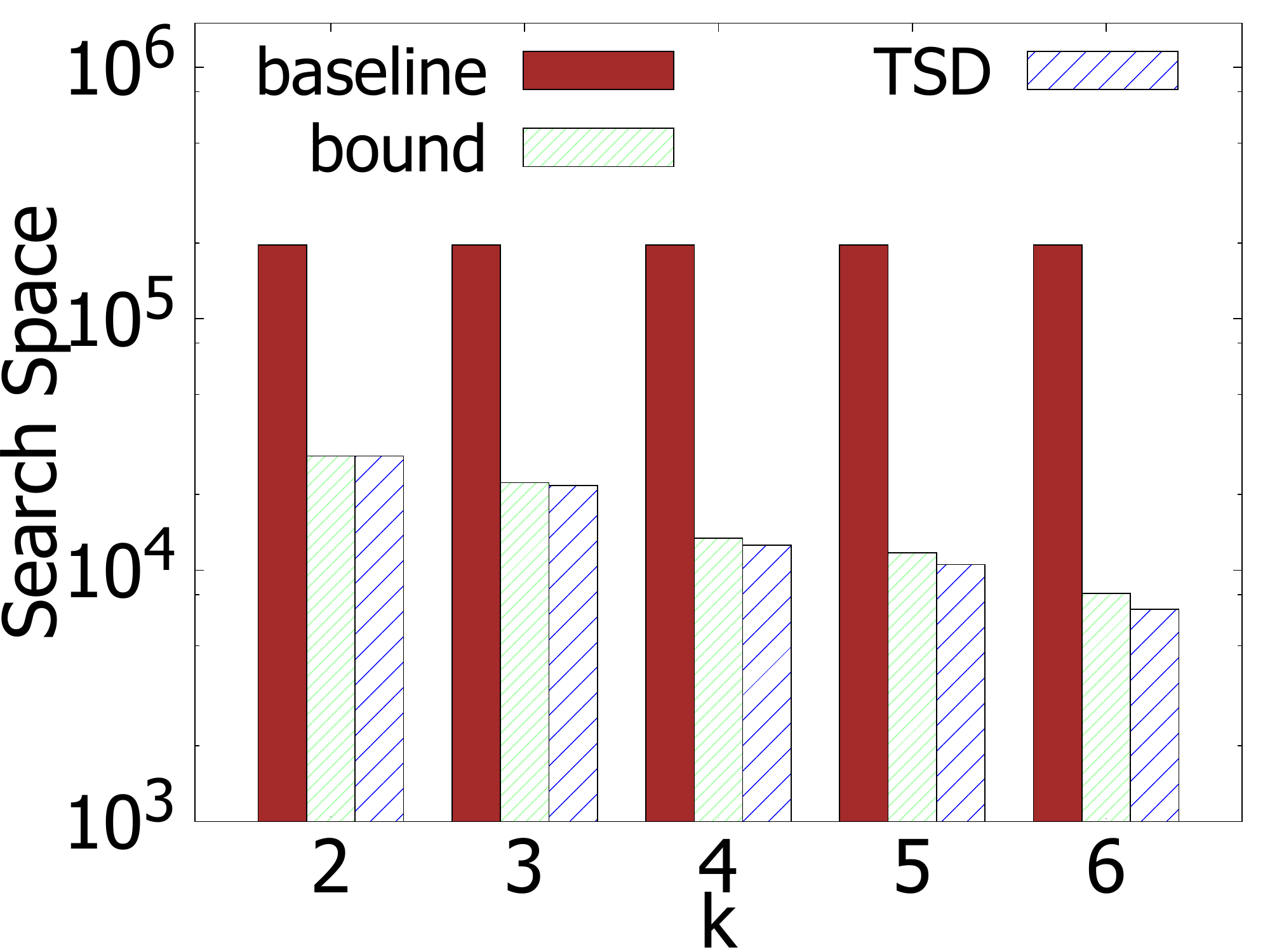}} 
\subfigure[LiveJournal]{\includegraphics[width=0.33\linewidth]{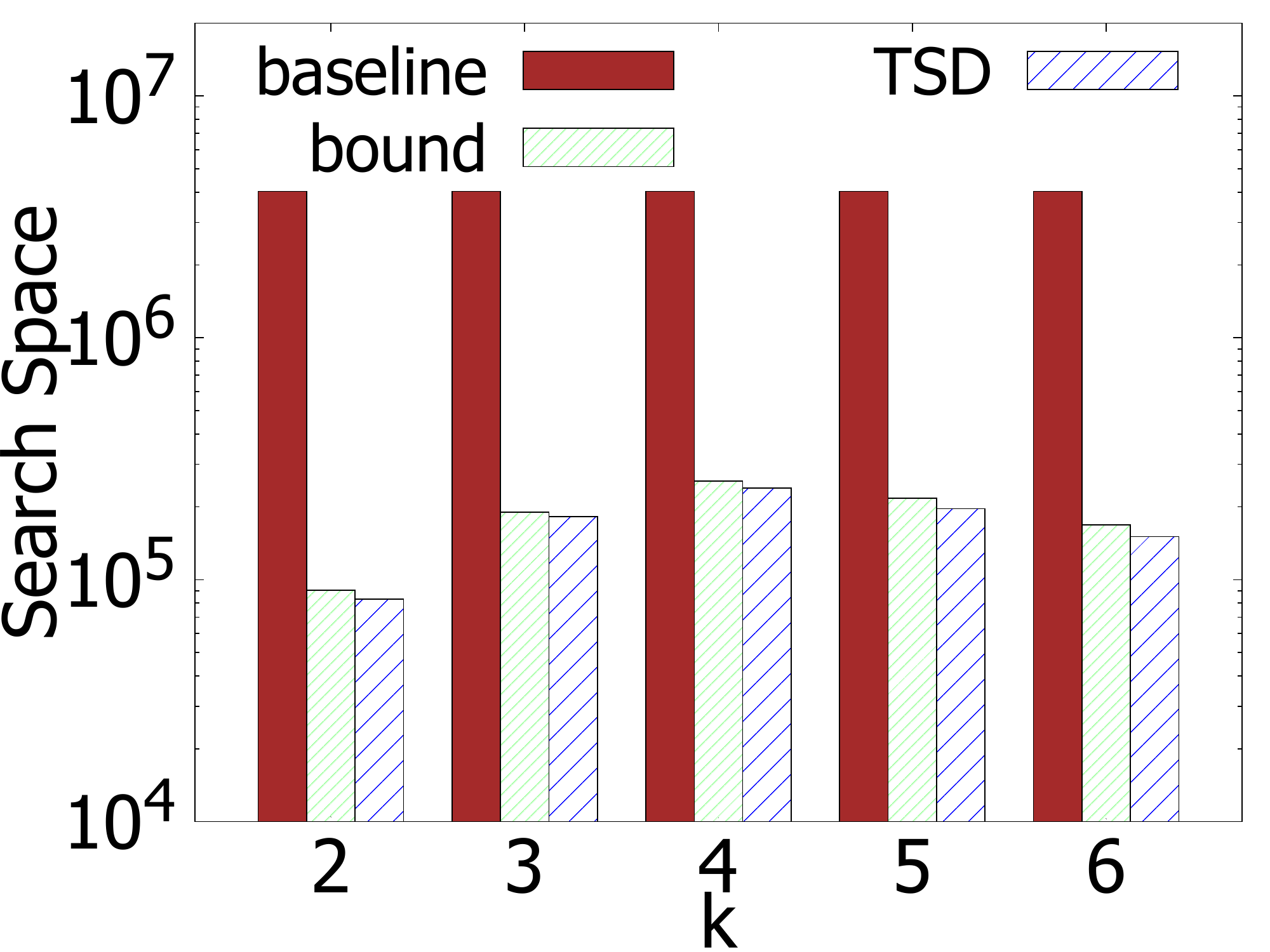}} 
\subfigure[Orkut]{\includegraphics[width=0.33\linewidth]{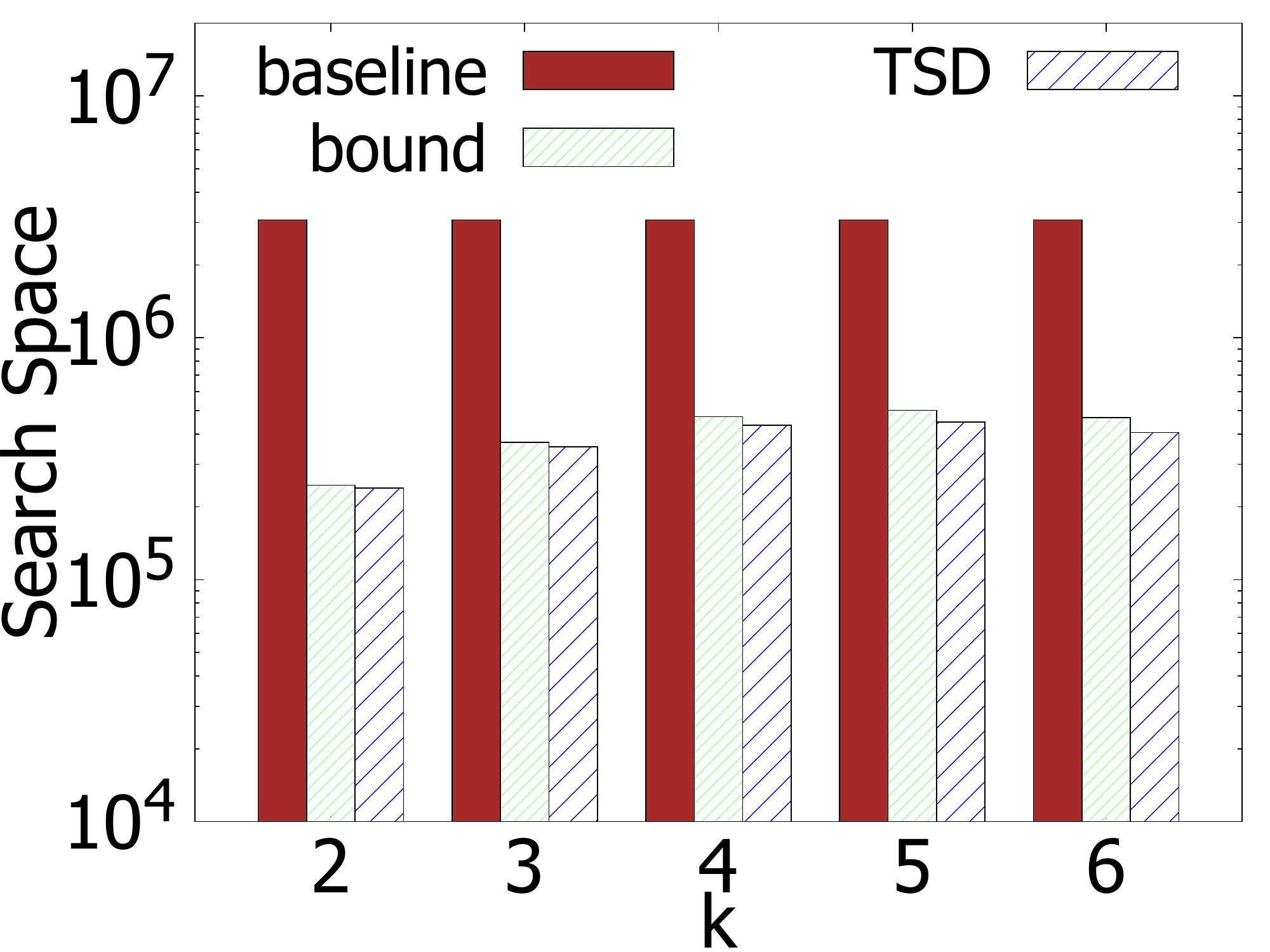}} }
\caption{Comparsion of \baseline, \bound, and \TSD in terms of search space.}
\label{fig.sp_com}
\end{figure}

\begin{figure}[t]
\centering \mbox{
\subfigure[Gowalla]{\includegraphics[width=0.33\linewidth]{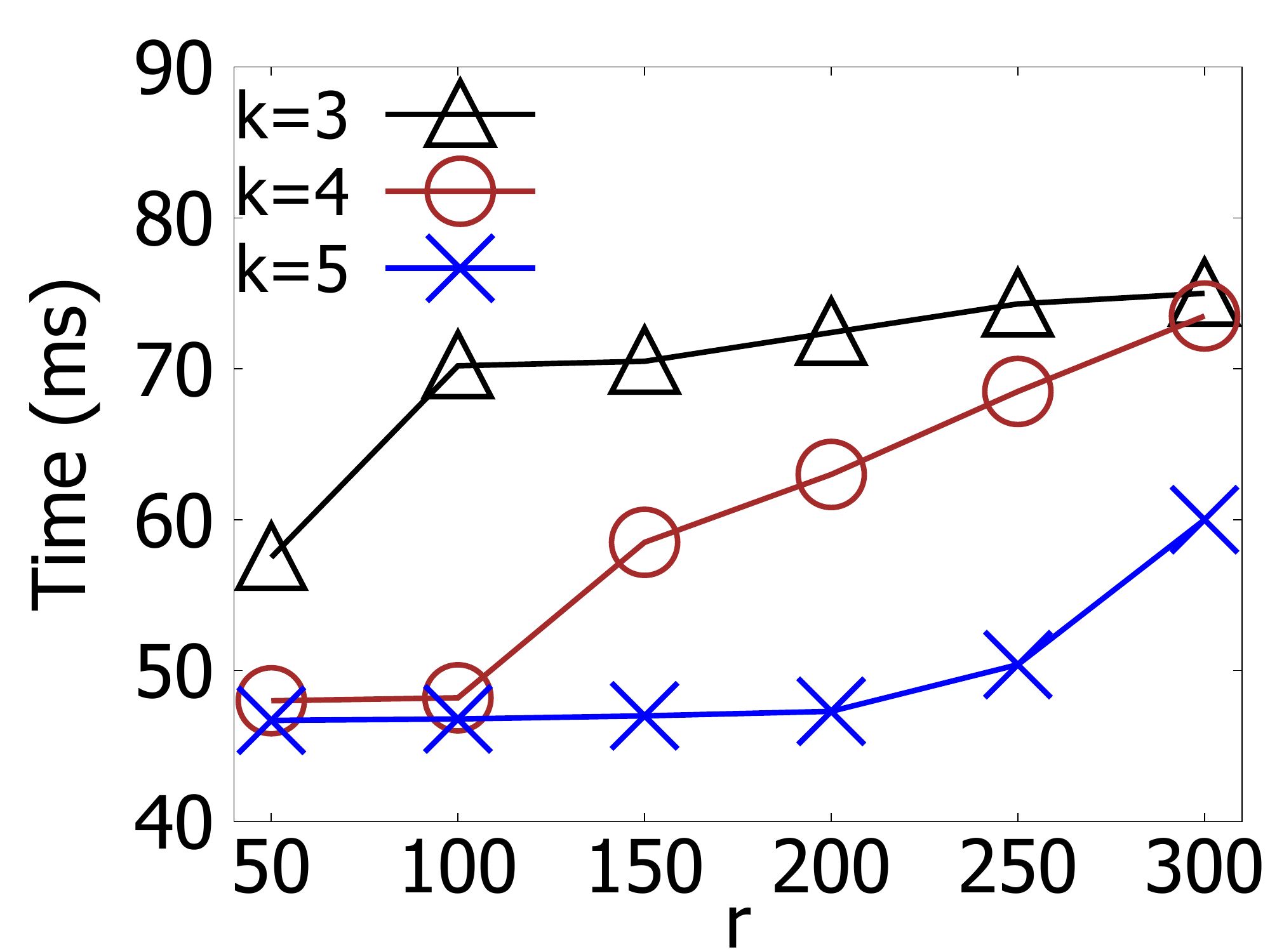}} 
\subfigure[LiveJournal]{\includegraphics[width=0.33\linewidth]{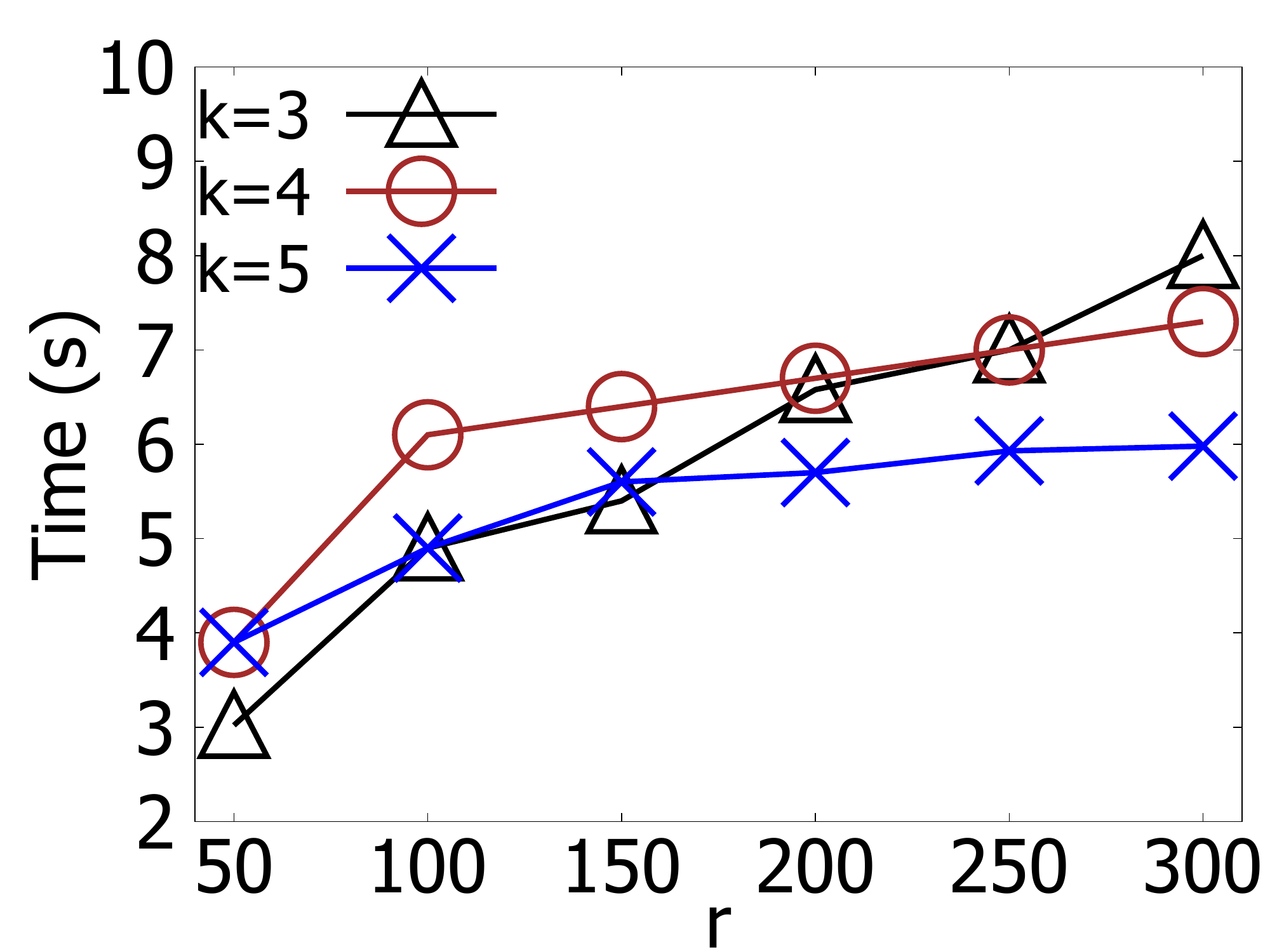}} 
\subfigure[Orkut]{\includegraphics[width=0.33\linewidth]{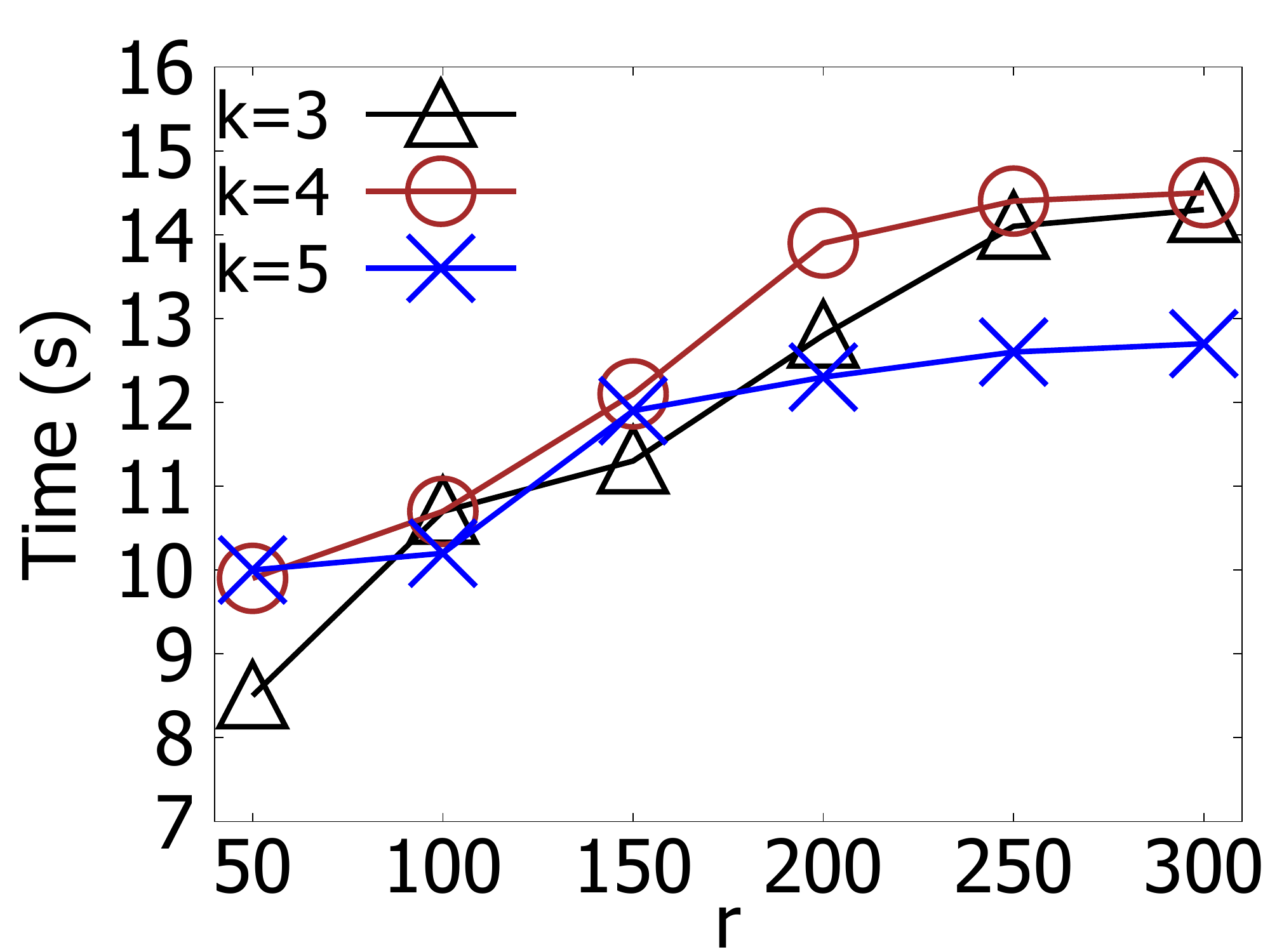}} }
\caption{Running time (in seconds) of \TSD varied by  $k$ and $r$.}
\label{fig.tsd_qt}
\end{figure}

\begin{figure}[t]
\centering \mbox{
\subfigure[Gowalla]{\includegraphics[width=0.33\linewidth]{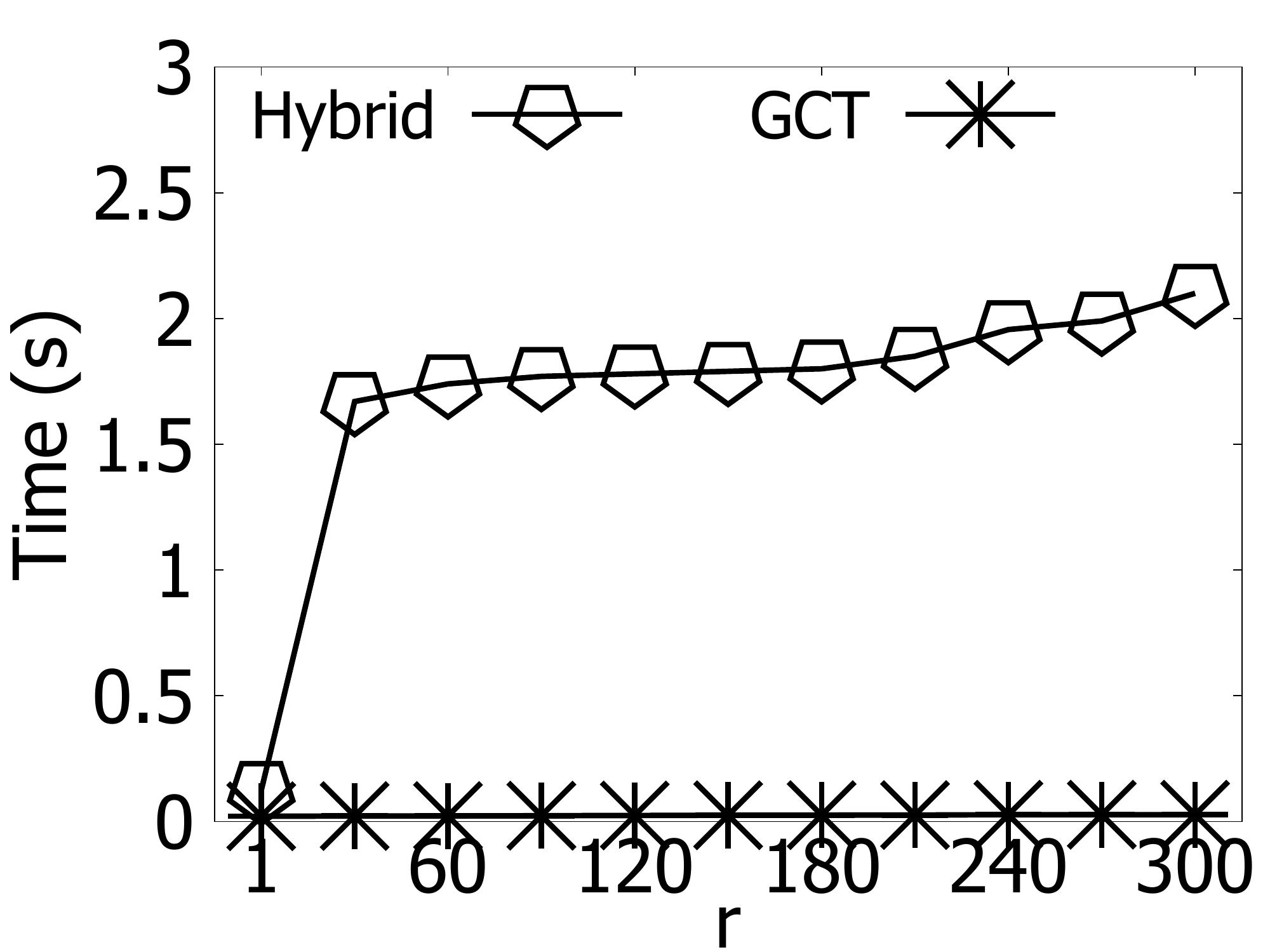}} 
\subfigure[LiveJournal]{\includegraphics[width=0.33\linewidth]{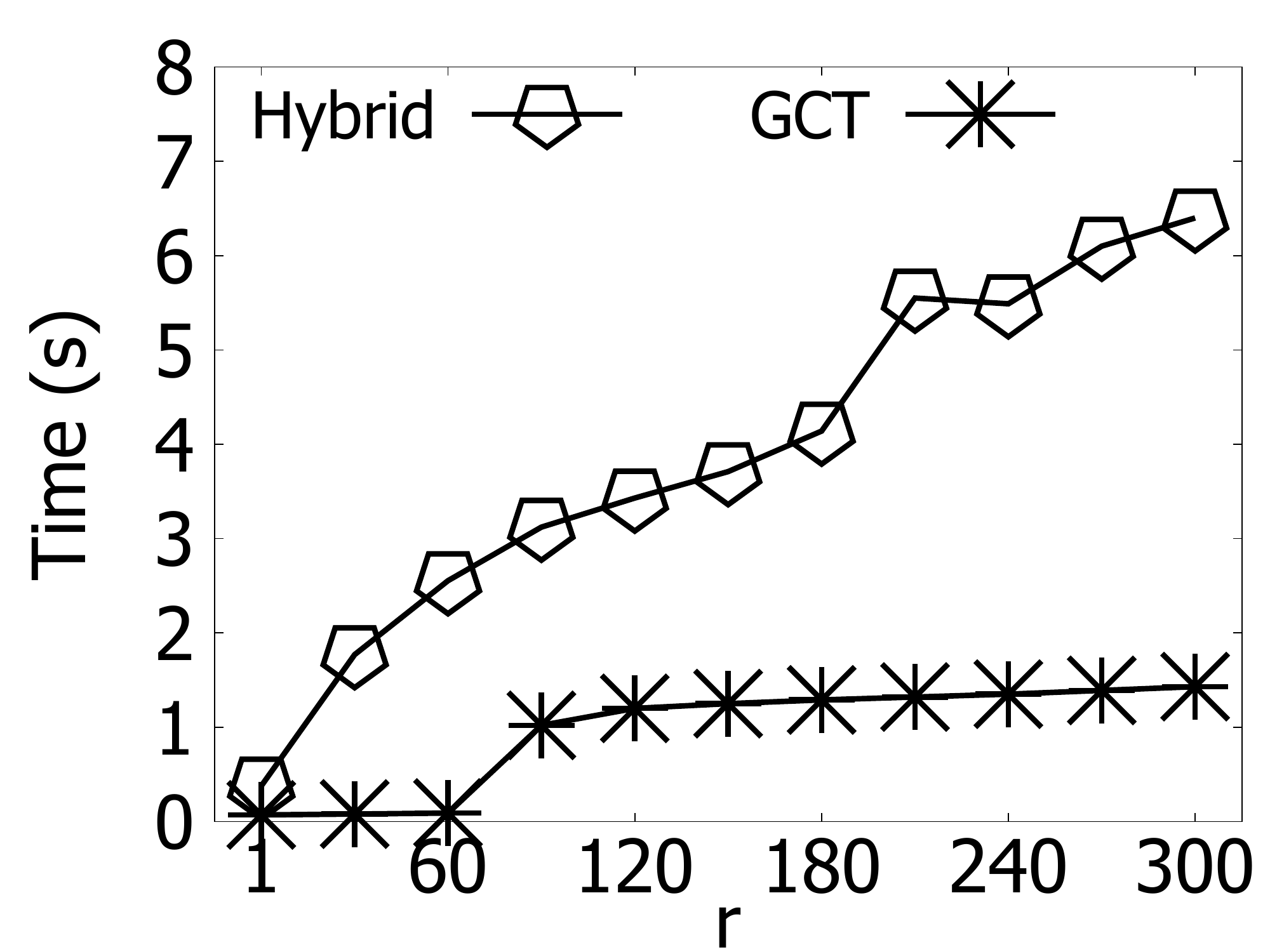}} 
\subfigure[Orkut]{\includegraphics[width=0.33\linewidth]{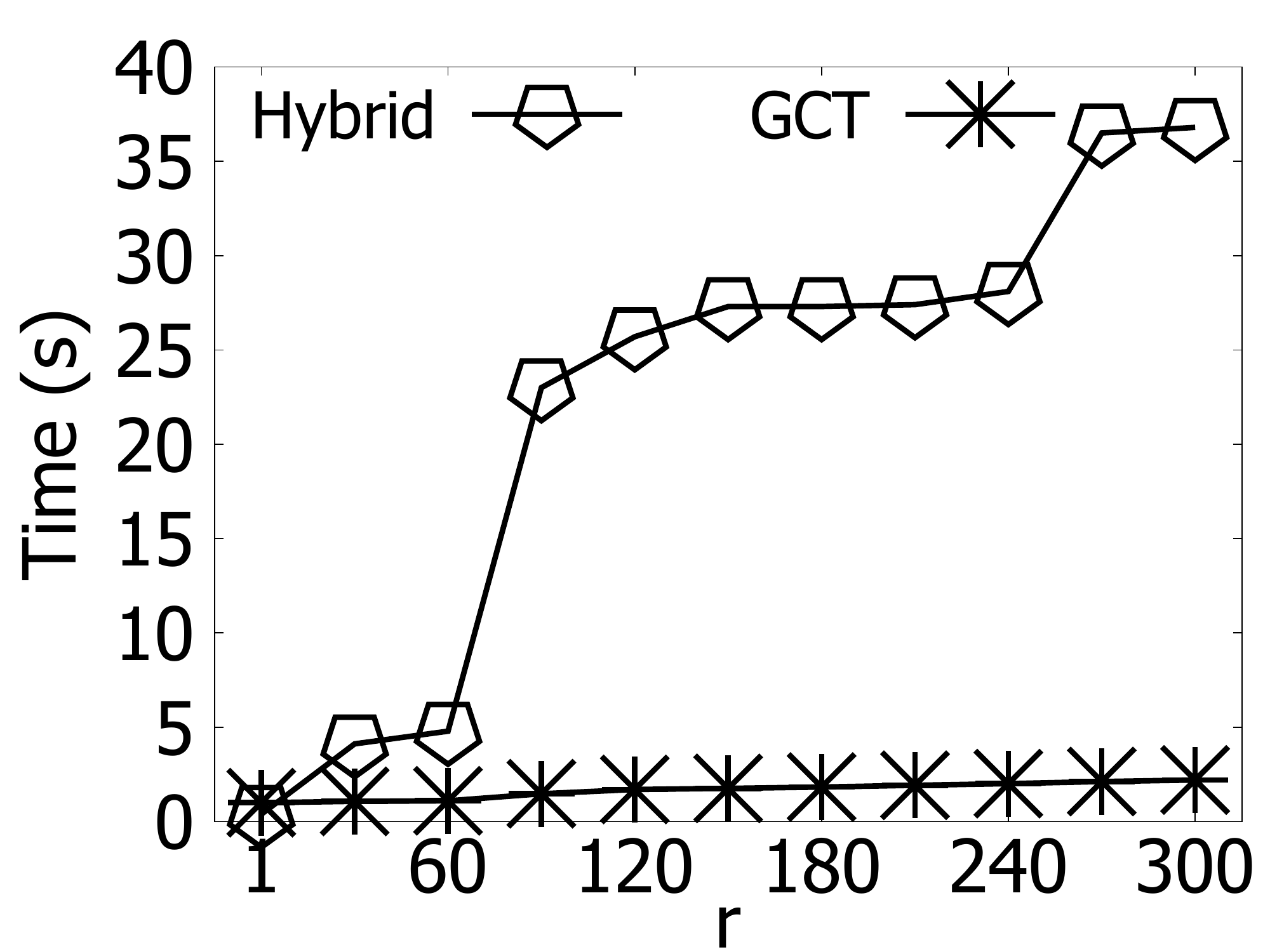}} }
\caption{Running time (in seconds) of \Naive and \ADV varied by $r$.}
\label{fig.comp_top}
\end{figure}






\begin{table}[t]
\begin{center}\vspace*{-0.1cm}
\scriptsize
\caption[]{\textbf{Running time (in seconds) of \TSD and \ADV on Livejournal for ego-network extraction and ego-network truss decomposition.}}\label{tab:compare_ego}
\begin{tabular}{|l|c|c|c|c|}

\hline \multirow{3}{*}{Network} & \multicolumn{2}{|c|}{Ego-network } & \multicolumn{2}{|c|}{Ego-Network Truss} \\
 & \multicolumn{2}{|c|}{Extraction Time} & \multicolumn{2}{|c|}{Decomposition Time} \\
 \cline{2-5} & \TSD & \ADV & \TSD & \ADV \\ \hline
Wiki-Vote     & 3.5s & \textbf{2.2s} & 6.6s & \textbf{4.5s} \\ \hline 
Email-Enron   & 4.4s & \textbf{2.2s} & 5.8s & \textbf{3.9s} \\ \hline
Epinions      & 14s  & \textbf{6.7s} & 18.8s & \textbf{11s} \\ \hline
Gowalla       & 31.2s & \textbf{8.53s} & 16.1s & \textbf{11.8s} \\ \hline
NotreDame     & 49.2s & \textbf{18.5s} & 226s & \textbf{160s} \\ \hline
Livejournal   & 1,094s & \textbf{663s} & 7,902s & \textbf{5,240s} \\ \hline
socfb-konect  & 1,399s & \textbf{135s} & 78.2s & \textbf{75.4s} \\ \hline
orkut         & 7,180s & \textbf{2,469s} & 7,350s & \textbf{4,349s} \\ \hline
\end{tabular}
\end{center}
 \vspace*{-0.4cm}
\end{table}

\subsection{Efficiency Evaluation}

\stitle{Exp-1 (Efficiency comparison on all datasets)}: 
We compare the efficiency of our proposed methods on all datasets. Table~\ref{tab:querytime} shows the results of running time and search space. 
Clearly, \TSD is the most efficient in terms of running time, and \baseline is the worst. 
\TSD uses less search space than \bound, indicating a stronger pruning ability of $\widetilde{score}(v)$  against  $\overline{score}(v)$ in Lemma~\ref{lemma.trussbound}. The speedup ratio $R_t$ between \TSD and \baseline is defined by $R_t= t_{\baseline} /t_{\TSD}$ where $t_{\baseline}$ and $t_{\TSD}$ are the running time of \baseline and \TSD respectively. The speedup ratio $R_t$ (column 5 in Table~\ref{tab:querytime}) ranges from 265 to 2,745. In other words, our method \TSD achieves up to 2,745X speedup on the network NotreDame.  In addition, the pruning ratio $R_s$ between \TSD and \baseline is defined by $R_s= S_{\baseline} /S_{\TSD}$ where $S_{\baseline}$ and $S_{\TSD}$ are the search space of \baseline and \TSD respectively.  The pruning ratio $R_s$ (column 9 in Table~\ref{tab:querytime}) ranges from 3.1 to 3,355, which reflects an efficient pruning strategy of \TSD.

\stitle{Exp-2 (Efficiency comparison of all different methods)}: We vary parameter $k$ to compare the efficiency of all different methods. We compare six methods of  \baseline, \bound,  \TSD, \ADV, \component, and \core on three datasets Gowalla, Livejournal, and Orkut. The results of running time and search space are respectively reported in Figure~\ref{fig.qt_com} and Figure~\ref{fig.sp_com}. Similar results can be also observed on other datasets. \ADV is a clear winner for the varied $k$ on all datasets. Thanks to  efficient  \ADVindex, \ADV significantly outperforms two state-of-the-art methods of \component and \core on large networks of LiveJournal and Orkut. Moreover, \ADV outperforms \TSD, indicating the superiority of a more compact \ADVindex against \TSDindex. In addition, we report the search space results in Figure~\ref{fig.sp_com}. It shows that the search space is significantly reduced by \bound against \baseline on all datasets, indicating the technical superiority of graph sparsification and the upper bound of $\overline{score}(v)$. \TSD performs the best in search space by leveraging another tight upper bound $\widetilde{score}(v)$, which learns structural information from the \TSDindex.  


\stitle{Exp-3 (Indexing scheme comparison between \TSD and \ADV)}: 
We compare two indexing methods of \TSD and \ADV in terms of index construction time, index size, and index-based query processing time of structural diversity search. The results of \TSD and \ADV on all dataset are reported in Table~\ref{tab:indexcompare}. 
The index size of \ADVindex is smaller than the size of \TSD, due to a compact structure of \ADVindex by discarding unnecessary edges within social contexts.
\ADV achieves a much faster index construction time than \TSD, thanks to the efficient techniques of fast ego-network extraction and bitmap-based truss decomposition. Specifically, Table~\ref{tab:compare_ego} reports the detailed running time of ego-network extraction and ego-network truss decomposition by \TSD and \ADV on all datasets. This reflects that \ADV achieves significant accelerations on both ego-network extraction and ego-network truss decomposition, which validates the superiority of our speed up techniques proposed in Section~\ref{sec.bitmap}. 
\ADVindex achieves faster index construction time and smaller index size. 
In addition, as shown in the columns 7 and 8 of Table~\ref{tab:indexcompare}, \ADV runs much faster than \TSD in terms of query time of structural diversity search.


\stitle{Exp-4 (Efficiency comparison of \ADV and \Naive)}: 
In this experiment, we compare \ADV with a very competitive method \Naive. 
As a hybrid approach of partial answer saving and online search, \Naive keeps in advanced the top-$r$ vertices for all possible $k$ and $r$. For an input query of parameters $k$ and $r$, \Naive can directly get the answer of top-$r$ vertices and then computes the corresponding social contexts using Algorithm~\ref{algo:comp-score} in an online manner. The main cost of \Naive is the social context computation.  Figure~\ref{fig.comp_top} shows the running time of \Naive and \ADV on three datasets by varying $r$ from 1 to 300 and $k=3$. \Naive is comparative to \ADV when $r=1$. However, when $r$ goes larger, \ADV is significantly faster than \Naive on all datasets, which reflects the superiority of our \ADVindex-based diversity search.

\jinbin{

\stitle{Exp-5 (Varying $k$ and $r$ for \TSD)}:  Figure~\ref{fig.tsd_qt} shows the running time of \TSD when varying  different parameters of $k$ and $r$. Each curve represents the \TSD using one value of parameter $k$. We observe that the running time mostly decreases with a larger value of $k$. \TSD takes a slight more time with the increased $r$, indicating a stable efficiency performance. Similar results are also observed on other datasets. 
}

\jinbin{

\stitle{Exp-6 (Scalability test)}: To evaluate the scalability of our proposed methods, we generate a series of power-law graphs using the PythonWeb Graph Generator\footnote{\small{\url{http://pywebgraph.sourceforge.net/}}}. We vary $|V|$ from 1,000,000 to 10,000,000, and $|E|=5|V|$. Figure~\ref{fig.sca}(a) shows the index construction time of \TSDindex, which scale  well with the increasing vertex number.  Figure~\ref{fig.sca}(b) shows the running time of \TSD. It takes a few seconds to process the truss-based structural diversity search on all networks.

}

\begin{figure}[t]
\centering \mbox{
\subfigure[\TSDindex construction]{\includegraphics[width=0.45\linewidth]{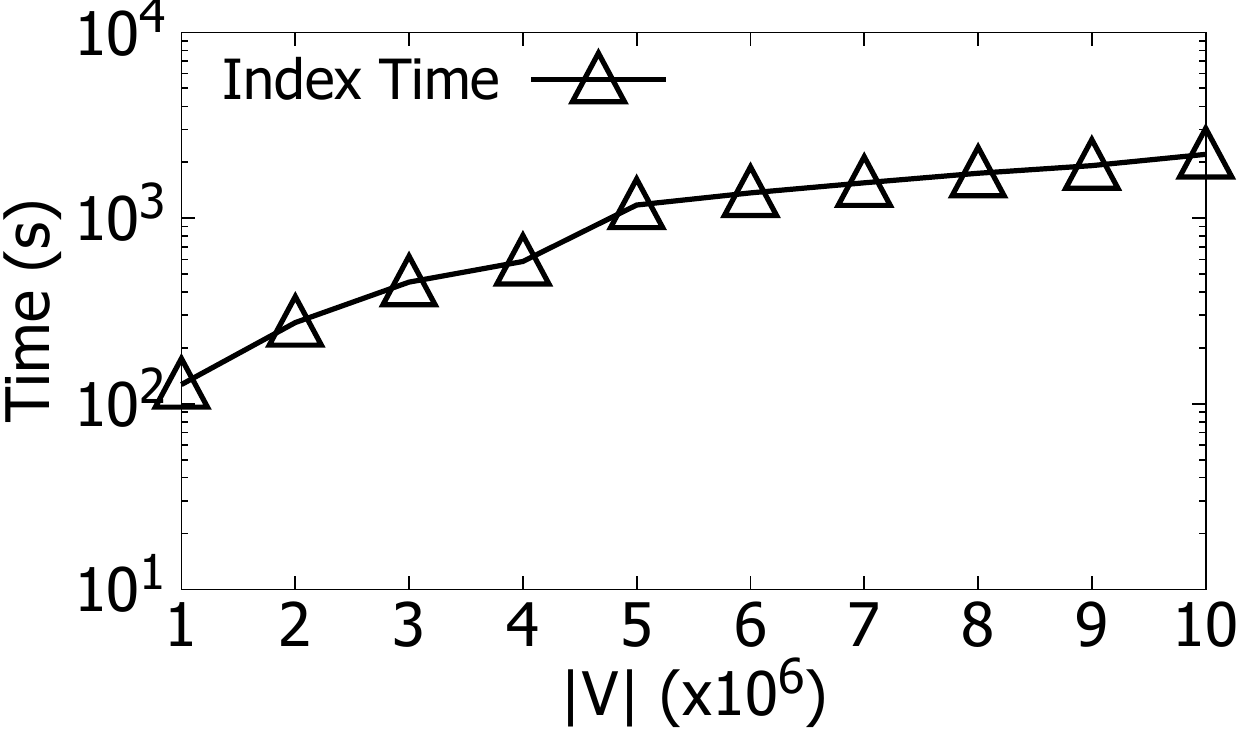}} 
\subfigure[\TSD]{\includegraphics[width=0.45\linewidth]{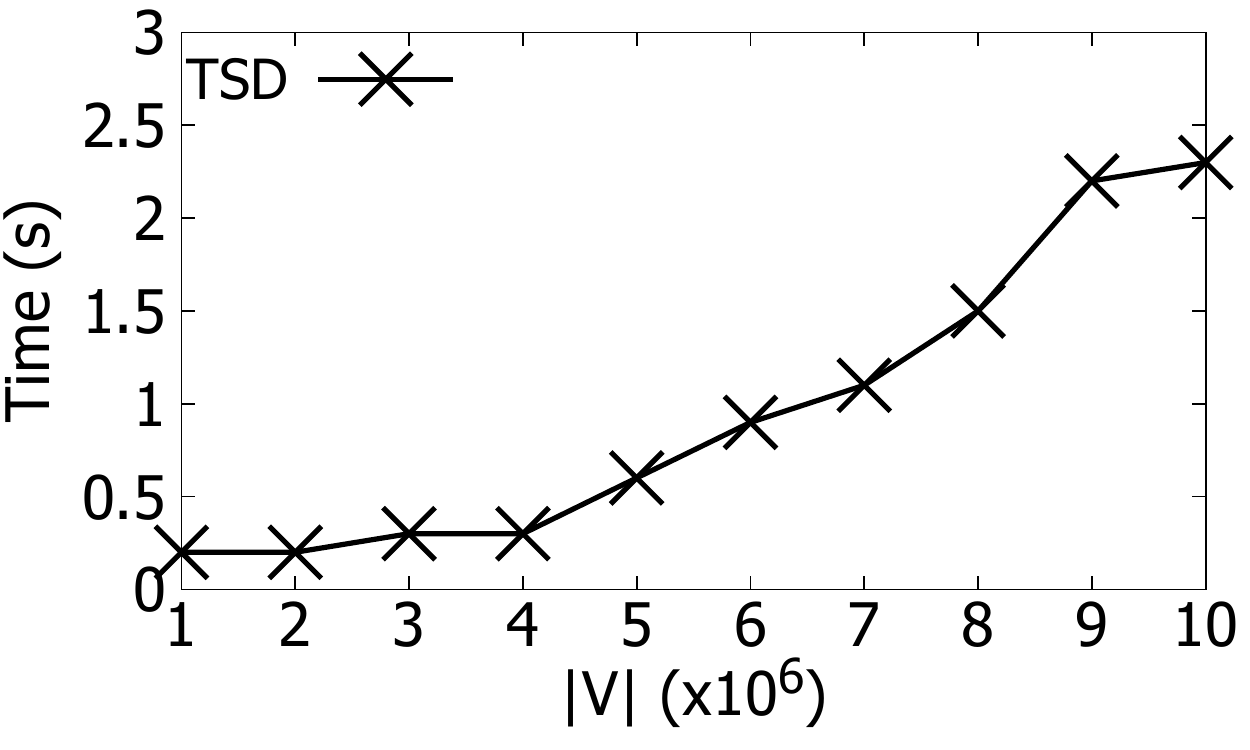}}}
 \vspace*{-0.4cm}
\caption{Scalability test of \TSDindex construction and \TSD  on power-law graphs in terms of running time (in seconds).}
\label{fig.sca}
\end{figure}

\begin{figure}[t]
\centering \mbox{
\subfigure[Gowalla]{\includegraphics[width=0.33\linewidth]{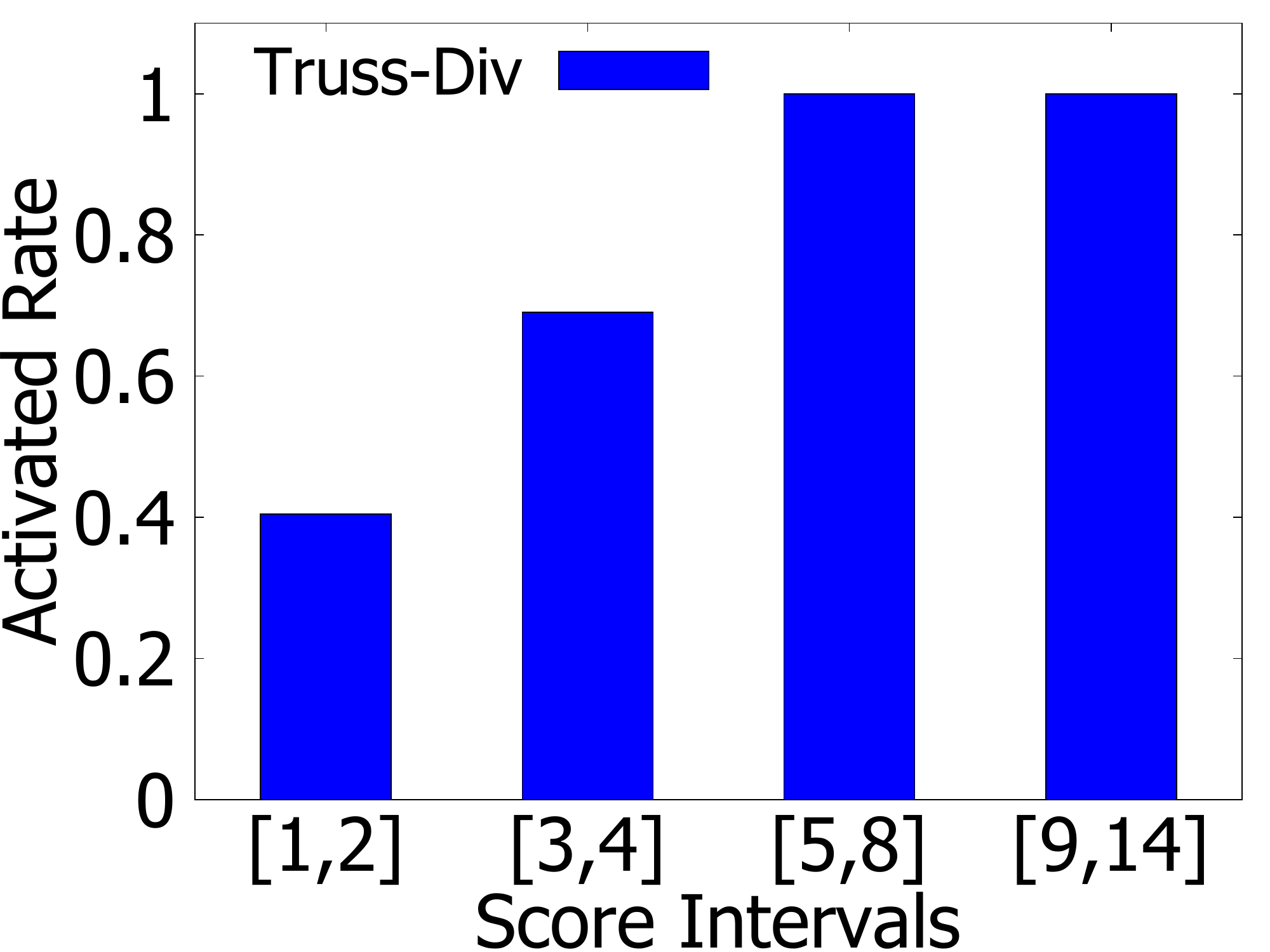}} 
\subfigure[LiveJournal]{\includegraphics[width=0.33\linewidth]{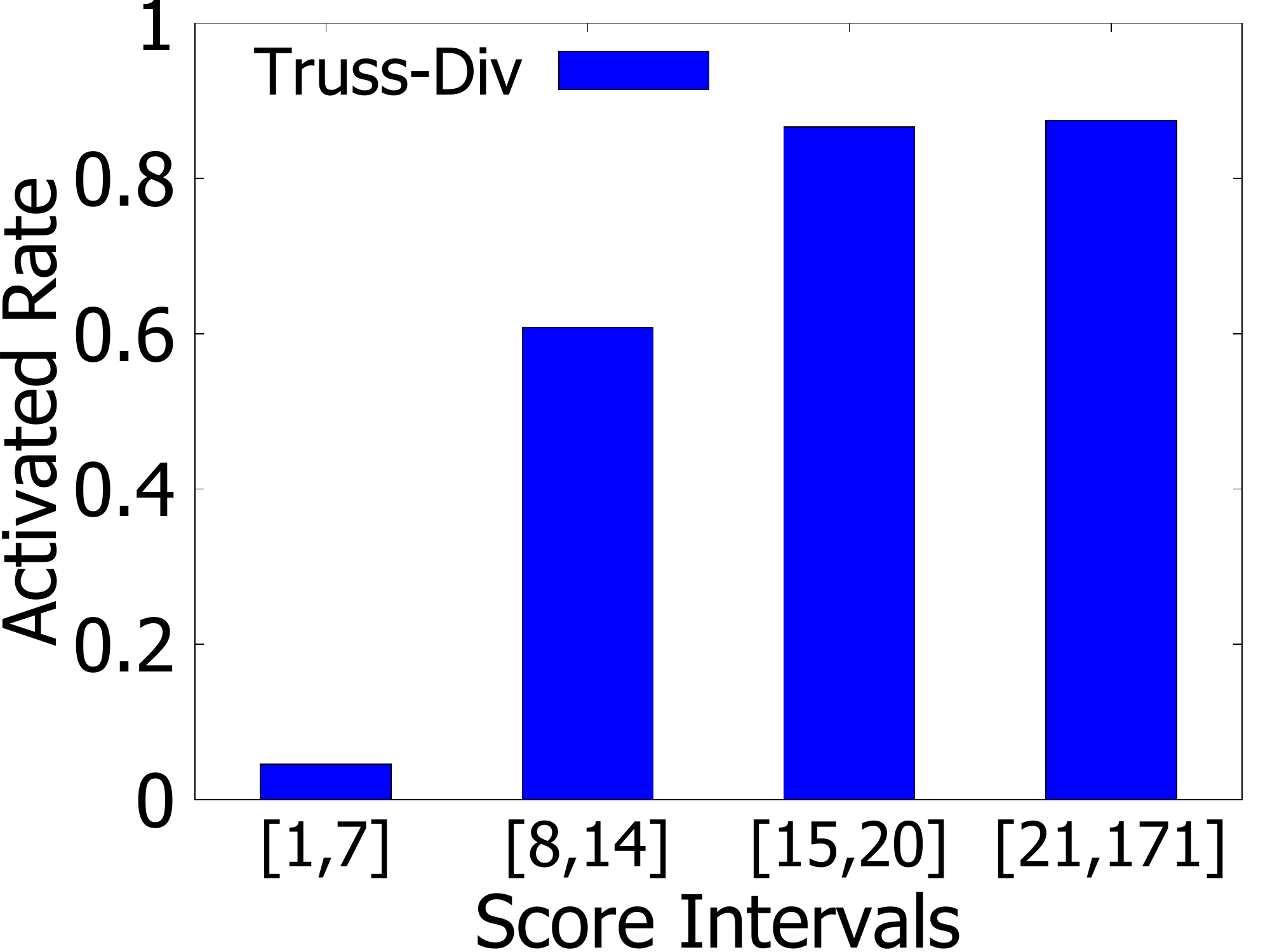}} 
\subfigure[Orkut]{\includegraphics[width=0.33\linewidth]{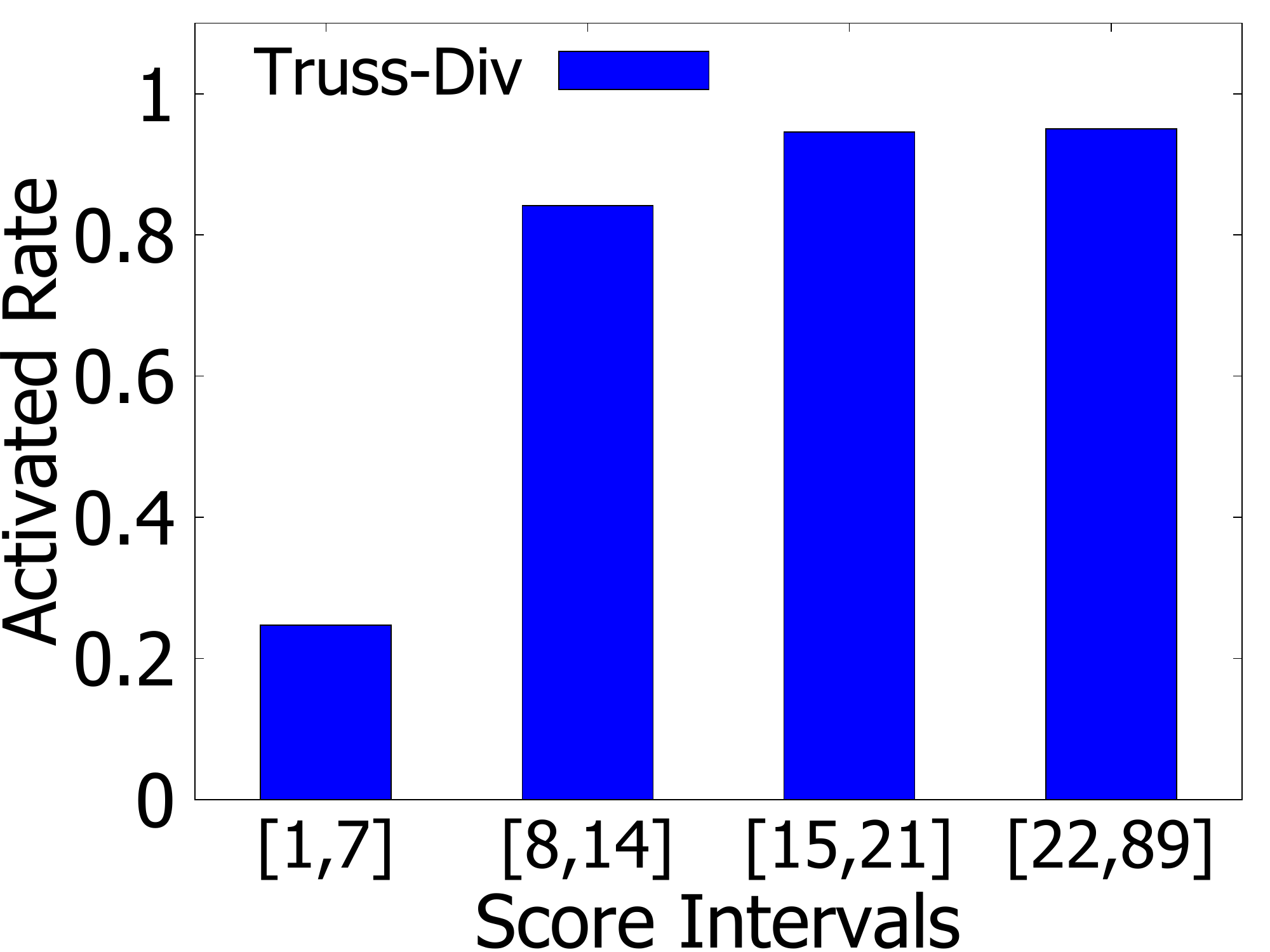}} }
 \vspace*{-0.4cm}
\caption{Correlation between social contagion and truss-based structural diversity.}
\label{fig.propagate_score}
\end{figure}

\subsection{Effectiveness Evaluation}
\eat{
This experiment evaluates the effectiveness of truss-based structural diversity model for social contagion.
We simulates the process of influence propagation using the independent cascade model \cite{goyal2011celf++}. We first apply influence maximization algorithm \cite{tang2015influence} on graph $G$ to obtain 50 vertices as a set of  activated seeds. Then, we evaluate the number of target vertices (output by different approaches) that get activated (social contagion) by these seeds in the influence propagation. The Monte Carlos sampling is performed in 10,000 times. We treat undirected graphs as directed graphs, by regarding each undirected edge $e=(u,v)$ as two directed edges $<$$u, v$$>$ and $<$$v, u$$>$, with the same influential probability $p(e) =0.01 $ by default.
}

\jbhuang{
This experiment evaluates the effectiveness of truss-based structural diversity model for social contagion. As mentioned in the introduction, social contagion is an information diffusion process that a user of a social network gets affected by the information propagated from his/her neighbors. In this experiment, we simulate the social contagion by the process of influence propagation using the independent cascade model \cite{goyal2011celf++,bian2020efficient}. In the independent cascade model, vertices in the input graph have two state: unactivated and activated. Initially, we apply influence maximization algorithm \cite{tang2015influence} on graph $G$ to obtain 50 vertices as a set of activated seeds. Then we uses these seeds to influence  their neighbors. If one of their neighbors get activated from the previous unactivated status, we say that this vertex gets contagion. For a activated seed $u$ and its unactivated neighbor $v$, the successful activation of $v$ from $u$ only depends on the edge probability between $u$ an $v$. We perform the Monte Carlos sampling for 10,000 times. Then, we evaluate the number of target vertices (output by different approaches) that get activated (social contagion) by these seeds in the influence propagation. We treat undirected graphs as directed graphs, by regarding each undirected edge $e=(u,v)$ as two directed edges $<$$u, v$$>$ and $<$$v, u$$>$, with the same influential probability $p(e) =0.01 $ by default.
}



\stitle{Exp-7 (Correlation between social contagion and truss-based structural diversity)}: This experiment attempts to validate the correlation between social contagion and truss-based structural diversity. We test whether the vertices with higher truss-based structural diversity scores would have higher probabilities to get activated. We set the parameter $k=4$. According to the scores of truss-based structural diversity, we partition the vertices into 4 groups with different score intervals from low to high. We report the activated rate of each group, that is, the number of activated vertices over the total number of vertices in this group. Figure~\ref{fig.propagate_score} reports the activated rates of all groups on three networks of Gowalla, LiveJournal, and Orkut. The results show that the vertices having higher scores are more easily to get activated. It confirms that truss-based structural diversity is a good predictor for social contagion.

\begin{figure}[t]
\centering \mbox{
\subfigure[Gowalla]{\includegraphics[width=0.33\linewidth]{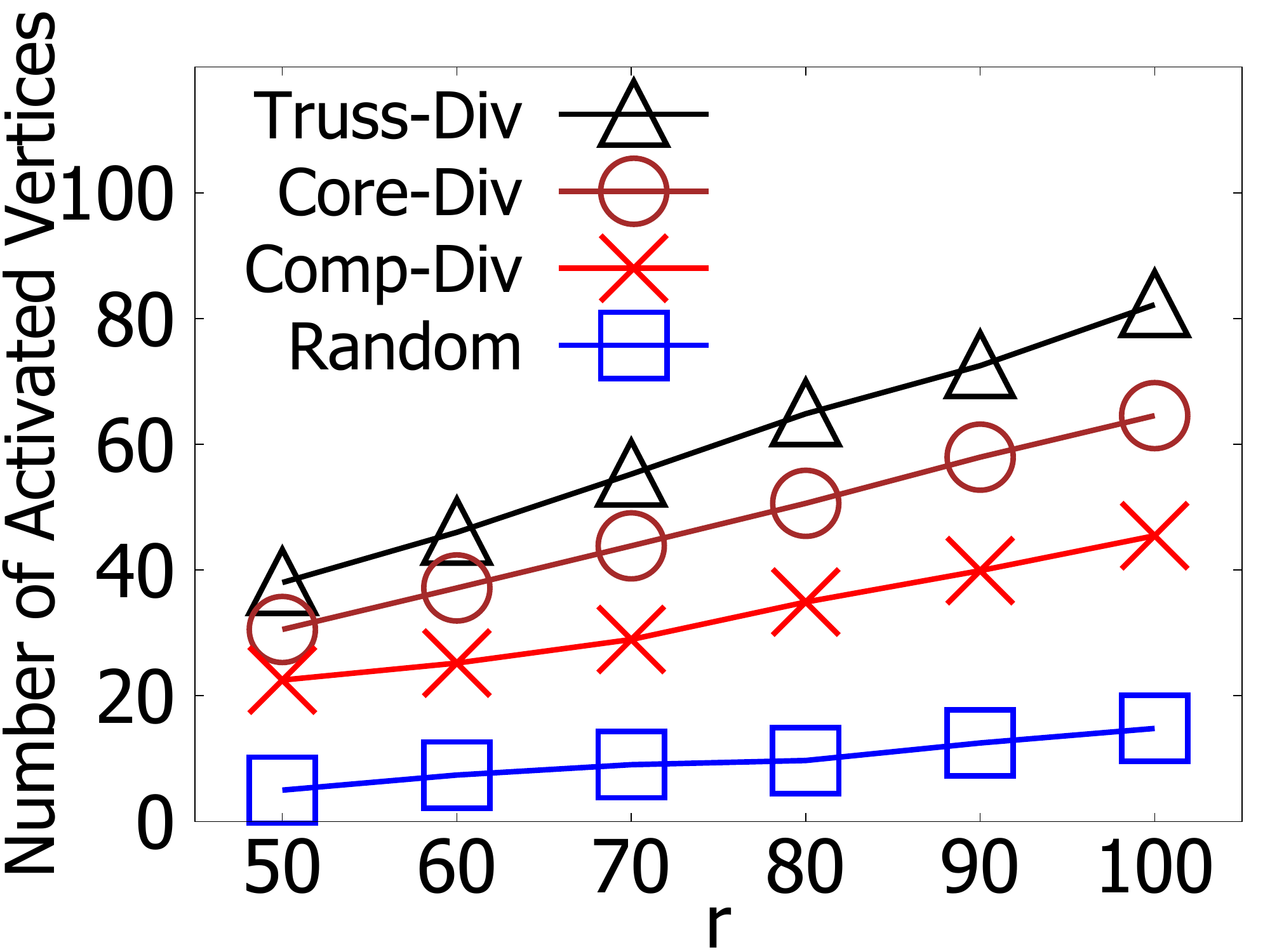}} 
\subfigure[LiveJournal]{\includegraphics[width=0.33\linewidth]{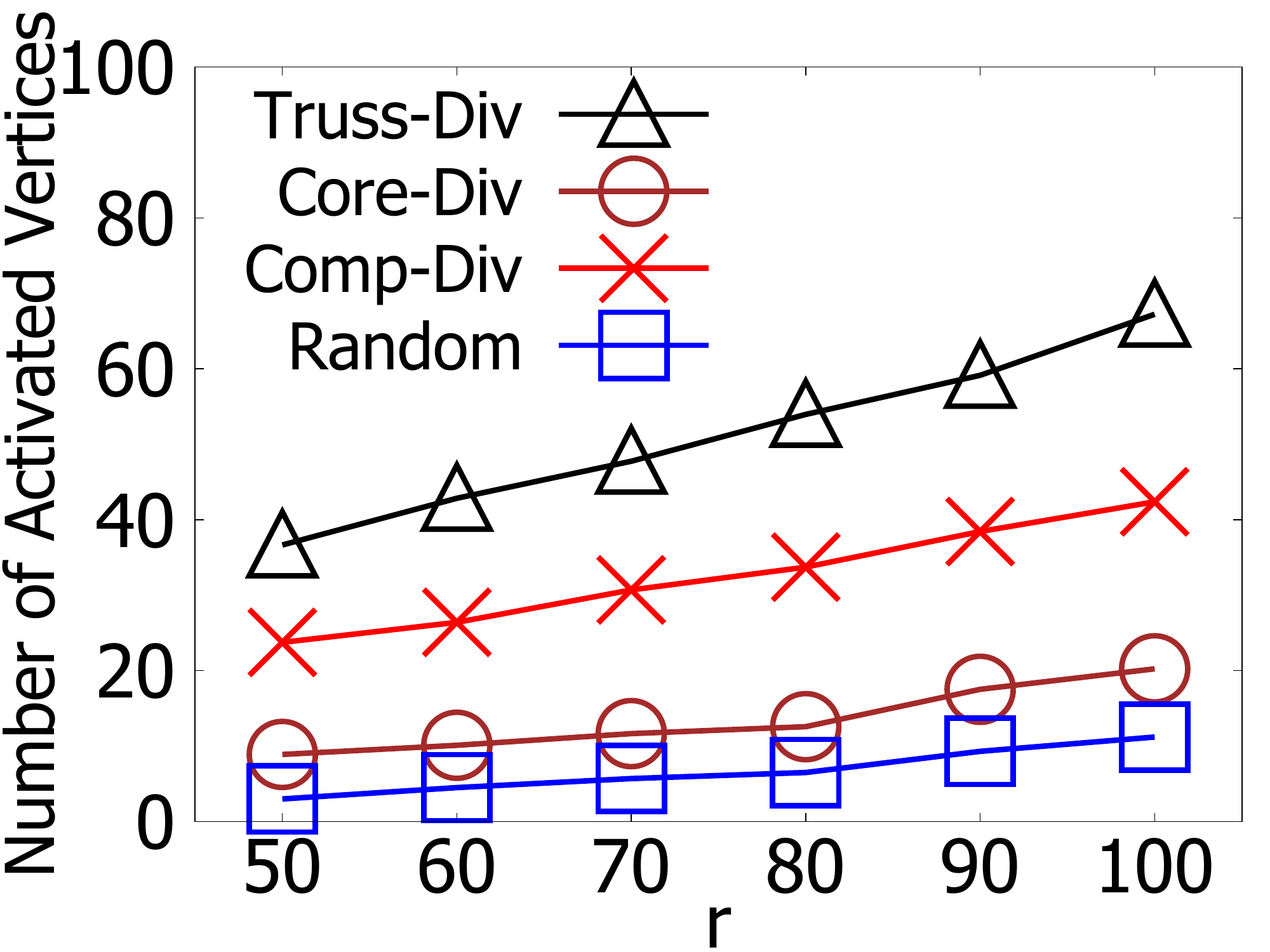}} 
\subfigure[Orkut]{\includegraphics[width=0.33\linewidth]{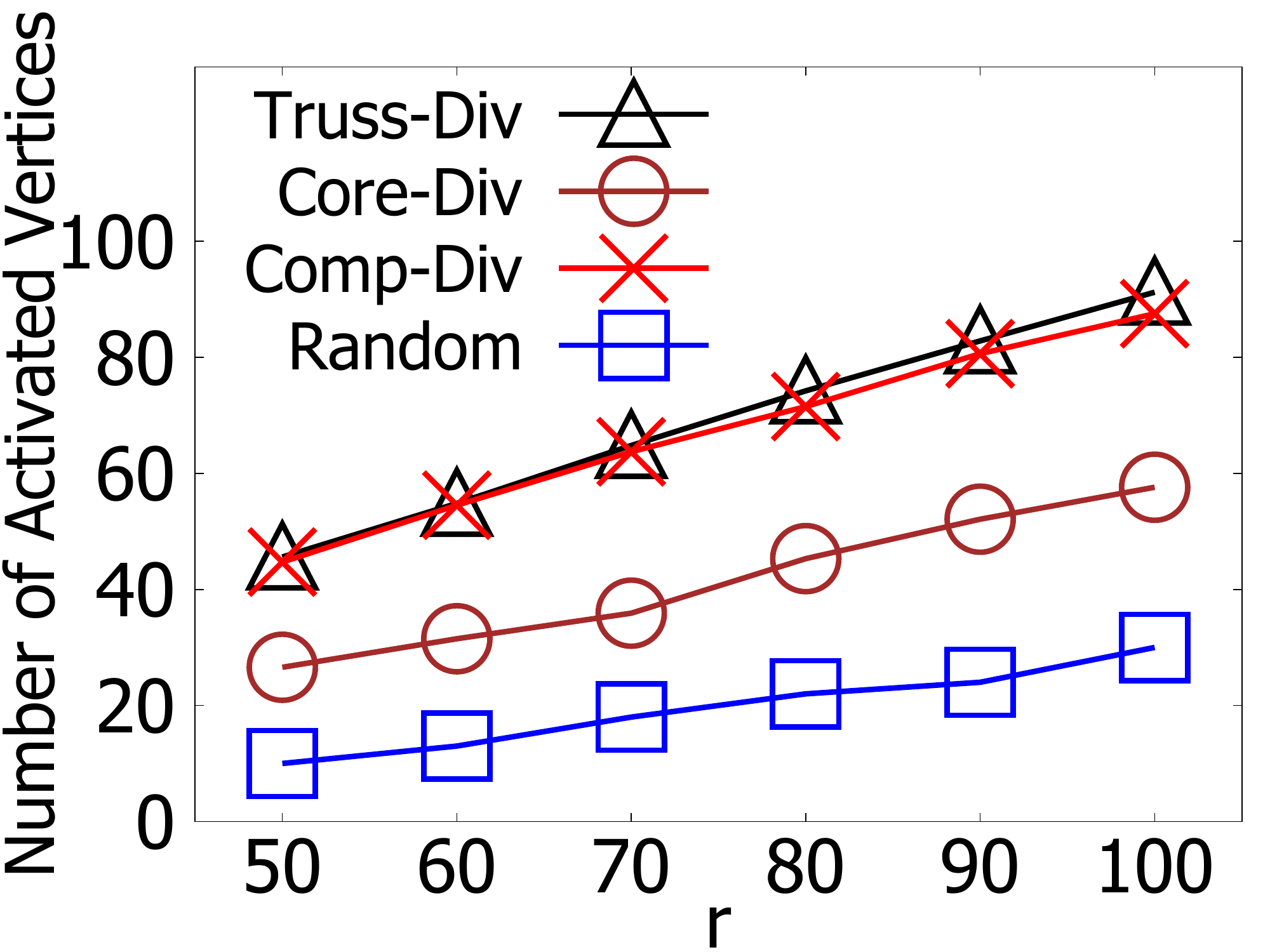}} }
\caption{Comparison of \random, \component, \core, and \TSD in terms of the number of activated vertices.}
\label{fig.propagate}
\end{figure}

\begin{figure}[t]
\centering \mbox{
\subfigure[Gowalla]{\includegraphics[width=0.33\linewidth]{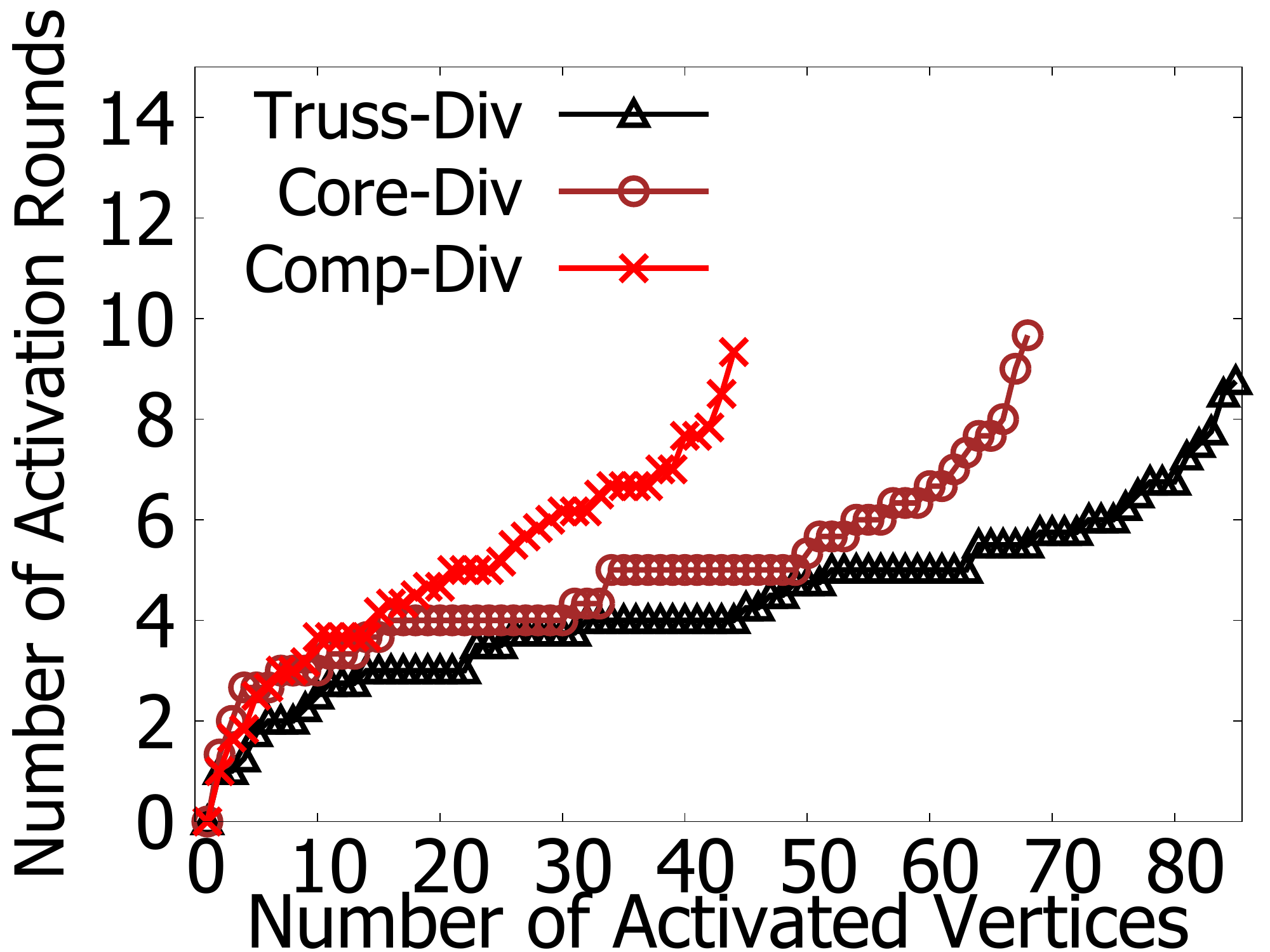}} 
\subfigure[LiveJournal]{\includegraphics[width=0.33\linewidth]{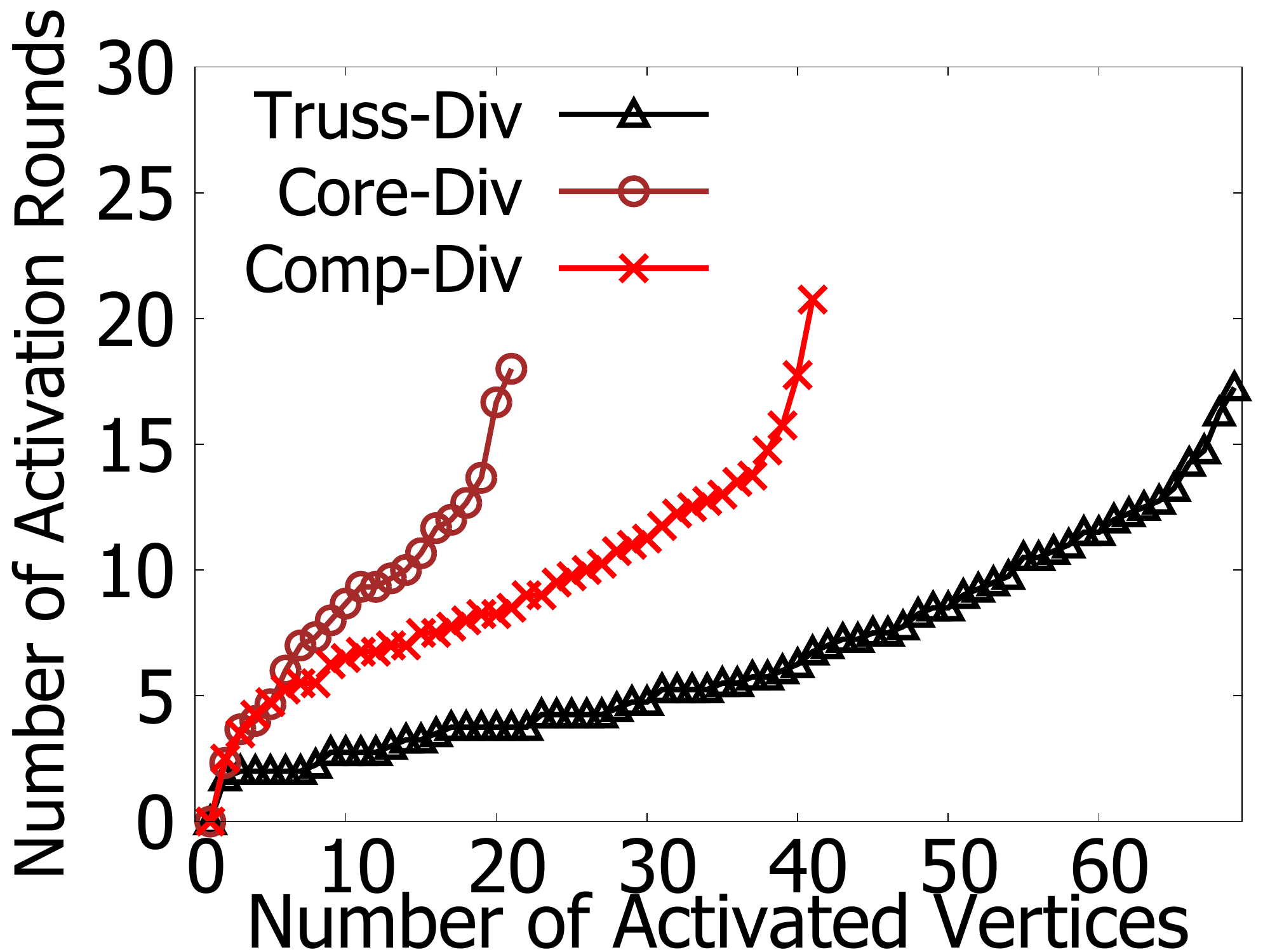}} 
\subfigure[Orkut]{\includegraphics[width=0.33\linewidth]{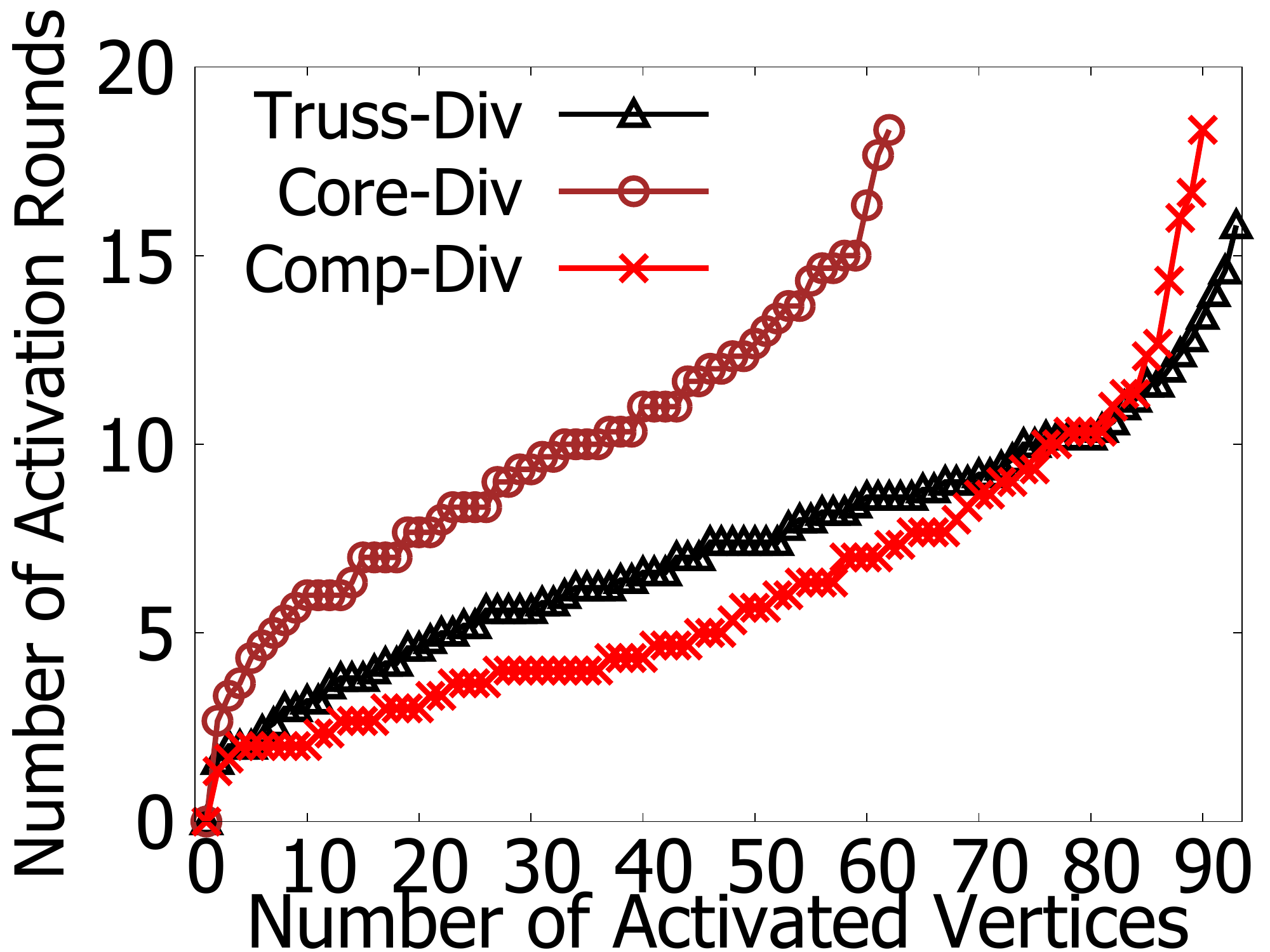}} }
\caption{Latency of activating top-100 results by three models.}
\label{fig.avgacttimes}
\end{figure}

\stitle{Exp-8 (Effectiveness comparison of different models)}: 
We apply all competitor methods \random, \component, \core, and our method \trussdiv to obtain $r$ vertices, by setting the parameter $k=4$ if necessary.  We evaluate how many vertices among those top-$r$ vertices selected by different methods will get activated in the influence propagation. 
The larger the number of activated vertices is, the better is. Figure~\ref{fig.propagate} shows the number of activated vertices by different methods  varied by parameter $r$. We can see that our method has more number of activated vertices than all the other methods, indicating the vertices with larger truss-based structural diversities have a higher probability to get affected by others.

\eat{
\begin{figure}[t]
\vspace*{-0.1cm}
\centering 
\includegraphics[width=0.5\linewidth]{figure/exp/Average_step.pdf}
\caption{Latency of activating top-100 results by three models.}
\label{fig.avgacttimes}
\end{figure}
}

\stitle{Exp-9 (Latency incurred to activate the results of different models)}: 
This experiment evaluates the latency (the number of activation rounds) incurred to activate the top-100 results of \trussdiv, \core and \component. Figure \ref{fig.avgacttimes} reports the average number of activation rounds w.r.t the number of activated vertices on three networks. 
\trussdiv achieves the smallest latency to activate the most number of vertices on Gowalla and Livejournal. \trussdiv is competitive with \component on Orkut, due to the imbalanced structural diversity distribution of top-100 results of \component. The activated speed of \component gets fast firstly and then slows down significantly. It shows that the vertices selected by \trussdiv are more quickly and easily to get social contagion than the \core and \component models. \eat{We also find that the latency has a potential relation with the distribution of structural diversity score inside the top-100 results computed by different models.}


\begin{figure*}[t]

\centering 
\includegraphics[width=0.98\linewidth]{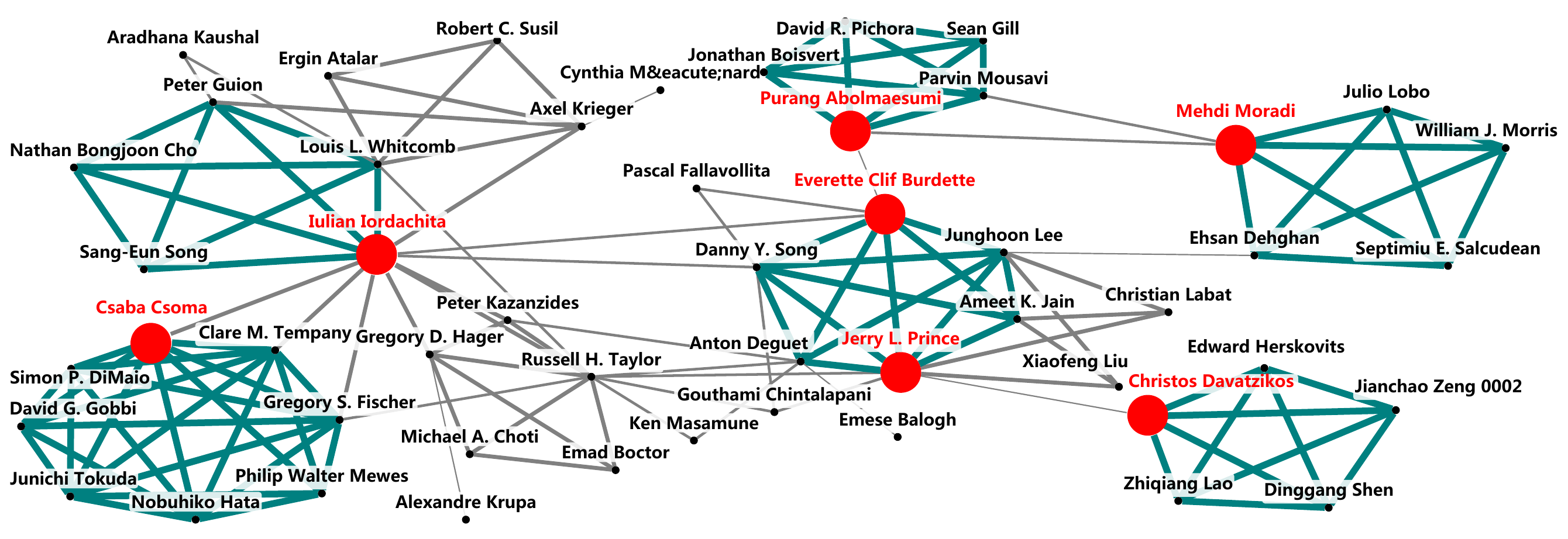}
\caption{A case study of structural diversity search on DBLP. Here, $k=5$ and $r=1$. This is an ego-network 
of ``Gabor Fichtinger''. Each component in green is a maximal connected 5-truss, which represents a distinct social context. 
}
\label{fig.casestudy}
\end{figure*}

\begin{figure}[t]
\centering \mbox{
\subfigure[\component]{\includegraphics[width=0.48\linewidth]{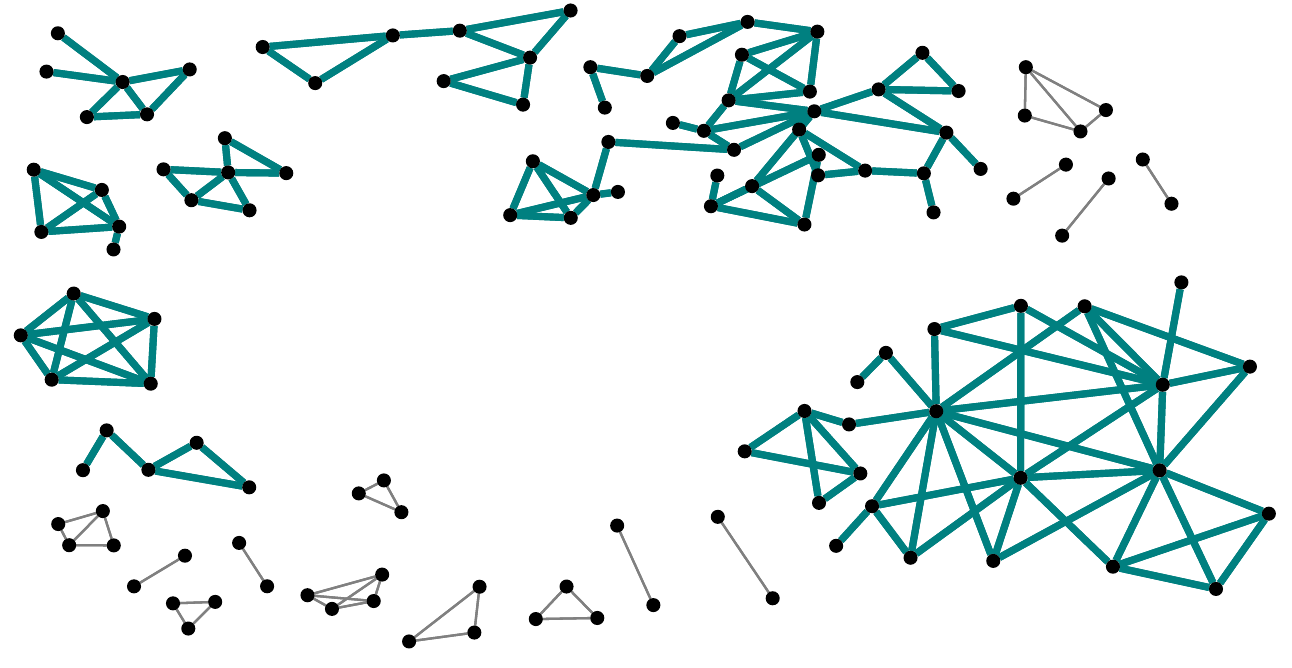}}  \quad
\subfigure[\core]{\includegraphics[width=0.48\linewidth]{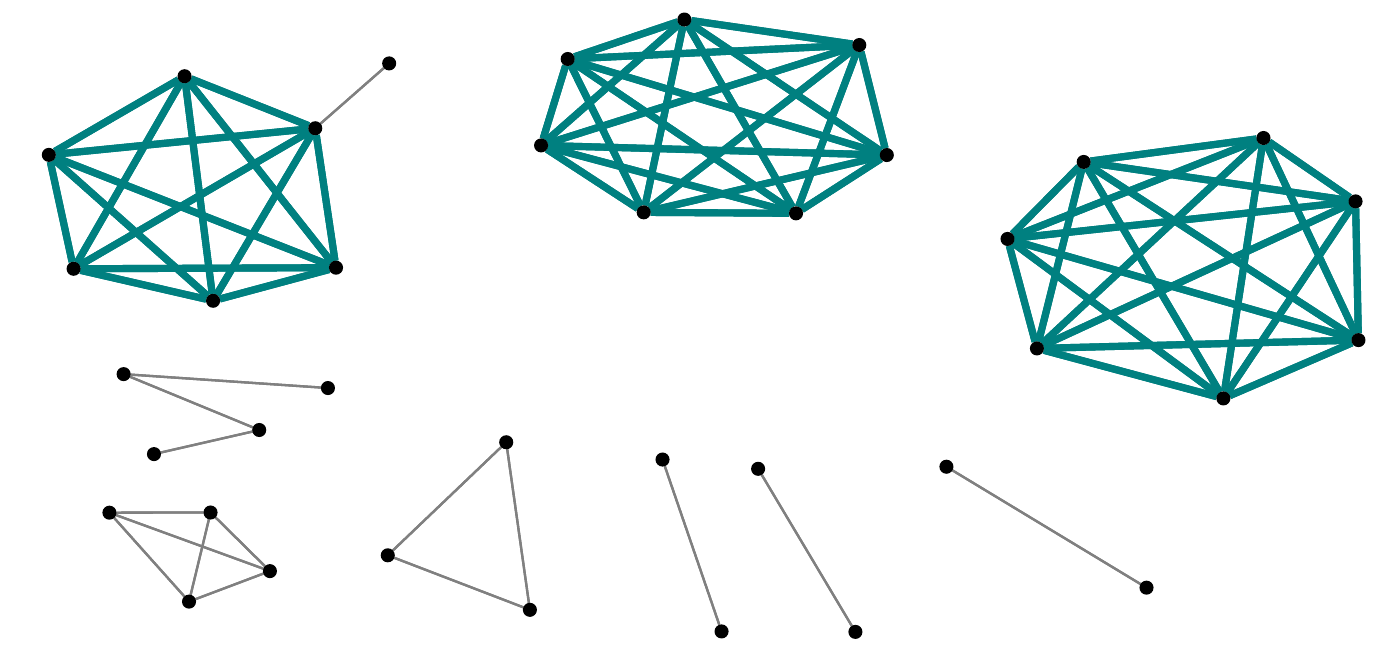}}}
\caption{Ego-networks of top-1 results by \core and \component on DBLP. Here $k=5$ and $r=1$. Social contexts are highlighted in green.}
\label{fig.egooftop}
\end{figure}


\subsection{Case Study on DBLP}

We conduct a case study on a collaboration network from DBLP.\footnote{\scriptsize{\url{https://dblp.uni-trier.de/xml}}} The DBLP network consists of 234,879 vertices and 542,814 edges. An author is represented by a vertex. An edge  between two authors indicates that they have co-authored for at least 3 times. We make a comprehensive comparison of \trussdiv, \component and \core models on the case studies of DBLP network.

\stitle{Exp-10 (Top-1 result by our truss-based model)}: 
We use the query $r=1$ and $k=5$ to test our top-$r$ truss-based structural diversity model. The answer is an author $v^*$ whose name is ``Gabor Fichtinger''. $v^*$ achieves the highest structural diversity score as $\score$$(v^*)$$=6$. Figure~\ref{fig.casestudy} uses a graph visualization tool to depict the \ego $G_{N(v^*)}$ of ``Gabor Fichtinger''. The edges of different trussness are depicted in different patterns. It consists of six maximal connected 5-trusses in green, which represent  six  semantic contents (e.g., 6 research groups working on different topics).  In contrast, we apply \component and \core on this same \ego $G_{N(v^*)}$ and obtain the following \emph{meaningless results}.

\squishlisttight
\item For \component, the whole network cannot be decomposed into multiple social contexts using the component-based model for any $k$-sized component \cite{chang2017scalable}, as the whole network $G_{N(v^*)}$ is a large connected component in Figure~\ref{fig.casestudy}. 
\item For \core, in Figure~\ref{fig.casestudy}, the six components in green are connected together to form a connected 4-core through the edges between the authors highlighted in red: "Csaba Csoma", "Iulian Iordachita", "Everette Clif Burdette", "Purang Abolmaesumi", "Mehdi Moradi", "Jerry L Prince",  and "Christos Davatzikos". 
\end{list} 

Hence, it is also difficult to apply the \component and \core models for effective structural diversity analysis on this complex \ego $G_{N(v^*)}$. This further shows the superiority of truss-based structural diversity model on the analysis of large-scale complex \egos. \eat{Moreover, we further compare 
\core and \component models with ours in Appendix E.}

\jinbin{
\stitle{Exp-11 (Top-1 results by \component and \core models)}: 
To further compare with \trussdiv, we use \component and \core methods to perform their diversity search under the same parameter setting ($k=5$ and $r=1$) on the DBLP network .  Figure~\ref{fig.egooftop} depicts the \egos of top-1 result of \component and \core respectively with eight and three identified social contexts. \component treats one component whose size is at least 5 as a social context. \core treats one maximal 5-core as a social context. Each identified social context is highlighted in green in Figure~\ref{fig.egooftop}. However, these social contexts are completely isolated in Figures~\ref{fig.egooftop}(a) and~\ref{fig.egooftop}(b), which are different from the connected social contexts by \trussdiv in Figure~\ref{fig.casestudy}. It further confirms that component-based and core-based models can find simple structure of isolated social contexts, but have limited decomposability to discover social contexts on complex networks.

\begin{table}[t]
\begin{center}\vspace*{-0.1cm}
\scriptsize
\caption[]{\textbf{Ego-network statistics of top-1 results on DBLP.}}\label{tab:topinfo}
\begin{tabular}{|c|c|c|c|c|c|c|}

\hline  \multirow{2}{*}{Methods} & Author & \multirow{2}{*}{$|V|$} & \multirow{2}{*}{$|E|$} & \multirow{2}{*}{Density} & \multirow{2}{*}{$|\context(v)|$} & Activated \\
 & Name (ego) & & &  &  & Probability \\
\hline \hline
\component & Ming Li & 130 & 344 & 2.64 & 8 &  0.44 \\ \hline
\core & Rui Li & 38 & 148 & 3.89  & 3 & 0.43 \\ \hline
\trussdiv & Gabor Fichtinger & 51 & 264 & \textbf{5.18} & 6 & \textbf{0.47} \\ \hline

\end{tabular}
\end{center}
 \vspace*{-0.4cm}
\end{table}

\stitle{Exp-12 (Quality Evaluation of Social Contexts)}: 
Table~\ref{tab:topinfo} reports the statistics of three \egos of top-1 result by \component, \core, and \trussdiv on DBLP. 
We report the author name of answers,  vertex size,  edge size, density, the number of social contexts (i.e., $|\context(v)|$), and  activated probability. 
We evaluate the activated probability of the center vertex $v^*$ influenced by its neighbors on its \ego. For each top-1 result, we construct a graph $H^*$ formed by the union of \ego $G_{N(v)}$ and $v*$ with incident edges $\{(v^*, u)\in E\}$. We assign the edge probability to 0.05 uniformly, and randomly select 10 influential seeds from $N(v)$. The top-1 result of \trussdiv achieves the highest activated probability of 0.47 on the average of 10,000 runs, which verifies the superiority of our truss-based structural diversity model. Moreover, the ego-network of ``Gabor Fichtinger '' by \trussdiv has the largest density of 5.18.

}

\section{Related Work}\label{sec.related}

Our work is closely related to structural diversity search and $k$-truss mining and indexing.
\eat{Our work is  related to structural diversity search and $k$-truss mining.}


\subsection{Structural Diversity Search}
Social decisions can significantly depend on the social network structure \cite{fowler2010cooperative,dong2017structural}. Ugander et al. \cite{UganderBMK12} conducted extensive studies on the Facebook  to show that the contagion probability of an individual is strongly related to its structural diversity in the \ego. Motivated by \cite{UganderBMK12}, Huang et al. \cite{huang2013top} studies the problem of structural diversity search to find $k$ vertices with the highest structural diversity in graphs. To improve the efficiency of \cite{huang2013top}, Chang et al. \cite{chang2017scalable} proposes a scalable algorithm by enumerating each triangle at most once in constant time. Structural diversity search based on a different $k$-core model is further studied in \cite{XHuang15}. The $k$-truss-based structural diversity studied in this work is also called $k$-brace-based structural diversity \cite{UganderBMK12}. In addition, there also exist numerous studies on \emph{top-$k$ query processing} \cite{IlyasBS08,ZhangCM02,AgrawalGHI09,AngelK11,QinYC12} by considering diversity in the returned ranking results. 
However, the problem of structural diversity search based on $k$-truss model has not been investigated by any study mentioned above. 

\subsection{K-Truss Mining and Indexing}
In the literature, there exist a large number of studies on $k$-truss mining and indexing. As a cohesive subgraph, $k$-truss requires that each edge has at least $(k-2)$ triangles within this subgraph \cite{cohen2008}. Interestingly, several equivalent concepts of $k$-truss termed as different names are independently studied. For example, $k$-truss has been named  as the $k$-dense community \cite{saito2008extracting,gregori2011k}, $k$-mutual-friend subgraph \cite{zhao2012large}, $k$-brace \cite{UganderBMK12}, and triangle $k$-core \cite{YZhang12}. The task of truss decomposition is to find the non-empty $k$-truss for all possible $k$'s in a graph. Wang and Cheng \cite{WangC12} propose a fast in-memory algorithm for truss decomposition. In addition, truss decomposition has also been studied in various computing settings (e.g., external-memory algorithms \cite{WangC12}, MapReduce algorithms \cite{Cohen09,chen2014distributed}, and shared-memory parallel systems \cite{smith2017truss}) and different types of graphs (e.g., uncertain graphs \cite{zou2017truss,huang2016truss,EsfahaniW0T019}, directed graphs \cite{takaguchi2016cycle}, and dynamic graphs \cite{YZhang12,huang2014querying}). 
\jinbin{
\eat{Recently, several 
$k$-truss-based indexes are proposed for the efficient retrieval of communities. For example, \cite{huang2014querying} proposes a TCP-index for $k$-truss community search. \cite{akbas2017truss} proposes a new index Equi-Truss to improve TCP-index.}
Recently, several community models are built on the 
$k$-truss \cite{huang2014querying,akbas2017truss,zheng2017finding,huang2017attribute}. 
Meanwhile, a number of $k$-truss-based indexes (e.g., TCP-index \cite{huang2014querying} and Equi-Truss \cite{akbas2017truss}) are proposed for another problem of community search, which supports the efficient retrieval of communities.
A detailed comparison of truss-based indexes is made below.

\stitle{Truss-based Index Comparison}. 
We introduce and compare three different indexes based on $k$-truss, including our \TSDindex, TCP-index \cite{huang2014querying}, and Equi-Truss \cite{akbas2017truss}. Among them, 
TCP-index and Equi-Truss are developed for $k$-truss community search \cite{huang2014querying}. A $k$-truss community is a maximal connected $k$-truss such that all edges are triangle connected via a series of adjacent triangles within this community. Huang et al. \cite{huang2014querying} proposes a tree-shaped structure of TCP-index for efficiently finding $k$-truss communities. To speed up the discovery of $k$-truss communities, Akbas and Zhao  \cite{akbas2017truss} propose a novel indexing technique of Equi-Truss by compressing TCP-index into a more compact structure. 

Specifically, the major differences of our \TSDindex in contrast to state-of-the-art TCP-index \cite{huang2014querying}  and Equi-Truss \cite{akbas2017truss} are listed as follows. First, TCP-index and Equi-Truss take the global trussness and triangle connectivity on the whole graph into consideration, while TSD-index only focuses on the local neighborhood induced subgraph without considering the triangle constraint. Second, the index construction of \TSDindex costs much more expensive than those of TCP-index and Equi-Truss, in terms of their time complexities \cite{huang2014querying,akbas2017truss}. Last but not least, \TSDindex and TCP-index have tree-shaped structures with different edge weights, and more importantly the meaning of edge weights are substantially different.  For example, Figure~\ref{fig.tcp_tsd}(a) shows the graph $G$. Consider a vertex $q_1$ in $G$, Figures~\ref{fig.tcp_tsd}(b) and~\ref{fig.tcp_tsd}(c) respectively show the corresponding TCP-index of $q_1$ and \TSD-index of $q_1$. All edges have different weights in two indexes in Figures~\ref{fig.tcp_tsd}(b) and \ref{fig.tcp_tsd}(c). Consider an edge $(q_2, q_3)$ of the TCP-index in Figure~\ref{fig.tcp_tsd}(b),  indicates that $(q_2, q_3)$ will be involved in a 4-truss community as the global graph $G$. However, the edge $(q_2, q_3)$ of the \TSDindex in Figure~\ref{fig.tcp_tsd}(c),  indicates that $(q_2, q_3)$ will be involved in a maximal connected 2-truss in the \ego  $G_{N(q_1)}$.

\begin{figure}[t]
\centering \mbox{
\subfigure[Graph $G$]{\includegraphics[width=0.3\linewidth]{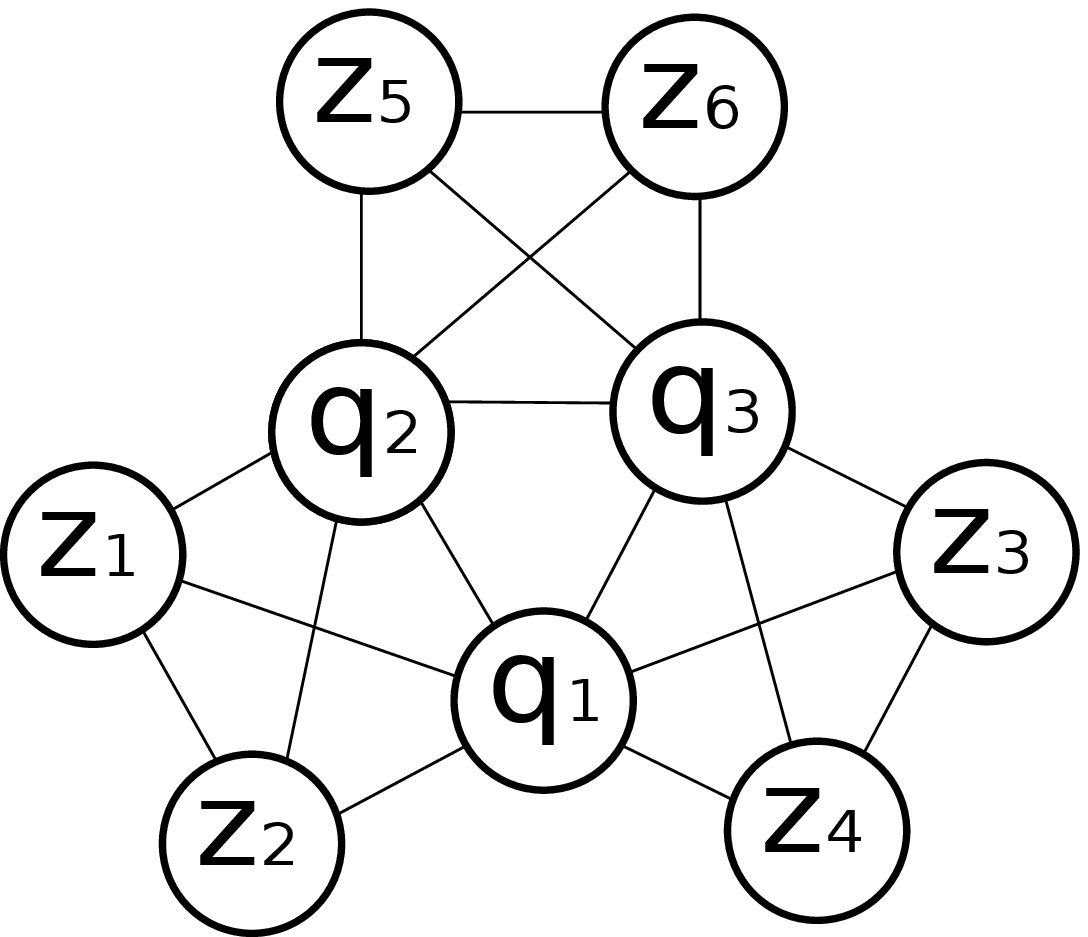}} \quad
\subfigure[TCP-Index of $q_1$]{\includegraphics[width=0.3\linewidth]{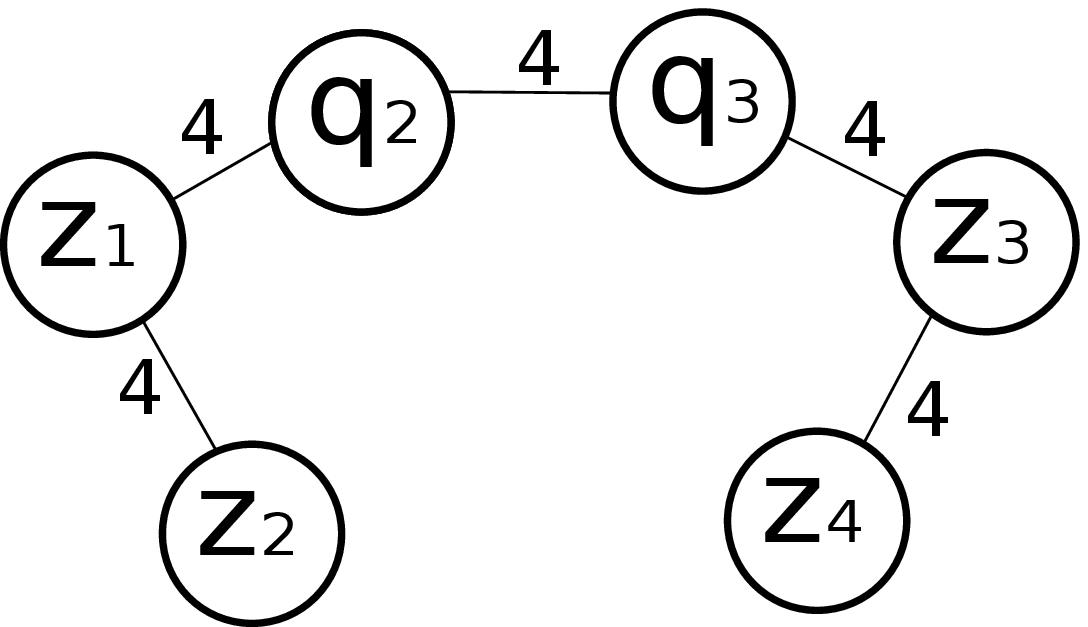}} \quad
\subfigure[TSD-Index of $q_1$]{\includegraphics[width=0.3\linewidth]{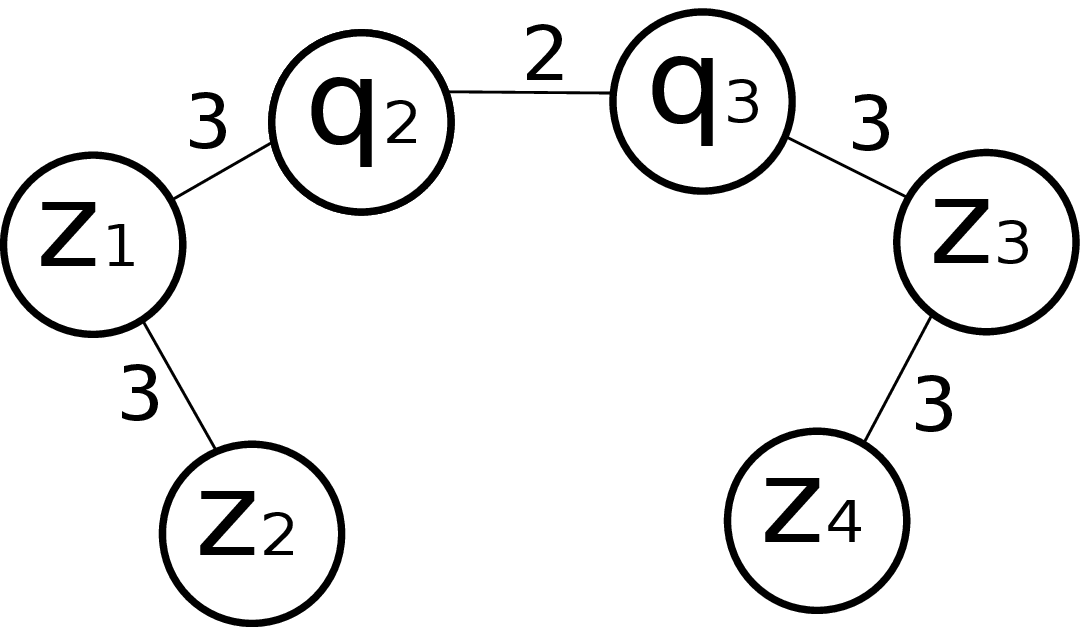}}
 }
\caption{Comparison: TSD-Index and TCP-Index}
\label{fig.tcp_tsd}
\end{figure}

 In contrast to the above studies, $k$-truss-based structural diversity search is firstly studied in this paper. Leveraging the micro-network analysis of \egos, we propose a novel tree-shaped structure of \TSDindex and efficient  algorithms to address our problem. 

 }


\section{Conclusions}\label{sec.con}
In this paper, we investigate the problem of truss-based structural diversity search over graphs. We propose a truss-based structural diversity model to discover social contexts, which has a strong decomposition to break up weak-tied social groups  in large-scale complex networks. We propose several efficient algorithms to solve the top-$r$ truss based structural diversity search problem. We first develop efficient techniques of graph sparsification and an upper bound for pruning. We also propose a well-designed and elegant \TSDindex for keeping the information of structural diversity which solves the problem in time linear to graph size. Moreover, we develop a new \GCT algorithm based on \ADVindex. 
Experiments also show the effectiveness and efficiency of our proposed truss-based structural diversity model and algorithms, against state-of-the-art component-based and core-based methods. 

\bibliographystyle{abbrv}
{\small
\bibliography{strucdiv}
}
\begin{IEEEbiography}[{\includegraphics[width=1in,height=1.25in,clip,keepaspectratio]{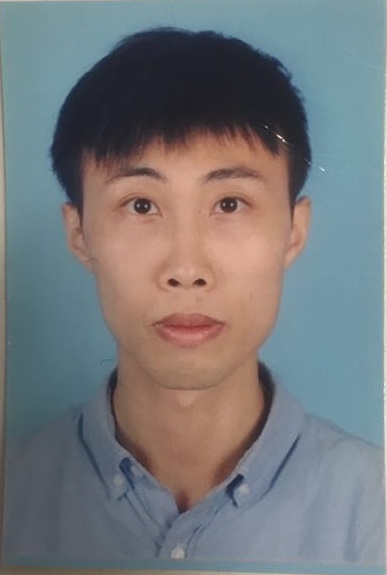}}]{Jinbin Huang}
 Jinbin Huang received his bachelor degree in Computer Science in South China University of Technology (SCUT). He is now a PhD student in Hong Kong Baptist University (HKBU). 
\end{IEEEbiography}
\vspace{-1cm}
\begin{IEEEbiography}[{\includegraphics[width=1in,height=1.25in,clip,keepaspectratio]{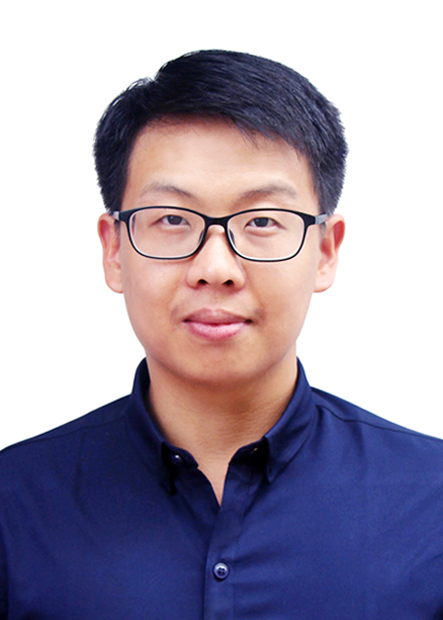}}]{Xin Huang}
Xin Huang received the PhD degree from the Chinese University of Hong Kong (CUHK) in 2014. He is currently an Assistant Professor at Hong Kong Baptist University. His research interests mainly focus on graph data management and mining.
\end{IEEEbiography}
\vspace{-1cm}
\begin{IEEEbiography}[{\includegraphics[width=1in,height=1.25in,clip,keepaspectratio]{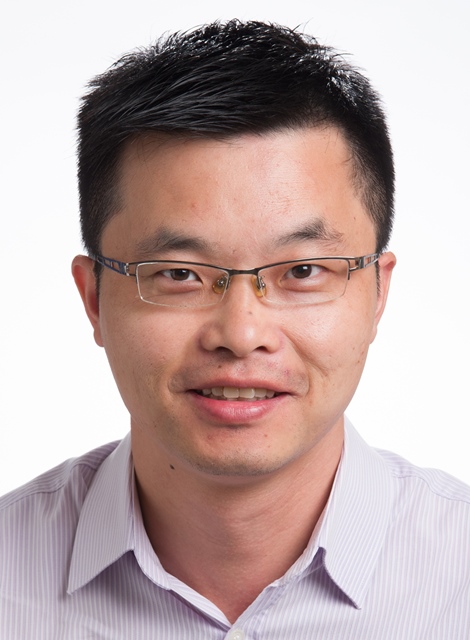}}]{Jianliang Xu}
Jianliang Xu received the Ph.D. degree from The Hong Kong University of Science and Technology. He is currently a Professor with the Department of Computer Science, Hong Kong Baptist University. 
He is an associate editor of the IEEE Transactions on Knowledge and Data
Engineering and the Proceedings of the VLDB Endowment 2018.
\end{IEEEbiography}


\end{document}